\documentclass[11pt,a4paper]{article}
\pdfoutput=1
\usepackage{indentfirst}
\usepackage{hyperref}
\usepackage{fullpage}
\usepackage{amssymb}
\usepackage{mathtools}
\usepackage{amsthm}
\usepackage{eucal}
\usepackage{dsfont}
\usepackage[T1]{fontenc}
\usepackage{lmodern}
\usepackage[stretch=10,shrink=10]{microtype}
\usepackage[capitalise]{cleveref}
\usepackage{color,soul}
\usepackage{CJK}
\usepackage{commutative-diagrams}
\usepackage[framemethod=tikz]{mdframed}
\usepackage{longtable}
\usepackage{amsopn}

\allowdisplaybreaks[1]

\def\A{\mathcal{A}}
\def\B{\mathcal{B}}
\def\C{\mathcal{C}}
\def\D{\mathcal{D}}

\def\F{\mathcal{F}}
\def\G{\mathcal{G}}
\def\H{\mathcal{H}}

\def\L{\mathcal{L}}
\def\M{\mathcal{M}}

\def\P{\mathcal{P}}
\def\Q{\mathcal{Q}}
\def\R{\mathcal{R}}

\def\T{\mathcal{T}}

\def\X{\mathcal{X}}
\def\Y{\mathcal{Y}}

\def\g{\textbf{g}}

\theoremstyle{plain}
\newtheorem{theorem}{Theorem}[section]
\newtheorem{lemma}[theorem]{Lemma}
\newtheorem{proposition}[theorem]{Proposition}

\theoremstyle{definition}
\newtheorem{definition}[theorem]{Definition}
\newtheorem{claim}[theorem]{Claim}

\newtheorem{remark}[theorem]{Remark}
\newtheorem{fact}[theorem]{Fact}

\DeclareMathAlphabet{\mathpzc}{OT1}{pzc}{m}{it}
\newcommand {\minusspace} {\: \! \!}

\newcommand {\fn} [2] {\ensuremath{ #1 \minusspace \br{ #2 } }}
\newcommand {\Fn} [2] {\ensuremath{ #1 \minusspace \Br{ #2 } }}
\newcommand{\reals}{{\mathbb R}}
\newcommand{\complex}{{\mathbb C}}

\newcommand {\set} [1] {\ensuremath{ \left\lbrace #1 \right\rbrace }}

\newcommand {\br} [1] {\ensuremath{ \left( #1 \right) }}
\newcommand {\Br} [1] {\ensuremath{ \left[ #1 \right] }}

\newcommand {\norm} [1] {\ensuremath{ \left\| #1 \right\| }}
\newcommand {\normsub} [2] {\ensuremath{ \norm{#1}_{#2} }}
\newcommand {\onenorm} [1] {\normsub{#1}{1}}
\newcommand {\twonorm} [1] {\normsub{#1}{2}}
\newcommand {\abs} [1] {\ensuremath{ \left| #1 \right| }}

\newcommand {\bra} [1] {\ensuremath{ \left\langle #1 \right| }}
\newcommand {\ket} [1] {\ensuremath{ \left| #1 \right\rangle }}
\newcommand {\ketbratwo} [2] {\ensuremath{ \left| #1 \middle\rangle \middle\langle #2 \right| }}
\newcommand {\ketbra} [1] {\ketbratwo{#1}{#1}}

\newcommand {\innerproduct} [2] {\ensuremath{\left \langle #1 , #2 \right \rangle}}
\newcommand{\nnorm}[1]{{\left\vert\kern-0.25ex\left\vert\kern-0.25ex\left\vert #1
		\right\vert\kern-0.25ex\right\vert\kern-0.25ex\right\vert}}

\newcommand {\defeq} {\ensuremath{ \stackrel{\mathrm{def}}{=} }}

\newcommand {\prob} [1] {\Fn{\Pr}{#1}}

\DeclareMathOperator*{\bigE}{\mathbb{E}}
\newcommand {\expec} [2] {\Fn{\bigE_{\substack{#1}}}{#2}}
\newcommand {\var} [1] {\Fn{\mathrm{Var}}{#1}}
\newcommand {\influence} {\ensuremath{ \mathrm{Inf} }}
\newcommand{\supp}[1]{\mathrm{supp}\br{#1}}

\newcommand {\Tr} {\ensuremath{ \mathrm{Tr} }}

\newcommand {\id} {\ensuremath{\mathds{1}}}

\newcommand{\conjugate}[1]{\overline{#1}}

\newcommand{\anticommutator}[2]{\set{#1,#2}}

\newcommand{\wt}[1]{\mathrm{wt}\br{#1}}

\newcommand{\tabincell}[1]{\begin{tabular}{c}#1\end{tabular}}
\newcommand{\subtabtwo}[5]{\cline{2-3}\tabincell{#1}&\multicolumn{1}{|c|}{#2}&\multicolumn{1}{c|}{#2}&\\\cline{2-3}\rule{0pt}{\properheight}&#3&#4&\tabincell{#5}\\}
\newcommand{\subtabone}[5]{\cline{2-3}\tabincell{#1}&\multicolumn{2}{|c|}{#2}&\\\cline{2-3}\rule{0pt}{\properheight}&#3&#4&\tabincell{#5}\\}
\newcommand{\properheight}{6mm}

\newcommand {\email} [1] {\href{mailto:#1}{\texttt{#1}}}

\newcommand {\Penghui} {Penghui Yao}
\newcommand{\SKL}{State Key Laboratory for Novel Software Technology}
\newcommand{\NJU}{ Nanjing University}

\newcommand{\pinv}[1]{\ensuremath{#1^+}}
\newcommand{\pos}[1]{\ensuremath{#1^{\mathpzc{pos}}}}
\newcommand{\sqrtpinv}[1]{\ensuremath{\br{#1^+}^{\frac{1}{2}}}}
\renewcommand{\vec}[1]{\stackrel{\longrightarrow}{#1}}

\begin{document}
\title{Nonlocal games with noisy maximally entangled states are decidable}
\author{Minglong Qin\thanks{\SKL, \NJU
  (\email{mf1833054@smail.nju.edu.cn} )}
\and \Penghui\thanks{\SKL, \NJU(\email{pyao@nju.edu.cn})}}
\date{}

%
%
%
%
%
		\maketitle

		\thispagestyle{empty}
		\begin{abstract}
This paper considers a special class of nonlocal games $(G,\psi)$, where $G$ is a two-player one-round game, and $\psi$ is a bipartite state independent of $G$. In the game $\br{G,\psi}$, the players are allowed to share arbitrarily many copies of $\psi$. The value of the game $(G,\psi)$, denoted by $\omega^*(G,\psi)$, is the supremum of the winning probability that the players can achieve with arbitrarily many copies of preshared states $\psi$. For a noisy maximally entangled state $\psi$, a two-player one-round game $G$ and an arbitrarily small precision $\epsilon>0$, this paper proves an upper bound on the number of copies of $\psi$ for the players to win the game with a probability $\epsilon$ close to $\omega^*(G,\psi)$. A noisy maximally entangled state is a two-qudit state with both marginals being completely mixed states and the maximal correlation being less than $1$. In particular, it includes $(1-\epsilon)\ketbra{\Psi_m}+\epsilon\frac{\id_m}{m}\otimes\frac{\id_m}{m}$ for $\epsilon>0$, where $\ket{\Psi_m}=\frac{1}{\sqrt{m}}\sum_{i=0}^{m-1}\ket{m,m}$ is an $m$-dimensional maximally entangled state. Hence, it is feasible to approximately compute  $\omega^*\br{G,\psi}$ to an arbitrarily precision. Recently, a breakthrough result by Ji, Natarajan, Vidick, Wright and Yuen showed that it is undecidable to approximate the values of nonlocal games to a constant precision, when the players preshare arbitrarily many copies of perfect maximally entangled states, which implies that $\mathrm{MIP}^*=\mathrm{RE}$. In contrast, our result implies the hardness of approximating nonlocal games collapses when the preshared maximally entangled states are noisy.

The paper develops a theory of Fourier analysis on matrix spaces by extending a number of techniques in Boolean analysis and Hermitian analysis to matrix spaces. We establish a series of new techniques, such as a quantum invariance principle and a hypercontractive inequality for random operators, which we believe have further applications.


\end{abstract}
		
%
%
%
	
	\setcounter{page}{1}

	\section{Introduction}
 {\em Interactive proof systems} are a fundamental concept to the theory of computing. It was first proposed by Babai~\cite{Babai:1985:TGT:22145.22192} and Goldwasser, Micali and Rackoff~\cite{Goldwasser:1985:KCI:22145.22178}, and later extended to the multiprover setting in~\cite{Ben-Or:1988:MIP:62212.62223}. The study of interactive proof systems is at the heart of the theory of computing, including the elegant characterizations $\mathrm{IP}=\mathrm{PSPACE}$~\cite{Shamir:1992:IP:146585.146609,Shen:1992:ISS:146585.146613} for single-prover interactive proof systems and  $\mathrm{MIP}=\mathrm{NEXP}$~\cite{Babai1991} for multiprover interactive proof systems. The latter result further led to the celebrated PCP theorem~\cite{Arora:1998:PVH:278298.278306, Arora:1998:PCP:273865.273901}.

The study on the power of interactive proof systems in the context of quantum computing also has a rich history. The model of single-prover quantum interactive proof systems was first studied by Watrous~\cite{Watrous:2003:PCQ:763677.763679}, followed by a series of works~\cite{Kitaev:2000:PAE:335305.335387,Marriott:2005:QAG:1391802.1391807,Gutoski:2007:TGT:1250790.1250873,Jain:2009:TQI:1747597.1748069}, which finally led to the seminal result $\mathrm{QIP}=\mathrm{PSPACE}$~\cite{Jain:2011:QP:2049697.2049704}. Quantum multiprover interactive proof systems are more complicated. A key assumption on the classical multiprover interactive proof systems is that the provers are not allowed to communicate, which means that their only resource in common is the shared randomness. In quantum multiprover interactive proof systems, this assumption is relaxed and the provers are allowed to share {\em entanglement}, and the corresponding complexity class is $\mathrm{MIP}^*$~\cite{Cleve:2004:CLN:1009378.1009560}. Surprisingly, understanding the power of $\mathrm{MIP}^*$ turns out to be extremely difficult. A trivial lower bound on $\mathrm{MIP}^*$ is $\mathrm{IP}$, or equivalently $\mathrm{PSPACE}$, which can be easily seen by ignoring all but one provers. Extending the techniques in~\cite{Babai1991} to the quantum setting, Ito and Vidick proved the containment of $\mathrm{NEXP}$ in $\mathrm{MIP}^*$~\cite{Ito:2012:MIP:2417500.2417883}. This lower bound was improved by a series of works in various settings~\cite{Ji:2016:CVQ:2897518.2897634,slofra:2016,Ji:2017:CQM:3055399.3055441,NWright:2019,slofstra_2019,FJVYuen:2019}, which lead to a very recent breakthrough result~\cite{JNVWY'20,JNVWYuen'20}, in which Ji, Natarajan, Vidick, Wright and Yuen constructed a quantum multiprover interactive proof system for the Halting problem, and thus proved that $\mathrm{MIP}^*=\mathrm{RE}$, where $\mathrm{RE}$  is the set of recursively enumerable languages.

This paper concerns {\em two-player one-round games}, a core model in computational complexity, which is closely related to multiprover interactive proof systems. A two-player one-round game $G$ is run by three parties, a referee and two non-communicating players. We define $G=\br{\X,\Y,\A,\B,\mu,V}$, where $\X,\Y,\A,\B$ are finite sets, $\mu$ is a distribution over $\X\times\Y$ and $V:\X\times\Y\times\A\times\B\rightarrow\set{0,1}$ is a predicate. Note that all $\X,\Y,\mu,V$ are public. The referee samples a pair of questions $\br{x,y}$ according to $\mu$, and sends $x$ and $y$ to the two players separately. The two players have to provide an answer to the referee by choosing from $\A$ and $\B$, respectively, denoted by $\br{a,b}$. The referee accepts the answers he receives if and only if $V\br{x,y,a,b}=1$. The only restriction on the players' strategies is that they are not allowed to exchange any information once the game has started. In the classical setting, the value of the game $\omega\br{G}$, the highest probability that the referee accepts the game, is
\[\omega\br{G}=\max_{h_A:\X\rightarrow\A\atop h_B:\Y\rightarrow\B}\sum_{xy}\mu\br{x,y}V(x,y,h_A(x),h_B(y)).\]
By the PCP theorem~\cite{Arora:1998:PVH:278298.278306, Arora:1998:PCP:273865.273901} it is $\mathrm{NP}$-hard to approximate $\omega(G)$ within a multiplicative constant. The {\em entangled games}, which are the same as the classical games except that the players are allowed to share arbitrary entangled states before they receive the questions, were first introduced by Cleve, H{\o}yer, Toner and Watrous~\cite{Cleve:2004:CLN:1009378.1009560}. They also defined the {\em entangled value} of a game, denoted by $\omega^*(G)$, to be the supremum of the probability that the referee accepts in a game when the players share entanglement,
\begin{equation}\label{eqn:omegastarG}
\omega^*\br{G}=\lim_{n\rightarrow\infty}\max_{\psi_{AB}\in\D_n\atop \set{P^x_a}_{x,a},\set{Q^y_b}_{y,b}}\sum_{xy}\mu\br{x,y}\sum_{ab}V(x,y,a,b)\Tr\br{\br{P^x_a\otimes Q^y_b}\psi_{AB}},
\end{equation}
where $\D_n$ is the set of $n$-dimensional density operators, $\set{P^x_a}_a$ and $\set{Q^y_b}_b$ are POVMs for any $x
\in\X,y\in\Y$, respectively. Namely, $\sum_{a\in\A}P^x_a=\id$, $\sum_{b\in\B}Q^y_b=\id$, $P^x_a\geq 0$ and  $Q^y_b\geq 0$.

The quantities $\omega^*\br{G}$ of nonlocal games are important in both physics and computer science. Cleve et al.~\cite{Cleve:2004:CLN:1009378.1009560} discovered the fact that $\omega^*(G)>\omega(G)$ for the so-called {\em CHSH games}~\cite{PhysRevLett.23.880} gives a re-interpretation of the {\em Bell's inequalities}~\cite{PhysicsPhysiqueFizika.1.195}, which plays a central role in all aspects of quantum mechanics. It is nowadays known that there exist games for which $\omega^*(G)=1$ while $\omega(G)$ is arbitrarily small~\cite{doi:10.1137/S0097539795280895,aravind:2002}. A large body of works have also been devoted to understanding the hardness of computing $\omega^*\br{G}$. It was shown in~\cite{Kempe:2008:EGH:1470582.1470594,Ito:2009:OTO:1602931.1603187} that approximating $\omega^*(G)$ to an inverse-polynomial accuracy is $\mathrm{NP}$-hard. Ji proved that it is $\mathrm{QMA}$-hard to approximate $\omega^*(G)$ for multiplayer games to an inverse-exponential accuracy~\cite{Ji:2016:CVQ:2897518.2897634}. Later it was shown by the same author that approximating $\omega^*(G)$ for multiplayer games to an inverse-polynomial accuracy is $\mathrm{MIP}^*$-complete~\cite{Ji:2017:CQM:3055399.3055441}. Slofstra in his seminal results~\cite{slofra:2016,slofstra_2019} proved that determining whether $\omega^*(G)=1$ for general games $G$ is undecidable. Moreover, a positive answer to the so-called Tsirelson's problem (see e.g.~\cite{doi:10.1142/S0129055X12500122}) implies the existence of an algorithm approximating $\omega^*(G)$ for general games. It is known that  Tsirelson's problem is related to Connes' Embedding Conjecture~\cite{10.2307/1971057}, which was a longstanding open problem in functional analysis~\cite{doi:10.1063/1.3514538,Ozawa2013}. In~\cite{JNVWY'20,JNVWYuen'20}, Ji et al. completely resolved this problem by proving that even approximating $\omega^*(G)$ for two-player nonlocal games to a constant accuracy is undecidable.

On the other hand, for a few known classes of games, $\omega^*(G)$ is computable, and sometimes computing $\omega^*(G)$ is even easier than computing $\omega(G)$. Cleve et al. in~\cite{Cleve:2004:CLN:1009378.1009560} gave a polynomial-time algorithm to exactly compute $\omega^*(G)$ for XOR games $G$ building on the work of Tsirelson~\cite{Cirel'son1980}. Kempe, Regev and Toner later presented a polynomial-time algorithm for $\omega^*(G)$ for unique games with a factor $6$ approximation to $1-\omega^*(G)$~\cite{doi:10.1137/090772885}. However, it is $\mathrm{NP}$-hard to approximate $\omega(G)$ for XOR games within a factor of $11/12+\epsilon$ and also $\mathrm{NP}$-hard to approximate $\omega(G)$ for $k$-XOR games within a factor of $1/2+\epsilon$ for $k \geq 3$~\cite{Hastad:2001:OIR:502090.502098}. Classical unique games are conjectured to be NP-hard~\cite{Khot:2002:PUG:509907.510017} as well.
 Both of the algorithms in~\cite{Cleve:2004:CLN:1009378.1009560,doi:10.1137/090772885} are based on convex optimization. In particular, a hierarchy of semidefinite programs was proposed in~\cite{Navascu_s_2008}. The optimal values converge to $\omega^*_{\mathrm{co}}(G)$, which are the values of games when the players employ commutative strategies, and thus $\omega_{\mathrm{co}}\br{G}$ is an upper bound on $\omega^*(G)$. However, the speed of convergence is unknown. There are some other classes of nonlocal games whose entangled values are known to be computable~\cite{benewatts_et_al:LIPIcs:2018:10103}. Readers may refer to the survey~\cite{8186772} for more details.

\subsection*{Main results}

In this paper, we prove that for any nonlocal game, if the players share arbitrarily many copies of noisy maximally entangled states (MES) of a fixed dimension, then it is feasible to approximate the supremum of the winning probability to an arbitrary precision.

\begin{theorem}[Main result, informal]
	Given a nonlocal game $G$, for any integer $m\geq2$, if the players share arbitrarily many copies of $m$-dimensional noisy MES $\Psi_{\epsilon}=\br{1-\epsilon}\ketbra{\Psi_m}+\epsilon\frac{\id_m}{m}\otimes\frac{\id_m}{m}$ for any $\epsilon>0$, where $\ket{\Psi_m}=\frac{1}{\sqrt{m}}\sum_{i=0}^{m-1}\ket{i,i}$ is an $m$-dimensional MES, there exists an explicitly computable $D=D(\epsilon,\delta,m,G)$ such that it suffices for the players to share $D$ copies of $\psi$ to achieve the winning probability at least $\omega^*\br{G,\Psi_{\epsilon}}-\delta$, where $\omega^*\br{G,\Psi_{\epsilon}}$ represents the value of the game, which is the supremum of the winning probability that the players can achieve with arbitrarily many copies of preshared states $\Psi_{\epsilon}$. Thus, it is feasible to approximate the value of the game $\omega^*(G,\Psi_{\epsilon})$ to arbitrary precision.
\end{theorem}

A natural and naive approach is to prove that the players are able to produce arbitrarily many copies of noisy MESs with bounded copies of perfect MESs and preshared classical randomness. The main result would be trivial if this were possible. However, this is not the case since the entanglement between the two players does not increase via local operations.

Ji et al. in~\cite{JNVWY'20} showed that if the players share arbitrarily many copies of perfect MES, then approximating the values of the games is as difficult as the Halting problem. Thus our result implies that the hardness of approximating nonlocal games is not robust against the noise of the preshared states.

The techniques developed in this paper are completely different from all previous approaches~\cite{Cleve:2004:CLN:1009378.1009560,doi:10.1137/090772885,Navascu_s_2008,Beigi:2010:LBV:2011451.2011453,doi:10.1142/S0129055X12500122,slofra:2016,benewatts_et_al:LIPIcs:2018:10103,slofstra_2019,JNVWY'20} . We generalize the framework of Boolean analysis, a well-studied and fruitful topic in theoretical computer science~\cite{Odonnell08}, to matrix spaces and reduce the problem to {\em local state transformations}~\cite{Delgosha2014}, a quantum analog of {\em non-interactive simulations of joint distributions}. This approach provides a new perspective with novel tools to study entangled nonlocal games. We develop a series of results about Fourier analysis on matrix spaces, which might be of independent of interest and have applications in quantum property testing, quantum machine learning, etc.

Non-interactive simulations of joint distributions are a fundamental problem in information theory and communication complexity. Consider two non-communicating players Alice and Bob, and suppose they are provided a sequence of independent samples $\br{x_1,y_1},\br{x_2,y_2},\ldots$ from a joint distribution $\mu$ on $\X\times\Y$, where Alice observes $x_1,x_2,\ldots$ and Bob observes $y_1,y_2,\ldots$. Without communicating with each other, for what joint distribution $\nu$ can Alice and Bob sample? The research on this problem dates back to the classic works by G\'acs and K\"orner~\cite{Gacs:1973}, Wyner~\cite{Wyner:1975:CIT:2263311.2268812} and Witsenhausen~\cite{doi:10.1137/0128010}, followed by fruitful subsequent works (see, for example,~\cite{7452414} and the references therein). Recently, Ghazi, Kamath and Sudan in~\cite{7782969} studied the decidability of non-interactive simulations of joint distributions by introducing a framework built on the theory of Boolean analysis and Hermitian analysis~\cite{MosselOdonnell:2010,Mossel:2010,Odonnell08}. Using this framework, the decidability is resolved in subsequent works~\cite{doi:10.1137/1.9781611975031.174,Ghazi:2018:DRP:3235586.3235614}.

In the quantum universe, it is natural to consider the quantum analog of non-interactive simulations of joint distributions, which is also referred to as local state transformations. Suppose the two non-communicating players Alice and Bob are provided with arbitrarily many copies of bipartite quantum states $\psi_{AB}$. Without communicating with each other, what bipartite quantum state $\phi_{AB}$ can Alice and Bob jointly create? Delgosha and Beigi first studied this problem and gave a necessary condition for the exact local state transformation of $\psi_{AB}$ to $\phi_{AB}$~\cite{Delgosha2014}. Other than this result, not much about this problem is known. Our paper essentially resolves the decidability of local state transformations, when $\psi_{AB}$ is a noisy MES and the target state $\phi_{AB}$ is a classical bipartite distribution. The proofs adopt the framework laid down in~\cite{7782969} and its subsequent works~\cite{doi:10.1137/1.9781611975031.174,Ghazi:2018:DRP:3235586.3235614}. It heavily uses Boolean analysis and Hermitian analysis on Gaussian spaces, which have been intensively studied and have rich applications in theoretical computer science~\cite{Odonnell08}. Some key components in Boolean analysis and Hermitian analysis, such as {\em hypercontractive inequalities}, have also been extended to quantum settings in various aspects, which have led to several interesting applications~\cite{4690981,doi:10.1063/1.4769269,Temme_2014,King2014,Delgosha2014,doi:10.1063/1.4933219}.
However, much less is known about Fourier analysis on matrix spaces or quantum operations compared to Boolean analysis or Hermitian analysis on Gaussian spaces. We systematically develop Fourier analysis on {\em random matrix spaces}, which are hybrids of matrix spaces and Gaussian spaces. In particular, this paper proves a {\em quantum invariance principle} and a hypercontractive inequality for {\em random operators}, successfully generalizing the framework established in~\cite{7782969,doi:10.1137/1.9781611975031.174,Ghazi:2018:DRP:3235586.3235614} to the quantum setting. Invariance principles are a core result in the analysis of Boolean functions, which have a number of applications in various areas(see \cite{10.1145/2395116.2395118} and the references therein). To the best of our knowledge, this is the first quantum invariance principle and the result in this paper is the first application of the invariance principle in quantum complexity theory. We believe that the tools developed in this paper are interesting in their own right and have further applications.

\subsection{Technical Contributions} In this paper, we treat the set $\M_m$  of $m\times m$ matrices as a Hilbert space of dimension $m^2$ by introducing the inner product $\innerproduct{A}{B}=\frac{1}{m}\Tr A^{\dagger}B$. Let $\B=\set{\B_0,\B_1,\ldots,\B_{m^2-1}}$ be an orthonormal basis in $\M_m$ with all elements being Hermitian and $\B_0=\id$, whose existence is guaranteed by \cref{lem:paulibasis}. It is easy to verify that the set $\set{\B_{\sigma}:\sigma\in\set{0,\ldots, m^2-1}^n}$, where $\B_{\sigma}\defeq\B_{\sigma_1}\otimes\B_{\sigma_2}\otimes\ldots\otimes\B_{\sigma_n}$, forms an orthonormal basis in $\M_m^{\otimes n}$. Any operator $P\in\M_m^{\otimes n}$ can be expressed as
\begin{equation}\label{eqn:Pexpression}
  P=\sum_{\sigma\in[m^2]_{\geq 0}^n}\widehat{P}\br{\sigma}\B_{\sigma},
\end{equation}
which can be viewed as a Fourier expansion of $P$ in terms of the basis $\B$, where $[n]_{\geq 0}$ represents the set $\set{0,\ldots, n-1}$.

\subsection*{Quantum invariance principle and quantum hypercontractive inequality}
Invariance principles are a powerful tool in Boolean analysis, which has found applications in various areas in theoretical computer science, such as inapproximation theory, derandomization, voting theory, etc~\cite{MosselOdonnell:2010}. Let's recall the invariance principle for functions in~\cite{MosselOdonnell:2010}. Let $f:\set{1,-1}^n\rightarrow\reals$ be a bounded-degree multilinear polynomial with small influence for all variables. An invariance principle asserts that
\begin{equation}\label{eqn:inv}
\expec{\mathbf{x}\sim\set{1,-1}^n}{\Psi\br{f\br{\mathbf{x}}}}\approx\expec{\mathbf{g}\sim\gamma_n}{\Psi\br{f\br{\mathbf{g}}}}
\end{equation}
for any Lipschitz continuous function $\Psi:\reals\rightarrow\reals$, where $\gamma_n$ is a standard $n$-dimensional normal distribution. Here the expression "$f\br{\mathbf{g}}$" is an abuse of notation indicating that the real numbers $\mathbf{g}_1,\ldots,\mathbf{g}_n$ are substituted into $f$'s Fourier expansion. The power of Eq. \cref{eqn:inv} is that we are able to interchange arbitrary random variables with Gaussian random variables.

One of our main contributions is establishing a quantum invariance principle. Let $P\in\M_m^{\otimes n}$ be a Hermitian operator with a Fourier expansion in Eq. \cref{eqn:Pexpression} satisfying that all registers have low influence (which will be specified). Here we view $P$ as an operator acting on $n$-partite quantum systems, where each system is of dimension $m$. The term "register" is referred to each system. For a $\C^2$ piecewise polynomial $\Psi$, it holds that
\begin{equation}\label{eqn:qinv}
\frac{1}{m^n}\Tr~\Psi\br{\sum_{\sigma\in[m^2]_{\geq 0}^n}\widehat{P}\br{\sigma}\B_{\sigma}}\approx\expec{\mathbf{g}\sim\gamma_{\br{m^2-1}n}}{\Psi\br{\sum_{\sigma\in[m^2]_{\geq 0}^n}\widehat{P}\br{\sigma}\prod_{i=1}^{n}\mathbf{g}_{i,\sigma_i}}},
\end{equation}
where $\mathbf{g}=\br{\mathbf{g}_{1,1},\ldots,\mathbf{g}_{_1,m^2-1},\ldots, \mathbf{g}_{n,1},\ldots,\mathbf{g}_{_1,m^2-1}}\sim\gamma_{n\br{m^2-1}}$ and $\mathbf{g}_{1,0}=\cdots=\mathbf{g}_{n,0}=1$.
Eq. \cref{eqn:qinv} enables us to turn the quantum registers to Gaussian random variables, which gives polynomials in Gaussian spaces.

The proof of Eq.~\cref{eqn:inv} employs the well-known Lindeberg-style method (see e.g., Chapter 11 in~\cite{Odonnell08}). It first substitutes $\Psi$ by a $\C^{\infty}$ approximation $\Psi'$, whose existence follows by standard results in approximation theory. Then by expanding both sides of Eq.~\cref{eqn:inv} to the third-order derivatives via Taylor series, we can prove that the difference between the left-hand side and the right-hand side is upper bounded by the norms of the third-order terms in the Taylor expansion of $\Psi'$. Applying the hypercontractive inequality for Boolean variables and the one for Gaussian variables, respectively, we can see that the difference between the two sides is small for functions with all variables having low influence.

Generalizing the Lindeberg-style method to matrices is not an easy task due to the non-commutativity of matrices. To this end, we need to investigate the analytical properties of matrix-functions. We use Fr\'echet derivatives, a notion of derivatives in Banach spaces, in which a similar form of Taylor expansions exists.  The differentiability of real functions and of matrix-functions with respect to Fr\'echet derivatives share many properties in common~\cite{SENDOV2007240}. We follow the same mechanism as the original Lindeberg-style method by substituting the basis elements with Gaussian variables and obtain random operators, which are hybrids of operators and random variables. The Taylor expansions of matrix-valued functions are in general complicated due to the nature of the non-commutativity as well. Fortunately, it suffices to prove a quantum invariance principle for a $\C^2$ piecewise polynomial for our purpose.

We prove that the difference between the two sides in Eq. \cref{eqn:qinv} is upper bounded by the 3rd order term in the Taylor expansion of $\Psi$, which, in turn, is upper bounded by the product of its $2$-norm and the square of its $4$-norm. Applying the hypercontractive inequality for random operators, we are able to upper bound the difference by the cube of its $2$-norm, which is further upper bounded by the influences of all registers.

Another difficulty that follows is to prove a hypercontractive inequality for random operators, which is expected to show that the $4$-norm of a bounded-degree random operator can be upper bounded by its $2$-norm. As random operators are hybrids of operators and Gaussian variables, our proof delicately combines the hypercontractivity for Gaussian variables and the hypercontractivity for operators. The former has already been proved in~\cite{PawelWolff2007}. King in~\cite{King2003} proved a hypercontractive inequality for unital qubit channels, which immediately implies a hypercontractive inequality for operators in $\M_2^{\otimes n}$. However, since his proof heavily relies on the properties of qubit channels, it fails to generalize to $\M_m^{\otimes n}$ for $m>2$. This paper focuses only on the hypercontractivity for depolarizing channels in any dimension. Thus we use the properties specific to depolarizing channels to prove a hypercontractive inequality for operators via an inductive argument. The desired result follows after we combine a hypercontractive inequality for operators and a hypercontractive inequality for random variables in~\cite{PawelWolff2007}.

\subsection*{Quantum dimension reductions} Dimension reductions are a versatile tool in theoretical computer science. One of the most well-known dimension reduction techniques is arguably the Johnson–Lindenstrauss lemma~\cite{JL:1984}, which enables us to embed a set of points from a high-dimensional space into a space of much lower dimension with slight distortion. Dimension reductions are a natural approach to upper bound the complexity of quantum nonlocal games. However, proving quantum dimension reductions seems to be difficult and sometimes even impossible~\cite{10.1007/978-3-642-22006-7_8,7426395}. In this paper, we present a new quantum dimension reduction technique by establishing a new quantum invariance principle and combining it with a recent dimension reduction result in Gaussian spaces~\cite{Ghazi:2018:DRP:3235586.3235614}.

\bigskip

To the best of our knowledge, this is the first invariance principle in the quantum setting. We are also not aware of any other invariance principle proved via Fr\'echet derivatives prior to our result. We believe our invariance principle and quantum dimension reduction are interesting in their own right and should have further applications.

\subsection{Proof Overview}

We are interested in the decidability of the following decision problem, which is a special case of local state transformations.

\bigskip

\textbf{Q}. Given $0<\epsilon,\delta<1$, integers $t,m\geq 2$ and a distribution $\mu$ over $[t]\times[t]$, suppose Alice and Bob share arbitrarily many copies of the bipartite quantum states $\Psi_{\epsilon}=\br{1-\epsilon}\ketbra{\Psi}+\epsilon\frac{\id_m}{m}\otimes\frac{\id_m}{m}$, where $\ket{\Psi}=\frac{1}{\sqrt{m}}\sum_{i=0}^{m-1}\ket{i,i}$ is an $m$-dimensional maximally entangled state.

\noindent \textbf{Yes}. Alice and Bob are able to jointly sample a distribution $\mu'$ over $[t]\times[t]$ which is $\delta$-close to $\mu$, i.e., $\onenorm{\mu'-\mu}\leq\delta$.

\noindent \textbf{No}. Any distribution $\mu'$ over $[t]\times[t]$ that Alice and Bob can jointly sample is $2\delta$-far from $\mu$, i.e., $\onenorm{\mu-\mu'}\geq2\delta$.

\bigskip

\subsection*{Special case: $\mu$ is a binary distribution}

We first consider the special case that $\mu$ is a binary distribution, namely $t=2$. Suppose the POVM's performed by Alice and Bob are  $\set{P,\id-P}$ and $\set{Q,\id-Q}$, respectively, where $P, Q\in\H_m^{\otimes n}$ and $0\leq P,Q\leq\id$. Here $\H_m^{\otimes n}$ represents the set of all Hermitian operators acting on an $n$-qudit quantum system, where each qudit is of dimension $m$. The proof is to construct a universal bound $D$, which is independent of the measurements, and transformations $f_n,g_n:\H_m^{\otimes n}\rightarrow\H_m^{\otimes D}$, such that the requirements in \cref{fig:requirements} are satisfied.

\begin{figure}[h]
	\begin{mdframed}
		\textbf{Requirements.}
		\begin{enumerate}
			\item $0\leq f_n\br{P}\leq\id~\mbox{and}~0\leq g_n\br{Q}\leq\id;$
			\item $\frac{1}{m^n}\Tr~P\approx\frac{1}{m^D}\Tr~f_n\br{P} ~\mbox{and}~ \frac{1}{m^n}\Tr~Q_n\approx\frac{1}{m^D}\Tr~g_n\br{Q};$
			\item $\Tr\br{\br{P\otimes Q}\Psi_{\epsilon}^{\otimes n}}\approx\Tr\br{\br{f_n\br{P}\otimes g_n\br{Q}}\Psi_{\epsilon}^{\otimes D}}.$
		\end{enumerate}
	\end{mdframed}
	\caption{Requirements for binary output distributions}\label{fig:requirements}
\end{figure}

The first item implies that $\set{f_n\br{P},\id-f_n\br{P}}$ and $\set{g_n\br{Q},\id-g_n\br{Q}}$ are both valid POVMs. The second item implies that the probability that Alice outputs $0$ is almost unchanged under the transformation $f_n$.  Same for the probability that Bob outputs $0$. The last item implies that the probability that both Alice and Bob output $0$ is almost unchanged. As Alice's and Bob's outputs are both binary, it follows that the distribution of the joint output is almost unchanged.

The constructions of the transformations $f_n$ and $g_n$ are built on the framework introduced in~\cite{7782969}, which, in turn, was based on the results in Boolean analysis and Hermitian analysis~\cite{MosselOdonnell:2010,Mossel:2010}. Let $\A=\set{\A_0,\ldots,\A_{m^2-1}}$ and $\B=\set{\B_0,\ldots,\B_{m^2-1}}$, which will be specified later, be two properly chosen orthonormal bases in $\M_m$. The Fourier expansions of $P$ and $Q$ on bases $\A$ and $\B$ can be expressed as
\[P=\sum_{\sigma\in[m^2]_{\geq 0}^n}\widehat{P}\br{\sigma}\A_{\sigma}~\mbox{and}~Q=\sum_{\sigma\in[m^2]_{\geq 0}^n}\widehat{Q}\br{\sigma}\B_{\sigma},\]
respectively, where $[n]_{\geq 0}$ represents the set $\set{0,\ldots, n-1}$. The steps for constructing the transformations $f_n$ and $g_n$ are specified as follows.

\begin{itemize}
	\item \textbf{Smooth operators}.
	
	\textbf{Objective}:
     \begin{itemize}
     \item bounded-degree approximation.
	 \end{itemize}
\medskip

 We first convert the operators $\br{P,Q}$ to a new pair of operators that can be approximated by bounded-degree operators via smoothing operations. The standard techniques in Boolean analysis~\cite{Odonnell08} motivate us to apply a {\em noise operator} $\Delta_{\rho}$ for some $\rho\in(0,1)$ to both $P$ and $Q$. Note that the noise operator is also referred to as a depolarizing channel in quantum information theory~\cite{NC00}, and its counterpart in the Gaussian space is called {\em Ornstein-Uhlenbeck operator}.
Then we obtain
	
	\[\Delta_{\rho}P=\sum_{\sigma\in[m^2]_{\geq 0}^n}\widehat{P}\br{\sigma}\rho^{|\sigma|}\A_{\sigma}~\mbox{and}~\Delta_{\rho}Q=\sum_{\sigma\in[m^2]_{\geq 0}^n}\widehat{Q}\br{\sigma}\rho^{|\sigma|}\B_{\sigma},\]
	where $|\sigma|=\abs{\set{i:\sigma_i\neq 0}}$.
	Here $\Delta_{\rho}P$ and $\Delta_{\rho}Q$ are independent of the choices of bases by \cref{def:bonamibeckner} and \cref{lem:bonamibecknerdef} item 1. It is not hard to verify that the requirements specified in item 1 and item 2 in \cref{fig:requirements} are satisfied. To meet the requirements in item 3, we need the notion of {\em quantum maximal correlation} introduced by Beigi~\cite{Beigi:2013}, which extends the {\em maximal correlation coefficients}~\cite{hirschfeld:1935,Gebelein:1941,Renyi1959} in classical information theory. Note that after smoothing the operators, the weights of high-degree parts of both $P$ and $Q$, namely, the parts with high $\abs{\sigma}$, decrease exponentially. Thus, both operators can be approximated by bounded-degree operators.
	
	\item \textbf{Joint regularity}.
	
	\textbf{Objective}:
\begin{itemize}
  \item bounded-degree approximation;
  \item bounded number of high-influential registers.
\end{itemize}
\medskip

	We prove that the number of high-influential registers is bounded. Let $H$ be the set of all high-influential registers.  The influence of a register to a Hermitian operator on a multipartite quantum system is defined in \cref{def:influencegaussian}, which was first introduced by Montanaro in~\cite{doi:10.1063/1.4769269}, analogous to the {\em influence} defined in Boolean analysis~\cite{Odonnell08}. It informally speaking, measures how much the register can influence the operator. This step follows by quantizing a well-known result in Boolean analysis. For any bounded function, the total influence, i.e., the summation of the influences of all variables, is upper bounded by the degree of the function. Note that both $P$ and $Q$ can be approximated by bounded-degree operators after the first step. By Markov's inequality, the size of $H$ can be bounded.
	
	\item \textbf{Invariance from operators to random operators}.

	\textbf{Objective}:
\begin{itemize}
\item bounded-degree approximation;
\item bounded number of quantum registers;
\item unbounded number of Gaussian variables.
\end{itemize}
	\medskip
	
	In this step, we use correlated Gaussian variables to substitute for all the quantum registers in $P$ and $Q$ with low influence, after which we get random operators.

Prior to the substitution, we prove that there exist bases $\A$ and $\B$ such that
	\[\Tr\br{\br{\A_i\otimes\B_j}\psi_{AB}}=\delta_{i,j}c_i,\]
	for $1=c_0\geq c_1\geq c_2\geq \cdots\geq c_{m^2-1}\geq 0$, where $\A_0=\B_0=\id$.
	Hence
	\begin{equation}\label{eqn:pqreg}
	\Tr\br{\br{P\otimes Q}\psi_{AB}^{\otimes n}}=\sum_{\sigma\in[m^2]_{\geq 0}^n}c_{\sigma}\widehat{P}\br{\sigma}\widehat{Q}\br{\sigma},
	\end{equation}
	where $c_{\sigma}=c_{\sigma_1}\cdot c_{\sigma_2}\cdots c_{\sigma_n}$ for any $\sigma\in[m^2]_{\geq 0}^n$.

	We now introduce independent joint random variables $$\set{\br{\mathbf{g}_{i,j},\mathbf{h}_{i,j}}}_{i\notin H,j\in[m^2]_{\geq 0}},$$ where $\mathbf{g}_{i,0}=\mathbf{h}_{i,0}=1$ and $\br{\mathbf{g}_{i,j},\mathbf{h}_{i,j}}\sim\G_{c_j}$ for $i\notin H,j\geq 1$ and $\rho$-correlated two-dimensional Gaussian distribution $\G_{\rho}$. Given the Fourier expansions of $P$ and $Q$, we substitute all the matrix bases in the quantum registers not in $H$ by $\br{\mathbf{g}_{i,j},\mathbf{h}_{i,j}}$, respectively, to obtain the random operators distributed over all operators acting on the registers in $H$ as follows.
	
	\begin{equation}\label{eqn:intro2}
	\mathbf{P}=\sum_{\sigma\in[m^2]_{\geq 0}^n}\widehat{P}\br{\sigma}\br{\prod_{i\notin H}\mathbf{g}_{i,\sigma_i}}\A_{\sigma_H}\end{equation}
and
$$
\mathbf{Q}=\sum_{\sigma\in[m^2]_{\geq 0}^n}\widehat{Q}\br{\sigma}\br{\prod_{i\notin H}\mathbf{h}_{i,\sigma_i}}
	\B_{\sigma_H}.$$
	
	It is easy to verify the requirements specified by item 2 and item 3 in \cref{fig:requirements} are satisfied in expectation. However, $\mathbf{P}$ and $\mathbf{Q}$ are in general not POVM elements, which means that item 1 is violated. To meet item 1, it suffices to show that both of the random operators are $\ell_2$-close to POVM elements in expectation. Let $\R:\H_m^{\otimes n}\rightarrow\H_m^{\otimes n}$ be a rounding map of the set of POVM elements, namely $\R\br{X}=\arg\min\set{\twonorm{X-Y}:0\leq Y\leq\id}$, where $\twonorm{\cdot}$ is the Schatten 2-norm. Hence, we prove that

\begin{equation}\label{eqn:PRP}
  \expec{}{\twonorm{\mathbf{P}-\R\br{\mathbf{P}}}^2}\approx 0~\mbox{and}~\expec{}{\norm{\mathbf{Q}-\R\br{\mathbf{Q}}}^2_2}\approx 0.
\end{equation}
It is not hard to verify that
$\twonorm{\R(X)-X}^2=\Tr~\theta\br{X}$, where
 \begin{equation}\label{eqn:zetasingle}
 	\theta\br{x}=\begin{cases}
 	(x-1)^2~&\mbox{if $x\geq 1$}\\
 	x^2~&\mbox{if $x\leq 0$}\\
 	0~&\mbox{otherwise}.
 	\end{cases}
 \end{equation}
Thus, proving \cref{eqn:PRP} is equivalent to proving that
\begin{equation}\label{eqn:PthetaP}
  \expec{}{\Tr~\theta\br{\mathbf{P}}}\approx0~\mbox{and}~\expec{}{\Tr~\theta\br{\mathbf{Q}}}\approx0.
\end{equation}

As described above, we need to prove an invariance principle for the function $\theta$ defined in Eq.~\cref{eqn:PthetaP}, which is a $\C^1$ piecewise polynomial. We are able to construct a $\C^2$ approximation of $\theta$, denoted by $\theta_{\lambda}$, which is also a piecewise polynomial. We prove an invariance principle for $\theta_{\lambda}\br{\cdot}$ and show that both $\expec{}{\sum_i\Tr~\theta\br{\mathbf{P}_i}}$ and $\expec{}{\sum_i\Tr~\theta\br{\mathbf{Q}_i}}$ are upper bounded by the 3rd order terms in the Taylor expansion of $\theta_{\lambda}$, which, in turn, are upper bounded by the cube of their $2$-norm due to a hypercontractive inequality for random operators.

	\item \textbf{Dimension reduction}.
	
	\textbf{Objective}:
\begin{itemize}
\item bounded number of quantum registers;
\item bounded number of Gaussian variables.
\end{itemize}
\medskip

	For any $\sigma_H\in[m^2]^{|H|}_{\geq 0}$, let
\[p_{\sigma_H}\br{\vec{\mathbf{g}}}=\sum_{\tau\in[m^2]^n_{\geq 0}:\tau_H=\sigma_H}\widehat{P}\br{\tau}\br{\prod_{i\notin H}\mathbf{g}_{i,\tau_i}}\]
and
\[q_{\sigma_H}\br{\vec{\mathbf{h}}}=\sum_{\tau\in[m^2]^n_{\geq 0}:\tau_H=\sigma_H}\widehat{Q}\br{\tau}\br{\prod_{i\notin H}\mathbf{g}_{i,\tau_i}},\]
where $p_{\sigma_H},q_{\sigma_H}:\reals^{\br{m^2-1}(n-|H|)}\rightarrow\reals$ and \[\vec{\mathbf{g}}=\br{\mathbf{g}_{i,j}}_{i\notin H,j\in[m^2-1]}, \vec{\mathbf{h}}=\br{\mathbf{h}_{i,j}}_{i\notin H,j\in[m^2-1]}.\]

Then
$\br{p_{\sigma_H}\br{\vec{\mathbf{g}}}}_{\sigma_H\in[m^2]^{|H|}_{\geq 0}}$ and $\br{q_{\sigma_H}\br{\vec{\mathbf{h}}}}_{\sigma_H\in[m^2]^{|H|}_{\geq 0}}$ are the Fourier coefficients of $\mathbf{P}$ and $\mathbf{Q}$, respectively. Namely,
	\begin{equation}\label{eqn:intro3}
	\mathbf{P}=\sum_{\sigma_H\in[m^2]_{\geq 0}^{|H|}}p_{\sigma_H}\br{\vec{\mathbf{g}}}\A_{\sigma_H}~\mbox{and}~\mathbf{Q}=\sum_{\sigma_H\in[m^2]_{\geq 0}^{|H|}}q_{\sigma_H}\br{\vec{\mathbf{h}}}\B_{\sigma_H},
	\end{equation}
	
	After applying the {\em dimension reduction} in~\cite{Ghazi:2018:DRP:3235586.3235614} to polynomials $\br{p_{\sigma_H},q_{\sigma_H}}$ , the number of Gaussian random variables is reduced to a bounded number. Meanwhile, all the requirements listed in \cref{fig:requirements} are still satisfied in expectation.

	\item \textbf{Smooth random operators}.
	
	\textbf{Objective}:
\begin{itemize}
\item bounded-degree approximation;
\item bounded number of quantum registers;
\item  bounded number of Gaussian variables.
\end{itemize}
\medskip	
		The remaining steps are mainly concerned with removing the Gaussian variables in $\mathbf{P}$ and $\mathbf{Q}$ to obtain  POVM elements. To this end, we perform transformations similar to the ones in the previous steps.  We first apply the noise operator again to both operators to reduce the weight of the high degree parts.
	
	\item \textbf{Multilinearization}.
	
	\textbf{Objective}:

\begin{itemize}
\item bounded-degree approximation;
\item bounded number of quantum registers;
\item  bounded number of Gaussian variables;
\item multilinear.
\end{itemize}
	\medskip
	
	The degrees of Gaussian variables occurring in the functions might be unbounded. Prior to the substitution of Gaussian variables by quantum registers, we need the multilinearization lemma in~\cite{Ghazi:2018:DRP:3235586.3235614} to reduce the power of each Gaussian variable to either $0$ or $1$. Namely, the polynomials $p_{\sigma_H}$'s and $q_{\sigma_H}$'s in Eq.~\cref{eqn:intro3} are all multilinear after this step.

	\item \textbf{Invariance from random operators to operators}.

\textbf{Objective}:
\begin{itemize}
  \item bounded number of quantum registers.
\end{itemize}
\medskip
		
	In the final step, we substitute the Gaussian variables by properly chosen basis elements and round both operators to  POVM elements. Again, we need to apply a quantum invariance principle to ensure that all the requirements in \cref{fig:requirements} are satisfied.
\end{itemize}

\subsection*{General case: $\mu$ is a general bipartite distribution}

Suppose Alice and Bob share $n$ copies of $\Psi_{\epsilon}$ and the POVMs Alice and Bob peform are $\br{P_1,\ldots, P_t}$ and $\br{Q_1,\ldots, Q_t}$, respectively.
We need to construct a universal bound $D$, which is independent of the measurements, and the transformations $f_n,g_n:\H_m^{\otimes n}\rightarrow\H_m^{\otimes D}$ such that the requirements in \cref{fig:requirementsnonbinary} are satisfied. Note that item 1 (b) requires that the output of the  transformations is a pair of sub-POVMs rather than POVMs. However, it can be proved that this technical constraint does not affect our final conclusion.

\begin{figure}[h]
	\begin{mdframed}
		\textbf{Requirements}
		
		\bigskip
		
		Let $\br{\widetilde{P}_1,\ldots, \widetilde{P}_t}=f_n\br{P_1,\ldots, P_t}$ and $\br{\widetilde{Q}_1,\ldots,\widetilde{Q}_t}=g_n\br{Q_1,\ldots, Q_t}$.
		
		\begin{enumerate}
			\item $\br{\widetilde{P}_1,\ldots,\widetilde{P}_t}$ and $\br{\widetilde{Q}_1,\ldots,\widetilde{Q}_t}$ are valid sub-POVMs. Namely,
\begin{enumerate}
  \item $\forall i,j\in[t]~ 0\leq\widetilde{P}_i\leq \id$ and $0\leq\widetilde{Q}_i\leq \id$.
  \item $\sum_{i=1}^{t}\widetilde{P}_i\leq\id$ and $\sum_{i=1}^{t}\widetilde{Q}_i\leq\id$.
\end{enumerate}
			
			\item $\forall i,j\in[t]~ \Tr\br{\br{P_i\otimes Q_j}\Psi_m^{\otimes n}}\approx\Tr\br{\br{\widetilde{P}_i\otimes\widetilde{Q}_j}\Psi_m^{\otimes D}}.$
		\end{enumerate}
	\end{mdframed}
	\caption{Requirements for general output distributions}\label{fig:requirementsnonbinary}
\end{figure}

A natural approach is to apply the transformations $f_n$ and $g_n$ constructed for the binary case to all $P_i$'s and $Q_i$'s, respectively. Note that Eq.\cref{eqn:PthetaP} and the quantum invariance principle guarantee item 1(a). However, as the transformations are applied to $P_i$'s and $Q_i$'s individually, we cannot directly conclude item 1(b). The difficulty arises when trying to prove item 1(b) in \cref{fig:requirementsnonbinary} as it enforces global constraints on $P_i$'s and $Q_i$'s.

Given $X_1,\dots,X_t\in\H_m^{\otimes n}$, let \[\R\br{X_1,\ldots,X_t}=\arg\min\set{\sum_{i=1}^t\twonorm{X_i-Y_i}^2:\br{\forall i}0\leq Y_i\leq\id, \sum_i Y_i\leq\id}\]
be a rounding map of the set of all sub-POVMs. And
\begin{equation}\label{eqn:theta}
	\theta\br{X_1,\ldots, X_t}=\twonorm{\br{X_1,\ldots,X_t}-\R\br{X_1,\ldots, X_t}}^2.
\end{equation}

To obtain a similar result as  Eq.~\cref{eqn:PthetaP}, an intuitive approach is to prove a quantum version of the multivariable invariance principle~\cite{IMossel:2012}. However, this is technically challenging. In the classical setting, i.e., all $X_i$'s are real numbers, to establish a multivariable invariance principle, it requires us to compute the Taylor expansion of a smooth approximating function for $\theta$~\cite{IMossel:2012}. However, the high-order Fr\'echet derivatives of the functions with multiple matrices as inputs are highly involved in general.

To avoid the difficulty of analyzing functions with multiple matrices as inputs, we upper bound the function $\theta$ in Eq.~\cref{eqn:theta} by a function that takes only a single matrix as input. Namely,

\begin{equation}\label{eqn:rzeta}
 \theta\br{P_1,\ldots, P_t}\leq O\br{t\cdot\sum_{i=1}^t\Tr~\zeta\br{P_i}},
\end{equation}
where
\begin{equation}\label{eqn:introzeta}
	\zeta\br{x}=\begin{cases}x^2~&\mbox{if $x\leq 0$}\\ 0~&\mbox{otherwise}\end{cases}.
\end{equation}
With Eq.~\cref{eqn:rzeta}, it is sufficient for us to prove that

\begin{equation}\label{eqn:PzetaP}
  \expec{}{\sum_i\Tr~\zeta\br{\mathbf{P}_i}}\approx0~\mbox{and}~\expec{}{\sum_i\Tr~\zeta\br{\mathbf{Q}_i}}\approx0,
\end{equation}
which is again obtained by a quantum invariance principle described in the previous subsection.

The constructions of the maps $f_n$ and $g_n$ are summarized in \cref{fig:construction}.

\begin{figure}[!htbp]
\centering
\label{fig:construction}
\begin{tabular}{cccc}
&$P$&$Q$&\tabincell{$\H_m^{\otimes n}$\\ $ 0\leq P,Q\leq\id$}\\

\subtabtwo{\textbf{Smooth}\\\small{\textbf{objective}: bounded deg}}{\cref{lem:smoothing of strategies}}{$P^{(1)}$}{$Q^{(1)}$}{$\H_m^{\otimes n}$\\ $ 0\leq P^{(1)},Q^{(1)}\leq\id$}

\subtabone{\textbf{Regularization}\\\small{bounded deg}\\\small{bounded high inf registers}}{\cref{lem:regular}}{$P^{(1)}$}{$Q^{(1)}$}{$\H_m^{\otimes n}$\\ $ 0\leq P^{(1)},Q^{(1)}\leq\id$}

\subtabtwo{\textbf{Invariance principle}\\\small{bounded deg}\\\small{bounded matrix bases}\\\small{unbounded Gaussian vars}}{\cref{lem:jointinvariance}}{$\mathbf{P}^{(2)}$}{$\mathbf{Q}^{(2)}$}{$L^2\br{\H^{\otimes h}_m,\gamma_{2(m^2-1)(n-h)}}$}

\subtabone{\textbf{Dimension reduction}\\\small{bounded q. registers}\\\small{bounded Gausian vars}}{\cref{lem:dimensionreduction}}{$\mathbf{P}^{(3)}$}{$\mathbf{Q}^{(3)}$}{$L^2\br{\H^{\otimes h}_m,\gamma_{n_0}}$}

\subtabtwo{\textbf{Smooth}\\\small{bounded q. registers}\\\small{bounded Gaussian vars}\\\small{bounded deg}}{\cref{lem:smoothgaussian}}{$\mathbf{P}^{(4)}$}{$\mathbf{Q}^{(4)}$}{$L^2\br{\H^{\otimes h}_m,\gamma_{n_0}}$}

\subtabtwo{\textbf{Multilinearization}\\\small{bounded q. registers}\\\small{bounded Gaussian vars}\\\small{bounded deg \& multilinear}}{\cref{lem:multiliniearization}}{$\mathbf{P}^{(5)}$}{$\mathbf{Q}^{(5)}$}{$L^2\br{\H^{\otimes h}_m,\gamma_{n_0n_1}}$}

\subtabtwo{\textbf{Invariance principle}\\\small{bounded q. registers}}{\cref{lem:invariancejointgaussian}}{$P^{(6)}$}{$Q^{(6)}$}{$\H_m^{h+n_0n_1}$}

\subtabtwo{\textbf{Rounding}\\\small{sub-POVMs}}{\cref{lem:closedelta}}{$\widetilde{P}$}{$\widetilde{Q}$}{$\H_m^{h+n_0n_1}$\\$0\leq\widetilde{P},\widetilde{Q}\leq\id$\\$\sum_i\widetilde{P}_i\leq\id,\sum_i\widetilde{Q}_i\leq\id$}
\end{tabular} \caption{Construction of the transformations}
\end{figure}

Note that in nonlocal games the players perform different measurements for different inputs. Thus, we may encounter consistency issues when applying the transformations mentioned above to nonlocal games. Suppose the players' strategies are $\br{\set{P,\id-P},\set{Q,\id-Q}}$ and $\br{\set{P,\id-P},\set{Q',\id-Q'}}$ for questions $(x,y)$ and $\br{x,y'}$, respectively. We need to ensure that the resulting POVMs on Alice's side are same in both cases. Namely, the transformation for each player should be independent of the other player. Note that only the regularization and dimension reduction steps jointly depend on both operators. Fortunately, the dependence between the two players in both steps can be removed by applying a union bound on all of the possible question pairs.

\section*{Organization} \Cref{sec:pre} summarizes some useful concepts and basic facts on quantum mechanics, the analysis on Gaussian spaces, matrix spaces and random operators. \Cref{sec:markov} introduces the notion of quantum maximal correlations, a key concept in this work, with several crucial properties. The main results and the proofs are stated and proved in \Cref{sec:mainresult}. A summary with further work and open problems are listed in \Cref{sec:openproblems}. The step of smoothing operators is in ~\Cref{sec:smoothoperators}. ~\Cref{sec:jointregular} proves a joint regularity lemma. \Cref{sec:hypercontractive} proves a hypercontractive inequality for random operators. \Cref{sec:reduction} presents a reduction from the rounding distance to the function $\Tr\zeta\br{\cdot}$. Quantum invariance principles are proved in \Cref{sec:invariance}. \Cref{sec:dimensionreduction} proves a dimension reduction for random operators. The step of smoothing random operators is presented in \Cref{sec:smoothrandom}. The step of multilinearization is proved in \Cref{sec:multilinear}. \Cref{sec:frechet} summarizes  basic facts on Fr\'echet derivatives. \Cref{sec:zetataylor} and \Cref{sec:appinvariance} present some deferred proofs. \Cref{sec:tablenotations} summarizes all the notations in a table for readers to look up.

	\section{Preliminaries}\label{sec:pre}
	This paper uses a number of notations. For readers' convenience, all the notations are summarized in \cref{sec:tablenotations}. For an integer $n\geq 1$, let $[n]$ and $[n]_{\geq 0}$ represent the sets $\set{1,\ldots, n}$ and $\set{0,\ldots, n-1}$, respectively. Given a finite set $\X$ and a natural number  $k$, let $\X^k$ be the set $\X\times\cdots\times\X$, the Cartesian product of $\X$, $k$ times. Given $a=a_1,\ldots, a_k$ and a set $S\subseteq[k]$, we write $a_S$ to represent the projection of $a$ to the coordinates specified in $S$. For any $i\in[k]$, $a_{-i}$ represents $a_1,\ldots, a_{i-1},a_{i+1},\ldots, a_n$ and $a_{<i}$ represents $a_1,\ldots, a_{i-1}$. $a_{\leq i},a_{>i}, a_{\geq i}$ are defined similarly. For any $\sigma\in\mathbb{Z}_{\geq 0}^k$, we define $\abs{\sigma}=\abs{\set{i:\sigma_i\neq 0}}$ and $\mathrm{wt}\br{\sigma}=\sum_i\sigma_i$. Let $\mu$ be a probability distribution on $\X$, and $\mu\br{x}$ represent the probability of $x\in\X$ according to $\mu$. Let $X$ be a random variable distributed according to $\mu$. We use the same symbol to represent a random variable and its distribution whenever it is clear from the context. The expectation of a function $f$ on $\X$ is defined as \[\expec{}{f(X)}=\expec{\mathbf{x}\sim X}{f(\mathbf{x})}=\sum_{x\in\X}\prob{X=x}\cdot f\br{x}=\sum_x\mu\br{x}\cdot f\br{x},\] where $\mathbf{x}\sim X$ represents that $\mathbf{x}$ is drawn according to $X$. For any two distributions $\mu$ and $\nu$, $\br{\mathbf{x},\mathbf{y}}\sim\mu\otimes\nu$ represents that $\mathbf{x}$ and $\mathbf{y}$ are drawn from $\mu$ and $\nu$, independently. For two distributions $p$ and $q$, the $\ell_1$-distance between $p$ and $q$ is defined to be \[\onenorm{p-q}=\sum_x\abs{p\br{x}-q\br{x}}.\]

	In this paper, the lower-cased letters in bold $\mathbf{x},\mathbf{y},\cdots$ are reserved for random variables. The capital letters in bold, $\mathbf{P},\mathbf{Q},\ldots$ are reserved for random operators defined below.

	\subsection{Quantum mechanics}

	We briefly review the formalism of quantum mechanics over a finite dimensional system. For a more thorough treatment, readers may refer to~\cite{NC00,Wat08}. Given a quantum system $A$, it is associated with a finite dimensional Hilbert space, which, by abuse of notation, is also denoted by $A$. We denote by $\M\br{A}$ and $\H\br{A}$ the set of all linear operators and the set of all Hermitian operators in the space, respectively. The identity operator in $A$ is denoted by $\id_A$.  If the dimension of $A$ is $m$, then we write $\M\br{A}=\M_m$, $\H\br{A}=\H_m$ and $\id_A=\id_m$. The subscripts may be dropped whenever it is clear from the context. A quantum state in the quantum system $A$ is represented by a {\em density operator} $\rho_A$, a positive semidefinite operator over the Hilbert space $A$ with unit trace. We denote by $\D\br{A}$ the set of all density operators in $A$. A quantum state is {\em pure} if the density operator is a rank-one projector $\ketbra{\psi}$, which is also represented by $\ket{\psi}$ for convenience. Composite quantum systems are associated with the {\em (Kronecker) tensor product space} of the underlying spaces, i.e., for the quantum systems $A$ and $B$, the composition of the two systems are represented by $A\otimes B$ with the sets of the linear operators, the Hermitian operators and the density operators denoted by $\M\br{A\otimes B}$, $\H\br{A\otimes B}$ and $\D\br{A\otimes B}$, respectively. We sometimes use the shorthand $AB$ for $A\otimes B$. The sets of the linear operators and the Hermitian operators in the composition of $n$ $m$-dimensional Hilbert spaces are denoted by $\M_m^{\otimes n}$ and  $\H_m^{\otimes n}$, respectively. With a slight abuse of notations, we assume that $\M_m^{\otimes n}=\complex$ and $\H_m^{\otimes n}=\reals$ when $n=0$. A {\em quantum channel} from the input system $A$ to the output system $B$ is represented by a {\em completely positive, trace-preserving linear map} (CPTP map).  The set of all quantum channels from a system $A$ to a system $B$ is denoted by $\L\br{A,B}$. A quantum channel in $\L\br{A,B}$ is {\em unital} if it maps $\id_A$ to $\id_B$. A {\em quantum operator} on $A$ is a channel with both the input system and the output system being $A$. The set of the quantum operators on $A$ is denoted by $\L\br{A}$. An important operation on a composite system $A\otimes B$ is the {\em partial trace} $\Tr_B\br{\cdot}$ which effectively derives the marginal state of the subsystem $A$ from the quantum state $\rho_{AB}$.  The partial trace is given by \[\Tr_B\rho_{AB}=\sum_i\br{\id_A\otimes\bra{i}}\rho_{AB}\br{\id_A\otimes\ket{i}}\] where $\set{\ket{i}}$ is an orthonormal basis in $B$. The partial trace is a valid quantum channel in $\L\br{A\otimes B, A}$. Note that the action is independent of the choices of basis $\set{\ket{i}}$, so we unambiguously write $\rho_A=\Tr_B\rho_{AB}$. In this paper, we may also apply the partial trace to an arbitrary operator in $\M\br{A\otimes B}$. A pure state evolution on a system $A$ with a state $\ket{\psi}$ is represented by a unitary operator $U^A$, denoted by $U^A\ket{\psi}$.  An evolution on the register $B$ of a state $\ket{\psi}_{AB}$ under the action of a unitary $U^B$  is represented by $\br{\id^A\otimes U^B}\ket{\psi}_{AB}$.  The superscripts and the subscripts might be dropped whenever it is clear from the context. A {\em quantum measurement} is represented by a {\em positive-operator valued measure} (POVM), which is a set of positive semidefinite operators $\set{M_1,\ldots, M_n}$ satisfying $\sum_{i=1}^nM_i=\id$, where $n$ is the number of possible measurement outcomes. Suppose the state of the quantum system is $\rho$. The probability that it outputs $i$ is $\Tr M_i\rho$. In this paper, we use $\vec{M}=\br{M_1,\ldots, M_n}$ to represent an ordered set of operators. We say $\vec{M}=\br{M_1,\ldots, M_n}$ is a POVM if for each i $M_i\geq 0$ is a POVM element with the corresponding output $i$ and $\sum_i M_i=\id$. We say $\vec{M}$ is a {\em sub-POVM} if $M_i\geq 0$ for $1\leq i\leq n$ and  $\sum_{i=1}^n M_i\leq \id$.
	
	In this paper, we need the following fact.
	\begin{fact}\label{fac:cauchyschwartz}
		Given registers $A, B$, operators $P\in\H\br{A}, Q\in\H\br{B}$ and a bipartite state $\psi_{AB}$, it holds that
		\begin{enumerate}
			\item $\Tr\br{\br{P\otimes\id_B}\psi_{AB}}=\Tr P\psi_A$;
			
			\item $\abs{\Tr\br{\br{P\otimes Q}\psi_{AB}}}\leq\br{\Tr P^2\psi_A}^{1/2}\cdot\br{\Tr Q^2\psi_B}^{1/2}$.
		\end{enumerate}
	\end{fact}

	\begin{proof}
	Item 1 follows from definitions.

For item 2, by Cauchy-Schwarz inequality, we have
	\begin{align*}
&\abs{\Tr\br{\br{P\otimes Q}\psi_{AB}}}\\
=&\abs{\Tr\br{\sqrt{\psi_{AB}}\br{P\otimes \id}\br{\id\otimes Q}\sqrt{\psi_{AB}}}}\\
\leq&\br{\Tr\br{P\otimes \id}^2\psi_{AB}}^{1/2}\cdot\br{\Tr \br{\id\otimes Q}^2\psi_{AB}}^{1/2}\\
=&\br{\Tr P^2\psi_A}^{1/2}\cdot\br{\Tr Q^2\psi_B}^{1/2}\quad\quad\mbox{(item 1)}
\end{align*}
	\end{proof}
	
	\subsection{Gaussian spaces}\label{subsec:gaussian}
	
	For $p\geq 1$ and any integer $n>0$, let $\gamma_n$ represent a standard $n$-dimensional normal distribution. A function $f:\reals^n\rightarrow\complex$ is in $L^p\br{\complex,\gamma_n}$ if \[\int_{\reals^n}\abs{f(x)}^p\gamma_n\br{dx}<\infty.\] All the functions considered in this paper are in $L^p\br{\complex,\gamma_n}$ for all $p\geq 1$ unless otherwise stated. We say $f\in L^p\br{\reals,\gamma_n}$ if $f\br{x}\in\reals$ for all $x$. We equip $L^2\br{\complex,\gamma_n}$ with an inner product \[\innerproduct{f}{g}_{\gamma_n}=\expec{\mathbf{x}\sim\gamma_n}{\conjugate{f\br{\mathbf{x}}}g\br{\mathbf{x}}}.\] Given $p\geq 1$ and $f\in L^p\br{\complex,\gamma_n}$, the {\em $p$-norm} of $f$ is defined to be \[\norm{f}_p=\br{\int_{\reals^n}\abs{f(x)}^p\gamma_n\br{dx}}^{\frac{1}{p}}.\] Then $\innerproduct{f}{f}=\twonorm{f}^2$. The set of {\em Hermite polynomials} forms an orthonormal basis in $L^2\br{\complex,\gamma_1}$ with respect to the inner product $\innerproduct{\cdot}{\cdot}_{\gamma_1}$. The Hermite polynomials $H_r:\reals\rightarrow\reals$ for $r\in\mathbb{Z}_{\geq 0}$ are defined as
	\begin{equation}\label{eqn:hermitebasis}
	H_0\br{x}=1; H_1\br{x}=x; H_r\br{x}=\frac{(-1)^r}{\sqrt{r!}}e^{x^2/2}\frac{d^r}{dx^r}e^{-x^2/2}.
	\end{equation}
	For any $\sigma\in\br{\sigma_1,\ldots,\sigma_n}\in\mathbb{Z}_{\geq 0}^n$, define
	$H_{\sigma}:\reals^n\rightarrow\reals$ as \begin{equation}\label{eqn:hermite}
	H_{\sigma}\br{x}=\prod_{i=1}^nH_{\sigma_i}\br{x_i}.
	\end{equation}
	The set $\set{H_{\sigma}:\sigma\in\mathbb{Z}_{\geq 0}^n}$ forms an orthonormal basis in $L^2\br{\complex,\gamma_n}$. Every function $f\in L^2\br{\complex,\gamma_n}$ has an {\em Hermite expansion}  as
$$f\br{x}=\sum_{\sigma\in\mathbb{Z}_{\geq 0}^n}\widehat{f}\br{\sigma}\cdot H_{\sigma}\br{x},$$
	where $\widehat{f}\br{\sigma}$'s are the {\em Hermite coefficients} of $f$, which can be obtained by $\widehat{f}\br{\sigma}=\innerproduct{H_{\sigma}}{f}_{\gamma_n}$. Recall that $\mathrm{wt}\br{\sigma}=\sum_{i=1}^n\sigma_i$ for $\sigma\in\mathbb{Z}_{\geq 0}^n$. The degree of $f$ is defined to be \[\deg\br{f}=\max\set{\mathrm{wt}\br{\sigma}:~\widehat{f}\br{\sigma}\neq 0}.\]

Just as in Fourier analysis, we have {\em Parseval's identity}, that is, \[\twonorm{f}^2=\sum_{\sigma\in\mathbb{Z}_{\geq 0}^n}\abs{\widehat{f}\br{\sigma}}^2.\]

We say $f\in L^2\br{\complex,\gamma_n}$ is {\em multilinear} if $\widehat{f}\br{\sigma}=0$ for $\sigma\notin\set{0,1}^n$.
	
	\begin{definition}\label{def:influencegaussian}
		Given a function $f\in L^2\br{\complex,\gamma_n}$,
		the {\em variance} of $f$ is defined to be
		\begin{equation}\label{eqn:variance}
		\var{f}=\expec{\mathbf{x}\sim \gamma_n}{\abs{f\br{\mathbf{x}}-\expec{}{f}}^2}.
		\end{equation}
		For any set $S\subseteq[n]$, the {\em conditional variance} $\var{f\br{\mathbf{x}}|\mathbf{x}_S}$ is defined to be
		\begin{equation}\label{eqn:conditionalvariance}
		\var{f\br{\mathbf{x}}|\mathbf{x}_S}=\expec{\mathbf{x}\sim \gamma_n}{\abs{f\br{\mathbf{x}}-\expec{}{f\br{\mathbf{x}}|\mathbf{x}_S}}^2\text{\Large$\mid$}\mathbf{x}_S}.
		\end{equation}
		The {\em influence} of the $i$-th coordinate(variable) on $f$, denoted by $\influence_i\br{f}$, is defined by
		\begin{equation}\label{eqn:influencegaussian}
		\influence_i\br{f}=\expec{\mathbf{x}\sim \gamma_n}{\var{f\br{\mathbf{x}}|\mathbf{x}_{-i}}}.
		\end{equation}
		The {\em total influence} of $f$ is defined by \[\influence\br{f}=\sum_i\influence_i\br{f}.\]
	\end{definition}
	The following fact summarizes the basic properties of variance and influence. Readers may refer to~\cite{Odonnell08} for a thorough treatment.
	\begin{fact}\label{fac:influencegaussian}~\cite{Odonnell08,MosselOdonnell:2010}
		Given $f\in L^2\br{\complex,\gamma_n}$, it holds that
		\begin{enumerate}
			\item $\widehat{f}\br{\sigma}\in \reals$ if $f\in L^2\br{\reals,\gamma_n}$;
			\item $\var{f}=\sum_{\sigma\neq 0^n}\abs{\widehat{f}\br{\sigma}}^2\leq\twonorm{f}^2$;
			\item $\influence_i\br{f}=\sum_{\sigma:\sigma_i\neq 0}\abs{\widehat{f}\br{\sigma}}^2$, and hence for all $i$, $\influence_i\br{f}\leq\var{f}$;
			\item $\influence\br{f}=\sum_{\sigma}\abs{\sigma}\abs{\widehat{f}\br{\sigma}}^2$;
			\item $\influence\br{f}\leq\deg\br{f}\var{f}$.
		\end{enumerate}
	\end{fact}
	
	\begin{definition}\label{def:ornstein}
		Given $0\leq\rho\leq 1$ and $f\in L^2\br{\complex,\gamma_n}$, we define the Ornstein-Uhlenbeck operator $U_{\rho}$ to be
		\[U_{\rho}f\br{z}=\expec{\mathbf{x}\sim \gamma_n}{f\br{\rho z+\sqrt{1-\rho^2}\mathbf{x}}}.\]
	\end{definition}
	
	\begin{fact}\label{fac:gaussiannoisy}~\cite[Page 338, Proposition 11.37]{Odonnell08}
		For any $0\leq\rho\leq 1$ and $f\in L^2\br{\complex,\gamma_n}$, it holds that
		\[U_{\rho}f=\sum_{\sigma\in\mathbb{Z}_{\geq0}^n}\widehat{f}\br{\sigma}\rho^{\wt{\sigma}}H_{\sigma}.\]
	\end{fact}
	
	%
	
	Given $p\geq 1$, a {\em vector-valued function} $f=\br{f_1,\ldots,f_k}:\reals^n\rightarrow\complex^k$ is in $L^p\br{\complex^k,\gamma_n}$ if $f_i\in L^p\br{\complex,\gamma_n}$ for all $i$. It is in $L^p\br{\reals^k,\gamma_n}$ if $f_i\in L^p\br{\reals,\gamma_n}$ for all $i$.  For any $f,g\in L^2\br{\complex^k,\gamma_n}$, the inner product of $f$ and $g$ is defined to be \[\innerproduct{f}{g}_{\gamma_n}=\sum_{t=1}^k\innerproduct{f_t}{g_t}_{\gamma_n}.\] The $p$-norm of $f$ is \[\norm{f}_p=\br{\sum_{t=1}^k\norm{f_t}_p^p}^{1/p}.\]
	
	For any $f\in L^2\br{\complex^k,\gamma_n}$, the Hermite coefficients of $f$ are the vectors  \[\widehat{f}\br{\sigma}=\br{\widehat{f_1}\br{\sigma},\ldots,\widehat{f_k}\br{\sigma}}.\] The degree of $f$ is \[\deg\br{f}=\max_t\deg\br{f_t}.\] We say $f$ is multilinear if $f_i$ is multilinear for $1\leq i\leq k$. The variance of $f$ is \[\var{f}=\sum_t\var{f_t}.\] The influence of the $i$-th coordinate (variable) on $f$ is \[\influence_i\br{f}=\sum_t\influence_i\br{f_t}.\] The total influence of $f$ is \[\influence\br{f}=\sum_i\influence_i\br{f}.\]
	The action of Ornstein-Uhlenbeck operator on $f$ is defined to be \[U_{\rho}f=\br{U_{\rho}f_1,\ldots, U_{\rho}f_k}.\]
	
	For any vector $v\in\complex^k$, the norm of $v$ is defined to be \[\norm{v}_2=\sqrt{\sum_{i=1}^k\abs{v_i}^2}.\]
	
\cref{fac:influencegaussian} and \cref{fac:gaussiannoisy} can be directly generalized to vector-valued functions.
	\begin{fact}\label{fac:vecfun}
		Given $f\in L^2\br{\complex^k,\gamma_n}$ and $0\leq\rho\leq 1$, it holds that
		\begin{enumerate}
			\item $\widehat{f}\br{\sigma}\in \reals^k$ if $f\in L^2\br{\reals^k,\gamma_n}$;
			\item $\var{f}=\sum_{\sigma\neq 0^n}\norm{\widehat{f}\br{\sigma}}_2^2\leq\twonorm{f}^2$;
			\item $\influence_i\br{f}=\sum_{\sigma:\sigma_i\neq 0}\norm{\widehat{f}\br{\sigma}}_2^2$, and hence for all $i$, $\influence_i\br{f}\leq\var{f}$;
			\item $\influence\br{f}=\sum_{\sigma}\abs{\sigma}\cdot\twonorm{\widehat{f}\br{\sigma}}^2$;
			\item $\influence\br{f}\leq\deg\br{f}\var{f}$;
			\item $U_\rho f=\sum_{\sigma\in\mathbb{Z}_{\geq 0}^n}\widehat{f}\br{\sigma}\rho^{\wt{\sigma}}H_{\sigma}$.
			
		\end{enumerate}
	\end{fact}

	For any $0\leq \rho\leq 1$, $\G_{\rho}$ represents a  $\rho$-correlated Gaussian distribution, which is a  $2$-dimensional Gaussian distribution \[\br{X,Y}\sim N\br{\begin{pmatrix}
       0 \\
       0
     \end{pmatrix},\begin{pmatrix}
                     1 & \rho \\
                     \rho & 1
                   \end{pmatrix}}.\] Namely, the marginal distributions $X$ and $Y$ are distributed according to $\gamma_1$ and $\expec{}{XY}=\rho$.
	
	\subsection{Matrix spaces and random matrix spaces}\label{subsec:matrixspace}
	For a matrix $M$, with a slight abuse of notations, we use $M_{i,j}$ or $M\br{i,j}$ to represent the $\br{i,j}$-th entry of $M$ whichever is convenient.

For $1\leq p\leq\infty$ the $p$-norm of $M$ is defined to be \[\norm{M}_p=\br{\sum_{i=1}^{\min\set{m,n}}s_i\br{M}^p}^{1/p},\] where $\br{s_1\br{M},s_2\br{M},\ldots}$ are the singular values of $M$ sorted in non-increasing order.  $\norm{M}=\norm{M}_{\infty}=s_1\br{M}$ when $p=\infty$. It is easy to verify that $\norm{M}_p\leq\norm{M}_q$ if $p\geq q$.  For $m=n$, the {\em normalized $p$-norm} of $M$ is defined as
\begin{equation}\label{eqn:nnormdef}
\nnorm{M}_p=\br{\frac{1}{m}\sum_{i=1}^ms_i\br{M}^p}^{1/p}
\end{equation}
 and $\nnorm{M}=\nnorm{M}_{\infty}=s_1\br{M}$. We have $\nnorm{M}_p\geq\nnorm{M}_q$ if $p\geq q$. For any $M\in\H_m$, $\br{\lambda_1\br{M},\ldots,\lambda_m\br{M}}$ represents the eigenvalues of $M$ in non-increasing order and $\abs{M}=\sqrt{M^{\dagger}M}$.

For a Hermitian matrix $X$, suppose it has a spectral decomposition $U^\dagger\Lambda U$, where $U$ is unitary and $\Lambda$ is diagonal. Define
\begin{equation}\label{eqn:geq0def}
\pos{X}= U^{\dagger}\pos{\Lambda} U
\end{equation}
where \pos{\Lambda} is diagonal and $\pos{\Lambda}_{ii}=\Lambda_{ii}$ if $\Lambda_{ii}\geq0$, and $\pos{\Lambda}_{ii}=0$ otherwise. Let $\pinv{X}$ be the {\em Moore-Penrose inverse} of $X$. Namely,
\begin{equation}\label{eqn:penrosedef}
  \pinv{X}= U^{\dagger}\pinv{\Lambda} U
\end{equation}
	where \pinv{\Lambda} is diagonal and $\pinv{\Lambda}_{ii}=\Lambda_{ii}^{-1}$ if $\Lambda_{ii}\ne0$, and $\pinv{\Lambda}_{ii}=0$ otherwise.

Given two matrices $A, B$ of the same dimension, the {\em Hadamard product} of $A$ and $B$ is $A\circ B$, where $\br{A\circ B}_{i,j}= A_{i,j}\cdot B_{i,j}$.
The {\em anticommutator} of $A$ and $B$ is defined to be $\anticommutator{A}{B}=AB+BA$.
	Given $P,Q\in\M_m$, we define
	\begin{equation}\label{eqn:innerproduct}	
	\innerproduct{P}{Q}=\frac{1}{m}\Tr~P^{\dagger}Q.
	\end{equation}
	
	\begin{fact}\label{fac:innerproduct}
		$\innerproduct{\cdot}{\cdot}$ is an inner product. $\br{\innerproduct{\cdot}{\cdot},\M_m}$ forms a Hilbert space of dimension $m^2$. For any $M\in\M_m$, $\nnorm{M}_2^2=\innerproduct{M}{M}$.
	\end{fact}
	
	We say $\set{\B_0,\ldots,\B_{m^2-1}}$ is a {\em standard orthonormal basis} in $\M_m$, if  it is an orthonormal basis with all elements being Hermitian and $\B_0=\id_m$.

	\begin{fact}\label{fac:unitarybasis}
	Given an orthonormal basis $\B=\set{\B_0,\ldots,\B_{m^2-1}}$, the set \linebreak $\B'=\set{\B'_0,\ldots,\B'_{m^2-1}}$ is also an orthonormal basis in $\M_m$ if and only if there exists an $\br{m^2}\times \br{m^2}$ unitary matrix $U$ with both rows and columns indexed by $[m^2]_{\geq 0}$ such that
\[\B'_i=\sum_{j=0}^{m^2-1}U_{i,j}\B_j\] for all
$0\leq i\leq m^2-1$. Moreover, if $\B$ is a standard orthonormal basis, the set $\B'$ is also a standard orthonormal basis in $\M_m$ if and only if $U$ is an orthogonal matrix and it satisfies $U_{0,j}=\delta_{0,j}$ for $0\leq j\leq m^2-1$.
\end{fact}

\begin{fact}\label{fac:paulimutiplecopy}
	Let $\set{\B_i}_{i=0}^{m^2-1}$ be a standard orthonormal basis in $\M_m$, then
	\[\set{\B_{\sigma}=\otimes_{i=1}^n\B_{\sigma_i}}_{\sigma\in[m^2]_{\geq 0}^n}\]
	is a standard orthonormal basis in $\M_m^{\otimes n}$.
\end{fact}

The following lemma guarantees the existence of standard orthonormal bases.
\begin{lemma}\label{lem:paulibasis}
	For any integer $m\geq 2$, there exists a standard orthonormal basis in $\M_m$.
\end{lemma}
\begin{proof}
Consider the space $\br{\H_m,\reals}$ i.e., the space of Hermitian matrices over the real number field, with respect to the inner product $\innerproduct{P}{Q}=\frac{1}{m}\Tr P^{\dagger}Q$. As the inner product of two Hermitian matrices is real, $\br{\H_m,\reals}$ forms a real Hilbert space. Apparently the dimension of $\br{\H_m,\reals}$ is $m^2$ and $\id_m\in\H_m$ with $\innerproduct{\id_m}{\id_m}=1$. Thus, there exists an orthonormal basis $\set{\B_i}_{0\leq i< m^2}$ in $\br{\H_m,\reals}$ with $\B_0=\id$. Note that $\set{\B_i}_{0\leq i< m^2}$ is still an orthonormal set in $\M_m$. Moreover, the dimension of $\br{\M_m,\complex}$, i.e., the space of $m\times m$ matrices over the complex number field, is also $m^2$. Thus, $\set{\B_i}_{0\leq i< m^2}$ forms a standard orthonormal basis in $\M_m$.
\end{proof}

Given a standard orthonormal basis $\B=\set{\B_i}_{i=0}^{m^2-1}$ in $\M_m$, every matrix $M\in\M_m^{\otimes n}$ has a {\em Fourier expansion} with respect to the basis $\B$ given by
\[M=\sum_{\sigma\in[m^2]_{\geq 0}^{n}}\widehat{M}\br{\sigma}\B_{\sigma},\]
where $\widehat{M}\br{\sigma}$'s are the {\em Fourier coefficients} of $M$ with respect to the basis $\B$. They can be obtained as $\widehat{M}\br{\sigma}=\innerproduct{\B_{\sigma}}{M}$. The basic properties of $\widehat{M}\br{\sigma}$'s are summarized in the following fact, which can be easily derived from the orthonormality of $\set{\B_{\sigma}}_{\sigma\in[m^2]_{\geq 0}^n}$.
\begin{fact}\label{fac:basicfourier}
	Given a standard orthonormal basis $\set{\B_i}_{i=0}^{m^2-1}$ in $\M_m$ and $M,N\in\M_m$, it holds that
	\begin{enumerate}
		\item $\widehat{M}\br{\sigma}$'s are real for all $\sigma\in[m^2]_{\geq 0}$ if and only if $M$ is Hermitian;
		\item $\innerproduct{M}{N}=\innerproduct{\id}{M^{\dagger}N}=\innerproduct{MN^{\dagger}}{\id}=\sum_{\sigma}\conjugate{\widehat{M}\br{\sigma}}\widehat{N}\br{\sigma}$;
		\item $\nnorm{M}_2^2=\innerproduct{M}{M}=\innerproduct{M^{\dagger}M}{\id}=\innerproduct{\id}{M^{\dagger}M}=\sum_{\sigma}\abs{\widehat{M}\br{\sigma}}^2$;
		\item $\innerproduct{\id}{M}=\widehat{M}\br{0}$.	
	\end{enumerate}
\end{fact}
The {\em variance} of a matrix $M\in\M_m$ is defined to be  $$\var{M}=\innerproduct{M}{M}-\innerproduct{M}{\id}\innerproduct{\id}{M}.$$
	The following lemma is easily verified.
	\begin{lemma}\label{lem:variance}
		Given a standard orthonormal basis $\set{\B_i}_{i=0}^{m^2-1}$ in $\M_m$ and $M\in\M_m$, it holds that
		\[\var{M}=\sum_{\sigma\neq 0}\abs{\widehat{M}\br{\sigma}}^2.\]
	\end{lemma}

	\begin{definition}
		Let $\B=\set{\B_i}_{i=0}^{m^2-1}$ be a standard orthonormal basis in $\M_m$, $P,Q\in\M_m^{\otimes n}$ and a subset $S\subseteq[n]$.
		\begin{enumerate}
			\item The degree of $P$ is defined to be \[\deg P=\max\set{\abs{\sigma}:\widehat{P}\br{\sigma}\neq 0}.\]
Recall that $\abs{\sigma}$ represents the number of nonzero entries of $\sigma$.
			\item For any $S\subseteq[n]$, \[P_S=\frac{1}{m^{|S^c|}}\Tr_{S^c}P.\]
			\item For any $S\subseteq[n]$,  \[\mathrm{Var}_S[P]=\br{P^{\dagger}P}_{S^c}-\br{P_{S^c}}^{\dagger}\br{P_{S^c}}.\] If $S=\set{i}$, we use $\mathrm{Var}_i[P]$ in short.
			\item For any $i\in[n]$, \[\influence_i\br{P}=\innerproduct{\id}{\mathrm{Var}_i[P]}.\]
			\item \[\influence\br{P}=\sum_i\influence_i\br{P}.\]
		\end{enumerate}
	\end{definition}
	
	With the notion of degrees, we define the low-degree part and the high-degree part of an operator.
	
	\begin{definition}\label{def:lowdegreehighdegree}
		Given integers $m,t>0$, a standard orthonormal basis $\B=\set{\B_i}_{i=0}^{m^2-1}$ in $\M_m$ and $P\in\M_m^{\otimes n}$, we define
		\[P^{\leq t}=\sum_{\sigma\in[m^2]_{\geq 0}^n:\abs{\sigma}\leq t}\widehat{P}\br{\sigma}\B_{\sigma};\]
		\[P^{\geq t}=\sum_{\sigma\in[m^2]_{\geq 0}^n:\abs{\sigma}\geq t}\widehat{P}\br{\sigma}\B_{\sigma}\]
		and
		\[P^{=t}=\sum_{\sigma\in[m^2]_{\geq 0}^n:\abs{\sigma}=t}\widehat{P}\br{\sigma}\B_{\sigma};\]
		where $\widehat{P}\br{\sigma}$'s are the Fourier coefficients of $P$ with respect to the basis $\B$.
	\end{definition}
	
\begin{lemma}\label{lem:pt}
	The degree of $P$ is independent of the choices of bases. Moreover, $P^{\leq t}, P^{\geq t}$ and $P^{=t}$ are also independent of the choices of bases.
\end{lemma}

\begin{proof}
	Let $\set{\B_{\sigma}}_{\sigma\in[m^2]_{\geq 0}}$ and $\set{\B'_{\sigma}}_{\sigma\in[m^2]_{\geq 0}}$ be two standard orthonormal bases in $\M_m$. From \cref{fac:unitarybasis},  there exists an $\br{m^2-1}\times \br{m^2-1}$ orthogonal matrix $U$ satisfying that \[\B_{\sigma}=\sum_{\sigma'=1}^{m^2-1}U_{\sigma,\sigma'}\B'_{\sigma'}\] for any $\sigma\in[m^2-1]$. Suppose $P=P^{=t}$ with respect to the basis $\set{\B_{\sigma}}_{\sigma\in[m^2]_{\geq 0}}$. By linearity, we may assume that $P=\B_{\sigma}=\bigotimes_{i=1}^n\B_{\sigma_i}$  without loss of generality. It is easy to verify that each term in the expansion of $P$ in the basis $\set{\B_i'}_{i=0}^{m^2-1}$ is also of degree $\abs{\sigma}$. The result follows.
\end{proof}

\begin{lemma}\label{lem:partialvariance}
	Given $P\in\M_m^{\otimes n}$, a standard orthonormal basis $\B=\set{\B_i}_{i=0}^{m^2-1}$ in $\M_m$ and a subset $S\subseteq[n]$, it holds that
	\begin{enumerate}
		\item $P_S=\innerproduct{\id}{P}_{S^c}=\sum_{\sigma:\sigma_{S^c}=\mathbf{0}}\widehat{P}\br{\sigma}\B_{\sigma_S}$, $\nnorm{P_S}_2\leq\nnorm{P}_2$;
		\item $\innerproduct{\id_{S^c}}{\mathrm{Var}_S[P]}=\sum_{\sigma:\sigma_S\neq\mathbf{0}}\abs{\widehat{P}\br{\sigma}}^2$;
		\item $\influence_i\br{P}=\sum_{\sigma:\sigma_i\neq0}\abs{\widehat{P}\br{\sigma}}^2$;
		\item $\influence\br{P}=\sum_{\sigma}\abs{\sigma}\abs{\widehat{P}\br{\sigma}}^2\leq\deg P\cdot\nnorm{P}^2_2$.
	\end{enumerate}
\end{lemma}
\begin{proof}
	\begin{enumerate}
		\item For the equality,\[P_S=\frac{1}{m^{|S^c|}}\sum_{\sigma}\widehat{P}\br{\sigma}\Tr_{S^c}\B_{\sigma}=\text{RHS}.\]
		
%
		For the inequality,
		\[\nnorm{P_S}^2_2=\sum_{\sigma:\sigma_{S^c}={\mathbf{0}}}\abs{\widehat{P}\br{\sigma}}^2\leq\sum_{\sigma}\abs{\widehat{P}\br{\sigma}}^2=\nnorm{P}^2_2,\]
		where both equalities are from \cref{fac:basicfourier} item 3.

		\item
		From  item 1,
		\[\br{P^{\dagger}P}_{S^c}=\sum_{\sigma,\sigma'}\conjugate{\widehat{P}\br{\sigma}}\widehat{P}\br{\sigma'}\br{\B_{\sigma}\B_{\sigma'}}_{S^c}=\sum_{\sigma,\sigma':\sigma_S=\sigma'_S}\conjugate{\widehat{P}\br{\sigma}}\widehat{P}\br{\sigma'}\B_{\sigma_{S^c}}\B_{\sigma'_{S^c}}.\]
		Meanwhile,
		\[\br{P_{S^c}}^{\dagger}\br{P_{S^c}}=\sum_{\sigma,\sigma':\sigma_S=\sigma'_S=\mathbf{0}}\conjugate{\widehat{P}\br{\sigma}}\widehat{P}\br{\sigma'}\B_{\sigma_{S^c}}\B_{\sigma'_{S^c}}.\]
		Therefore,
		\begin{eqnarray*}
			\innerproduct{\id_{S^c}}{\mathrm{Var}_S[P]}=\sum_{\sigma:\sigma_S\neq\mathbf{0}}\abs{\widehat{P}\br{\sigma}}^2.
		\end{eqnarray*}
		\item It follows from item 2 and the definition of $\influence_i\br{\cdot}$.
		\item It can be verified by a direct calculation.
	\end{enumerate}
\end{proof}

\begin{definition}\label{def:efronstein}\br{\textbf{Efron-Stein decomposition}} Given integers $n,d>0$, an operator $P\in\M_m^{\otimes n}$, a standard orthonormal basis $\set{\B_i}_{i=0}^{m^2-1}$ and $S\subseteq[n]$, set \[P[S]=\sum_{\sigma\in[m^2]_{\geq 0}^n:\supp{\sigma}=S}\widehat{P}\br{\sigma}\B_{\sigma},\] where $\supp{\sigma}=\set{i\in[n]:\sigma_i>0}$. The Efron-Stein decomposition of $P$ is \[P=\sum_{S\subseteq[n]}P[S].\]
\end{definition}
Again, the definition of $P[S]$ is independent of the choices of the basis $\set{\B_i}_{i=0}^{m^2-1}$, followed by the same argument for \cref{lem:pt}.

The following proposition can be obtained from the orthogonality of $\B_i$'s.
\begin{proposition}\label{prop:enfronsteinortho}
	Given integers $m,n>0$, $S\neq T\subseteq[n]$ and $P,Q\in\M_m^{\otimes n}$, it holds that $\innerproduct{P[S]}{Q[T]}=0$.
\end{proposition}
\begin{proposition}\label{prop:efronstein}
	Given integers $m,n>0$,  $P\in\M_m^{\otimes n}$ and $S,T\subseteq[n], S\not\subseteq T$, it holds that
	\[\Tr_{T^c}~P[S]=0.\]
\end{proposition}
\begin{proof}
	\[\Tr_{T^c}~P[S]=\Tr_{T^c}~\br{\sum_{\sigma:\supp{\sigma}=S}\widehat{P}\br{\sigma}\B_{\sigma}}=0,\]
	where the second equality holds because $S\cap T^c\neq\emptyset$.
\end{proof}

\begin{lemma}\label{lem:jointbasis}
	Given $\psi_{AB}$ with $\psi_A=\frac{\id_{m_A}}{m_A}$ and $\psi_B=\frac{\id_{m_B}}{m_B}$, where $m_A$ and $m_B$  are the dimensions of $A$ and $B$, respectively, there exist standard orthonormal bases $\set{\X_{\alpha}}_{\alpha\in[m_A^2]_{\geq 0}}$ and $\set{\Y_{\beta}}_{\beta\in[m_B^2]_{\geq 0}}$ in  $\M\br{A}$ and $\M\br{B}$, respectively, such that
	\[\Tr\br{\br{\X_{\alpha}\otimes\Y_{\beta}}\psi_{AB}}=\delta_{\alpha,\beta}c_{\alpha}\] for $\alpha\in[m_A^2]_{\geq 0},\beta\in[m_B^2]_{\geq 0}$ and $c_{\alpha}\geq 0$.
\end{lemma}

\begin{proof}
	Let $\set{\A_{\alpha}}_{\alpha\in[m_A^2]_{\geq 0}}$ and $\set{\B_{\beta}}_{\beta\in[m_B^2]_{\geq 0}}$ be arbitrary standard orthonormal bases in $\M_{m_A}$ and $\M_{m_B}$, respectively. Let $\br{M_{\alpha,\beta}}_{\alpha\in[m_A^2]_{\geq 0},\beta\in[m_B^2]_{\geq 0}}$ be an $m_A^2\times m_B^2$ matrix, where
	\[M_{\alpha,\beta}=\Tr\br{\br{\A_{\alpha}\otimes\B_{\beta}}\psi_{AB}}.\] Then $M$ is a real matrix of the form
	\begin{equation}
	M=\begin{pmatrix}
	1 & 0 & \cdots & 0 \\
	0 & & &          &    \\
	\raisebox{15pt}{\vdots}  & & \raisebox{15pt}{{\huge\mbox{{$M'$}}}}  & \\
	0 &  & &
	\end{pmatrix}.
	\end{equation}
	Let $M'=U^{\dagger}DV^{\dagger}$ be a singular value decomposition of $M'$ where $U, V$ are both orthogonal matrices and $D$ is a diagonal matrix.  For any $\alpha\in[m_A^2]_{\geq 0}$ and $\beta\in[m_B^2]_{\geq 0}$ set
	\[\X_{\alpha}=\begin{cases}\sum_{\alpha'=1}^{m_A^2-1}U_{\alpha,\alpha'}\A_{\alpha'}~&\mbox{if $\alpha\neq 0$}\\
	\id_{m_A}~&\mbox{otherwise},\end{cases}
	~\mbox{and}~
	\Y_{\beta}=\begin{cases}\sum_{\beta'=1}^{m_B^2-1}V_{\beta',\beta}\B_{\beta'}~&\mbox{if $\beta\neq 0$}\\
	\id_{m_B}~&\mbox{otherwise}.\end{cases}.\]
	From \cref{fac:unitarybasis}, $\set{\X_{\alpha}}_{\alpha=0}^{m_A^2-1}$ and $\set{\Y_{\beta}}_{\beta=0}^{m_B^2-1}$ are standard orthonormal bases in $\M\br{A}$ and $\M\br{B}$, respectively. Then
	\[\Tr\br{\br{\X_{\alpha}\otimes\Y_{\beta}}\psi_{AB}}=\begin{cases}\br{UM'V}_{\alpha,\beta}=\delta_{\alpha,\beta}D_{\alpha,\alpha}~&\mbox{if $\alpha,\beta>0$,}\\ \delta_{\br{0,0},\br{\alpha,\beta}}~&\mbox{otherwise}.\end{cases}\]

\end{proof}
\subsection{Random operators}\label{subsec:randomoperators}
From the previous subsections, we see that the matrix space $\M_m^{\otimes n}$ and Gaussian space $L^2\br{\complex, \gamma_n}$ are both Hilbert spaces. In this subsection, we unify both spaces by random operators.

	\begin{definition}\label{def:randomoperators}
		Given $p\geq 1$, integers $h, n, m>0$, we say $\mathbf{P}$ is a random operator if it can be expressed as
		\begin{equation}\label{eqn:randomoperatorexpansion}
		\mathbf{P}=\sum_{\sigma\in[m^2]_{\geq 0}^h}p_{\sigma}\br{\mathbf{g}}\B_{\sigma},
		\end{equation}
		where $\set{\B_i}_{i=0}^{m^2-1}$ is a standard orthonormal basis in $\M_m$, $p_{\sigma}:\reals^n\rightarrow\complex$ for all $\sigma\in[m^2]_{\geq 0}^h$ and $\mathbf{g}\sim \gamma_n.$ $\mathbf{P}\in L^p\br{\M_m^{\otimes h},\gamma_n}$ if $p_{\sigma}\in L^p\br{\complex,\gamma_n}$ for all $\sigma\in[m^2]_{\geq 0}^h$. Moreover, $\mathbf{P}\in L^p\br{\H_m^{\otimes n},\gamma_n}$ if $p_{\sigma}\in L^p\br{\reals,\gamma_n}$. Define a vector-valued function \[p=\br{p_{\sigma}}_{\sigma\in[m^2]_{\geq 0}^h}:\reals^n\rightarrow\complex^{m^{2h}}.\] We say $p$ is the {\em associated vector-valued function} of $\mathbf{P}$ under the basis $\set{\B_i}_{i=0}^{m^2-1}$.
	\end{definition}

	The following is a generalization of the $p$-norm in $L^2\br{\M_m^{\otimes h},\gamma_n}$.

	\begin{definition}\label{def:randop}\footnote{To clarify the potential ambiguity, we consider $\nnorm{\mathbf{P}}_p$ to be a random variable and use $N_p\br{\cdot}$ to represent the normalized $p$-norm of a random operator.}\label{def:randoperatorsbasic}
		Given $p\geq 1$, integers $n,h>0, m>1$ and random operator $\mathbf{P}\in L^p\br{\M_m^{\otimes h},\gamma_n}$, the normalized $p$-norm of $\mathbf{P}$ is \[N_p\br{\mathbf{P}}=\br{\expec{}{\nnorm{\mathbf{P}}_p^p}}^{\frac{1}{p}}.\] The degree of $\mathbf{P}$, denoted by $\deg\br{\mathbf{P}}$, is \[\max_{\sigma\in[m^2]_{\geq 0}^h}\deg\br{p_{\sigma}}.\] We say $\mathbf{P}$ is multilinear if $p_{\sigma}\br{\cdot}$ is multilinear for all $\sigma\in[m^2]_{\geq 0}^h$.
	\end{definition}
	
	\begin{lemma}\label{lem:randoperator}
	Given integers $n,h>0, m>1$, let $\mathbf{P}\in L^2\br{\M_m^{\otimes n},\gamma_h}$ with an  associated vector-valued function $p$ under a standard orthonormal basis. It holds that  $N_2\br{\mathbf{P}}=\twonorm{p}.$
\end{lemma}
\begin{proof}
	Consider
	\[N_2\br{\mathbf{P}}^2=\expec{}{\nnorm{\mathbf{P}}_2^2}=\expec{\mathbf{g}\sim \gamma_n}{\sum_{\sigma\in[m^2]_{\geq 0}^h}\abs{p_{\sigma}\br{\mathbf{g}}}^2}=\twonorm{p}^2,\]
	where the second equality follows from \cref{fac:basicfourier} item 3.
\end{proof}
\begin{lemma}\label{lem:influencerandomoperator}
	Given a multilinear random operator $\mathbf{P}\in L^2\br{\M_m^{\otimes h},\gamma_n}$ with degree $d$ and the associated vector-valued function $p$ under a standard orthonormal basis, it holds that
	\[\influence\br{p}\leq \deg\br{\mathbf{P}}N_2\br{\mathbf{P}}^2.\]
\end{lemma}
\begin{proof}
	Consider
	\[\influence\br{p}\leq\deg\br{p}\var{p}\leq\deg\br{p}\twonorm{p}^2=\deg\br{\mathbf{P}}\twonorm{p}^2=\deg\br{\mathbf{P}}N_2\br{\mathbf{P}}^2,\]
	where the first inequality follows from \cref{fac:vecfun} item 5; the second inequality follows from \cref{fac:vecfun} item 2; the first equality follows from the definition of the degree of a random operator; the second equality follows from \cref{lem:randoperator}.
\end{proof}

We say a pair of random operators $\br{\mathbf{P},\mathbf{Q}}\in L^p\br{\M_m^{\otimes h},\gamma_n}\times L^p\br{\M_m^{\otimes h},\gamma_n} $ are {\em joint random operators} if the random variables $\br{\mathbf{g},\mathbf{h}}$ in $\br{\mathbf{P},\mathbf{Q}}$ are drawn from a joint distribution $\G_{\rho}^{\otimes n}$ for $0\leq\rho\leq 1$.
	
\subsection{Miscellaneous}\label{subsec:misc}

Throughout this paper, functions $f:\reals\rightarrow\reals$ are also viewed as maps $f:\H_m\rightarrow\H_m$ defined as
\[f\br{P}=\sum_if\br{\lambda_i}\ketbra{v_i},\] where \[P=\sum_i\lambda_i\ketbra{v_i}\] is a spectral decomposition of $P$.

Given a closed convex set $\Delta\subseteq\reals^k$, we say a map $\R:\reals^k\rightarrow\reals^k$ is a rounding map of $\Delta$ if for any $x\in\reals^k$, $\R\br{x}$ is the element in $\Delta$ that is closest to $x$ in $\twonorm{\cdot}$ distance. The following well-known fact states that the rounding maps of closed convex sets are Lipschitz continuous with Lipschitz constant being $1$.

\begin{fact}\label{fac:rounding}\cite[Page 149, Proposition 3.2.1]{bertsekas2015convex}
	Let $\Delta$ be a nonempty closed convex set in $\reals^k$ with the rounding map $\R$. It holds that
	\[\twonorm{\R\br{x}-\R\br{y}}\leq\twonorm{x-y},\]
	for any $x,y\in\reals^k$.
	
	Thus, if $\Delta$ contains the element $\br{0,\ldots, 0}$, then $\R$ is a {\em contraction}. Namely, $\twonorm{\R\br{x}}\leq\twonorm{x}$ for any $x\in\reals^k$.
\end{fact}

\section{Markov super-operators, noise operators and maximal correlation}\label{sec:markov}
\begin{definition}\label{def:markovoperator}
Given quantum systems $A$ and $B$ and a bipartite state $\psi_{AB}$,  we define a {\em Markov super-operator} $\T:\M\br{B}\rightarrow\M\br{A}$ as follows.
	\[\Tr\br{\br{M^{\dagger}\otimes Q}\psi_{AB}}=\innerproduct{M}{\T\br{Q}},\]
	for any $M\in\M\br{A}$ and $Q\in\M\br{B}$.
\end{definition}

\begin{lemma}\label{lem:markovoperator}
	Given quantum systems $A$ and $B$ and a bipartite state $\psi_{AB}$,  $\T\br{Q}=m_A\Tr_B\br{\br{\id\otimes Q}\psi_{AB}}$.
\end{lemma}

\begin{proof}
	By \cref{def:markovoperator}, $\T\br{Q}$ must satisfy that
	\[m_A\Tr \br{M^{\dagger}\cdot\Tr_B\br{\br{\id\otimes Q}\psi_{AB}}}=\Tr M^{\dagger}\T\br{Q}\] for any $M\in\M\br{A}$. We conclude the result.
\end{proof}

\begin{lemma}\label{lem:markovtensor}
	Given quantum systems $A$ and $B$ and a bipartite state $\psi_{AB}$,  let $\T_B:\M\br{B}\rightarrow\M\br{A}$ and $\T_{B^n}:\M\br{\B^n}\rightarrow\M\br{A^n}$ be the Markov super-operator from $\M\br{B}$ to $\M\br{A}$ and the one from $\M\br{B^n}$ to $\M\br{A^n}$ with the corresponding bipartite states $\psi_{AB}$ and $\psi_{AB}^{\otimes n}$, respectively.  Then $\T_{B^n}=\otimes_{i=1}^n\T_{B}$.
\end{lemma}
\begin{proof}
By the linearity of $\T_B$ and $\T_{B^{\otimes n}}$, it suffices to show $\T_{B^n}(Q)=\otimes_{i=1}^n\T_{B}\br{Q_i}$ when $Q=\otimes_{i=1}^nQ_i$. By \cref{lem:markovoperator}, we have
\begin{align*}
\T\br{Q}&=m_A^n\Tr_{B^n}\br{\br{\id_{A^n}\otimes Q}\psi_{AB}^{\otimes n}}\\
&=\bigotimes_{i=1}^n\br{m_A\Tr_B\br{\br{\id_A\otimes Q_i}\psi_{AB}}}\\
&=\otimes_{i=1}^n\T_{B}\br{Q_i}.
\end{align*}

\end{proof}

\begin{definition}\label{def:bonamibeckner}
	Given a system $A$ of dimension $m$, $\rho\in[0,1]$, a noise operator $\Delta_{\rho}:\M\br{A}\rightarrow\M\br{A}$ on $\M(A)$ is defined as follows. For any $P\in\M\br{A}$,
	\[\Delta_{\rho}\br{P}=\rho P+\frac{1-\rho}{m}\br{\Tr P}\cdot\id_m.\]
	With a slight abuse of notations, the noise operator on  the space $\M\br{A^n}$, again denoted by $\Delta_{\rho}$, is defined as $\Delta_{\rho}=\Delta_{\rho}^{(1)}\cdot\Delta_{\rho}^{(2)}\cdots\Delta_{\rho}^{(n-1)}\cdot\Delta_{\rho}^{(n)}$, where $\Delta_{\rho}^{(i)}$ applies the noise operator $\Delta_{\rho}$ to the $i$'th register and keeps other registers untouched.
\end{definition}
The noise operators $\Delta_{\rho}$ are also called {\em depolarizing channels}~\cite{NC00} in quantum information theory, which are also analogs of the {\em Bonami-Beckner operators} in Fourier analysis~\cite{10.2307/1970980,AIF_1970__20_2_335_0}.

Recall that $\nnorm{\cdot}_p$ represents the normalized $p$-norm defined in \cref{eqn:nnormdef}.

\begin{lemma}\label{lem:bonamibecknerdef}
	Given integers $d,n,m>0$, $\rho\in[0,1]$, a standard orthonormal basis of $\M_m$: $\B=\set{\B_i}_{i=0}^{m^2-1}$, the following holds.
	\begin{enumerate}
		\item For any $P\in\M_m^{\otimes n}$ with a Fourier expansion $P=\sum_{\sigma\in[m^2]_{\geq 0}^n}\widehat{P}\br{\sigma}\B_{\sigma}$, it holds that
		\[\Delta_{\rho}\br{P}=\sum_{\sigma\in[m^2]_{\geq 0}^n}\rho^{\abs{\sigma}}\widehat{P}\br{\sigma}\B_{\sigma}.\]
		\item For any $P\in\M_m^{\otimes n}$, $\nnorm{\Delta_{\rho}\br{P}}_2\leq \nnorm{P}_2$ and $\norm{\Delta_{\rho}\br{P}}\leq\norm{P}$.
		\item If $P\geq0$, then $\Delta_{\rho}\br{P}\geq0$.
	\end{enumerate}
\end{lemma}

\begin{proof}
	Note that $\B_0=\id_m$. Item 1 follows from the definition directly.
	
	For item 2,
%
consider
	\[\nnorm{\Delta_{\rho}\br{P}}_2^2=\sum_{\sigma\in[m^2]_{\geq 0}^n}\rho^{2\abs{\sigma}}\abs{\widehat{P}\br{\sigma}}^2\leq\sum_{\sigma\in[m^2]_{\geq 0}^n}\abs{\widehat{P}\br{\sigma}}^2=\nnorm{P}_2^2.\]
For the second inequality, in item 2,
note that $\Delta_{\rho}=\Delta_{\rho}^{(1)}\cdot\Delta_{\rho}^{(2)}\cdots\Delta_{\rho}^{(n)}$. Here  $\Delta_{\rho}^{(i)}=\id\otimes\ldots\id\otimes\Delta_{\rho}\otimes\id\otimes\ldots\otimes\id$, where $\Delta_{\rho}$ acts on the $i$-th register. Thus by induction it suffices to show $\norm{\Delta_{\rho}^{(i)}(P)}\leq\norm{P}$ for $i\in[n]$. Without loss of generality, we may assume that $i=1$. From~\cref{def:bonamibeckner}, it is easy to see $\Delta_{\rho}^{(1)}(P)=\rho P+\frac{1-\rho}{m}\id_m\otimes\br{\Tr_1 P}$. Thus
\begin{align*}
\norm{\Delta_{\rho}^{(1)}(P)}&\leq\rho\norm{P}+\frac{1-\rho}{m}\norm{\id_m\otimes\br{\Tr_1 P}}\\
&=\rho\norm{P}+\frac{1-\rho}{m}\norm{\Tr_1 P}\\
&=\rho\norm{P}+\frac{1-\rho}{m}\norm{\sum_{i=1}^m\br{\bra{i}\otimes\id_{m^{n-1}}}P\br{\ket{i}\otimes\id_{m^{n-1}}}}\\
&\leq\rho\norm{P}+\frac{1-\rho}{m}\sum_{i=1}^m\norm{\bra{i}\otimes\id_{m^{n-1}}}\cdot\norm{P}\cdot\norm{\ket{i}\otimes\id_{m^{n-1}}}\\
&=\rho\norm{P}+\frac{1-\rho}{m}\cdot m\norm{P}\\
&=\norm{P}
\end{align*}

For item 3, it again suffices to show $\Delta_{\rho}^{(1)}(P)\geq0$ by induction. This is obvious since $\Delta_{\rho}^{(1)}(P)=\rho P+\frac{1-\rho}{m}\id_m\otimes\br{\Tr_1 P}$.
\end{proof}

Quantum maximal correlations introduced by Beigi~\cite{Beigi:2013} are crucial to our analysis. They are a generalization of maximal correlation coefficients~\cite{hirschfeld:1935,Gebelein:1941,Renyi1959} in classical information theory to the quantum setting.
\begin{definition}[Maximal correlation]~\cite{Beigi:2013}\label{def:maximalcorrelation}
	Given quantum systems $A, B$ and a bipartite state $\psi_{AB}$ with $\psi_A=\frac{\id_{m_A}}{m_A}$ and $\psi_B=\frac{\id_{m_B}}{m_B}$, the maximal correlation of $\psi_{AB}$ is defined to be
	\[\rho\br{\psi_{AB}}=\sup\set{\abs{\Tr\br{\br{P^{\dagger}\otimes Q}\psi_{AB}}}~:P\in\M\br{A}, Q\in\M\br{B},\atop\Tr~P=\Tr~Q=0, \nnorm{P}_2=\nnorm{Q}_2=1}.\]
\end{definition}

\begin{fact}~\cite{Beigi:2013}\label{fac:maximalcorrlationone}
	Given quantum systems $A, B$ and a bipartite quantum state $\psi_{AB}$ with $\psi_A=\frac{\id_{m_A}}{m_A}$ and $\psi_B=\frac{\id_{m_B}}{m_B}$, it holds that $\rho\br{\psi_{AB}}\leq 1$.
%
\end{fact}

	\begin{definition}\label{def:noisyepr}
		Given quantum systems $A$ and $B$ with $\dim\br{A}=\dim\br{B}=m$, a bipartite state $\psi_{AB}\in\D\br{A\otimes B}$ is an $m$-dimensional noisy maximally entangled state (MES) if $\psi_A=\psi_B=\frac{\id_m}{m}$ and its maximal correlation $\rho=\rho\br{\psi_{AB}}<1$.
	\end{definition}

	An interesting class of noisy MESs is the states obtained by depolarizing MESs with an arbitrarily small noise.
	\begin{lemma}\label{lem:noisyeprmaximalcorrelation}
		For any $0\leq\epsilon<1$ integer $m>1$, it holds that
		\[\rho\br{\br{1-\epsilon}\ketbra{\Psi}+\epsilon\frac{\id_{m}}{m}\otimes \frac{\id_{m}}{m}}=1-\epsilon,\]
		where $\ket{\Psi}=\frac{1}{\sqrt{m}}\sum_{i=0}^{m-1}|m,m\rangle$ is an $m$-dimensional MES.
	\end{lemma}
	\begin{proof}
		The case that $m=2$ is proved by Beigi in~\cite{Beigi:2013}. His proof can be directly generalized to any $m$. Here we provide a proof for completeness. Let $\Psi_{\epsilon}=\br{1-\epsilon}\ketbra{\Psi}+\epsilon\frac{\id_{m}}{m}\otimes\frac{\id_{m}}{m}$. From the definition of maximal correlations,
		\begin{eqnarray*}
			\rho\br{\Psi_{\epsilon}}=&\max\abs{\Tr\br{\br{X^{\dagger}\otimes Y}\Psi_{\epsilon}}}\\
			&\mathrm{s.t. } \Tr X=\Tr Y=0\\
			&\hspace{1.5 cm} \Tr X^{\dagger}X=\Tr Y^{\dagger}Y=m
		\end{eqnarray*}
		For $X$ and $Y$ satisfying the constraints above, we have
		\begin{eqnarray*}
			&&\abs{\Tr\br{\br{X^{\dagger}\otimes Y}\Psi_{\epsilon}}}\\
			&=&\br{1-\epsilon}\abs{\bra{\Psi}\br{X^{\dagger}\otimes Y}\ket{\Psi}}\\
			&=&\frac{1-\epsilon}{m}\abs{\Tr X^TY^{\dagger}}\\
			&\leq&\frac{1-\epsilon}{m}\twonorm{X}\twonorm{Y}\\
			&=&1-\epsilon,
		\end{eqnarray*}
		where the inequality is achieved by \[X=Y=\sum_{j=0}^{m-1}e^{2\pi\mathrm{i}\cdot j/ m}\ketbra{j}.\]
	\end{proof}

The following proposition provides a useful characterization of quantum maximal correlations.
\begin{proposition}\label{prop:maximalvariance}
	Given quantum systems $A,B$ and a bipartite state $\psi_{AB}$ with $\psi_A=\frac{\id_{m_A}}{m_A}$ and $\psi_B=\frac{\id_{m_B}}{m_B}$, for any $Q\in\M\br{B}$,
	\begin{equation}\label{eqn:TQ}
	\max\set{\abs{\Tr\br{\br{P^{\dagger}\otimes Q}\psi_{AB}}}: P\in\M\br{A}, \nnorm{P}_2=1}
	\end{equation}
	is achieved by
	\[P^*=\frac{\T\br{Q}}{\nnorm{\T\br{Q}}_2},\]
	with the maximum value $\nnorm{\T\br{Q}}_2$, where $\T:\M\br{B}\rightarrow\M\br{A}$ is a Markov super-operator in \cref{def:markovoperator}.
	Thus,
	\begin{equation}\label{eqn:Tmaxcorre}
	\rho\br{\psi_{AB}}=\max\set{\nnorm{\T\br{Q}}_2: Q\in\M\br{B}, \Tr~Q=0,\nnorm{Q}_2=1}.
	\end{equation}
	
	Moreover, the maximal correlation in \cref{def:maximalcorrelation} can be achieved by a pair of Hermitian operators $\br{P,Q}$.
\end{proposition}

\begin{proof}
	The proof complies with the one for Lemma 2.8 in~\cite{Mossel:2010}. Let $P\in\M\br{A}$ achieves the maximum value in Eq.~\cref{eqn:TQ}. Then it satisfies that $\nnorm{P}_2=1$. Write $P=\alpha P^*+\beta P'$, where $\abs{\alpha}^2+\abs{\beta}^2=1$, $\nnorm{P'}_2=1$ and $\innerproduct{P^*}{P'}=0$. By the definition of Markov super-operators
	\[0=\innerproduct{P'}{\T\br{Q}}=\Tr\br{\br{\br{P'}^{\dagger}\otimes Q}\psi_{AB}}.\]
	So we should set $\abs{\alpha}=1.$ Moreover,
	\[\Tr\br{\br{\T\br{Q}^{\dagger}\otimes Q}\psi_{AB}}=\nnorm{\T\br{Q}}^2_2.\]	

	Note that if $\Tr~Q=0$, then by the definition of Markov super-operators,
$$\Tr~\T\br{Q}=m_A\innerproduct{\id}{\T\br{Q}}=m_A\Tr\br{\br{\id\otimes Q}\psi_{AB}}=\frac{m_A}{m_B}\Tr~Q=0$$
where the third equality follows from \cref{fac:cauchyschwartz}.

	By~\cref{lem:markovoperator}, $\T\br{Q}$ is Hermitian if $Q$ is Hermitian. Thus, to prove that the maximal correlation in \cref{def:maximalcorrelation} can be achieved by a pair of Hermitian operators $\br{P,Q}$, it suffices to prove that the maximum in Eq.~\cref{eqn:Tmaxcorre} can be achieved by a Hermitian matrix $Q$. Suppose $Q=Q_1+\mathrm{i}\cdot Q_2$ achieves the maximum value in Eq.~\cref{eqn:TQ} with Hermitian matrices $Q_1\neq 0$ and $Q_2\neq 0$. Then \[\Tr~Q=\Tr~Q_1+\mathrm{i}\cdot \Tr~Q_2=0\]
	implies that  $\Tr~Q_1=\Tr~Q_2=0$ and
	\[1=\nnorm{Q}_2^2=\nnorm{Q_1}_2^2+\nnorm{Q_2}_2^2\]
By~\cref{lem:markovoperator}, both $\T\br{Q_1}$ and $\T\br{Q_2}$ are Hermitian. By the linearity of Markov super-operators, we have
\begin{align*}
\nnorm{\T\br{Q}}_2&=\br{\frac{\nnorm{\T\br{Q_1}}_2^2+\nnorm{\T\br{Q_2}}_2^2}{\nnorm{Q_1}_2^2+\nnorm{Q_2}_2^2}}^{\frac{1}{2}}\\&\leq\max\set{\nnorm{\T\br{\frac{Q_1}{\nnorm{Q_1}_2}}}_2,\nnorm{\T\br{\frac{Q_2}{\nnorm{Q_2}_2}}}_2}.
\end{align*}
	Thus at least one of $\frac{Q_1}{\nnorm{Q_1}_2}$ and $\frac{Q_2}{\nnorm{Q_2}_2}$ also achieves the maximum in Eq.~\cref{eqn:TQ}

\end{proof}

\begin{lemma}\label{lem:efronsteinortho}
	Given quantum systems $A, B$ with $\dim A=m_A$ and $\dim B= m_B$, a bipartite quantum state $\psi_{AB}$ with $\psi_A=\frac{\id_{m_A}}{m_A}$ and $\psi_B=\frac{\id_{m_B}}{m_B}$, let $\set{\A_{\sigma}}_{\sigma\in[m_A^2]_{\geq 0}}$ and $\set{\B_{\sigma}}_{\sigma\in[m_B^2]_{\geq 0}}$ be standard orthonormal bases in $\M\br{A}$ and $\M\br{B}$, respectively. It holds that
	\[\Tr\br{\br{\A_{\sigma}\otimes\B_{\tau}}\psi_{AB}^{\otimes n}}=0,\]
	whenever $\supp{\sigma}\neq\supp{\tau}$. Thus
	\[\Tr\br{\br{A[S]\otimes B[T]}\psi_{AB}^{\otimes n}}=0,\]
	whenever $S\neq T$, where $A[S]$ and $B[T]$ are obtained from the Efron-Stein decompositions of $A$ and $B$ in Definition~\ref{def:efronstein}.
\end{lemma}
\begin{proof}
	For any $\sigma\in[m_A^2]_{\geq 0}, \tau\in[m_B^2]_{\geq 0}$ and  $\sigma\neq 0$ and $\tau\neq0$, it holds that
	\[\Tr\br{\br{\A_{\sigma}\otimes\B_0}\psi_{AB}}=\frac{1}{m_A}\Tr~\A_{\sigma}=\innerproduct{\A_0}{\A_{\sigma}}=0.\]
	Symmetrically,
	$\Tr\br{\br{\A_0\otimes\B_{\tau}}\psi_{AB}}=0$.
	Thus, for any $\sigma\in[m_A^2]_{\geq 0}^n$ and $\tau\in[m_B^2]_{\geq 0}^n$ with $\supp{\sigma}\neq\supp{\tau}$, it holds that
	\[\Tr\br{\br{\A_{\sigma}\otimes\B_{\tau}}\psi_{AB}^{\otimes n}}=\prod_{i=1}^n\Tr\br{\br{\A_{\sigma_i}\otimes\B_{\tau_i}}\psi_{AB}}=0.\]	
\end{proof}

\begin{proposition}\label{prop:markovenfronstein}
	Given integers $n, m_A, m_B>0$, quantum systems $A$ and $B$, a bipartite quantum state $\psi_{AB}$ with $\psi_A=\frac{\id_{m_A}}{m_A}$ and $\psi_B=\frac{\id_{m_B}}{m_B}$, $Q\in\M\br{B^n}$ and $S\subseteq[n]$, it holds that
	\[\T\br{Q[S]}=\T\br{Q}[S],\]
	where $\T:\M\br{B^n}\rightarrow\M\br{A^n}$ is a Markov super-operator with respect to $\psi_{AB}^{\otimes n}$ defined in \cref{def:markovoperator}.
\end{proposition}
\begin{proof} It suffices to show that
	\[\innerproduct{X}{\T\br{Q[S]}}=\innerproduct{X}{\T\br{Q}[S]},\] for any $X\in\M\br{A^n}$.
	By the definition,
	\[\text{LHS}=\Tr\br{\br{X^{\dagger}\otimes Q[S]}\psi_{AB}^{\otimes n}}.\]
	By \cref{def:efronstein} and \cref{prop:enfronsteinortho}, we have
	\[\text{RHS}=\innerproduct{X[S]}{\T\br{Q}}=\Tr\br{\br{X[S]^{\dagger}\otimes Q}\psi_{AB}^{\otimes n}}.\]
	By \cref{lem:efronsteinortho} and the Efron-Stein decompositions of $X$ and $Q$ defined in \cref{def:efronstein},
	\[\Tr\br{\br{X[S]^{\dagger}\otimes Q}\psi_{AB}^{\otimes n}}=\Tr\br{\br{X^{\dagger}\otimes Q[S]}\psi_{AB}^{\otimes n}}=\Tr\br{\br{X[S]^{\dagger}\otimes Q[S]}\psi_{AB}^{\otimes n}}.\]
\end{proof}

\begin{proposition}\label{prop:markovoperatornorm}
	Given integers $n, m_A, m_B>0$, quantum systems $A, B$, a bipartite quantum state $\psi_{AB}$ with $\psi_A=\frac{\id_{m_A}}{m_A}$ and $\psi_B=\frac{\id_{m_B}}{m_B}$, $Q\in\M\br{B^n}$ and $S\subseteq[n]$, it holds that
	\[\nnorm{\T\br{Q[S]}}_2\leq\rho^{|S|}\nnorm{Q[S]}_2,\]
	where $\rho=\rho\br{\psi_{AB}}$.
\end{proposition}
\begin{proof}
	We may assume that $Q=Q[S]$ without loss of generality. It suffices to show the case that $S=[n]$. Let $\set{\A_i}_{i\in[m_A^2]_{\geq 0}}$ and $\set{\B_i}_{i\in[m_B^2]_{\geq 0}}$ be standard orthonormal bases in $\M\br{A}$ and $\M\br{B}$, respectively. We use a hybrid argument over $n$ registers. For $r\in[n]$, set $\T^{\br{r}}= \id_{\M\br{A}}^{\otimes\br{r-1}}\otimes\T\otimes\id_{\M\br{B}}^{\otimes\br{n-r}}$, $Q^{(0)}= Q$ and $Q^{\br{r}}=\T^{(r)}\br{Q^{(r-1)}}$, where $\id_{\M\br{A}}$ and $\id_{\M\br{B}}$ are the identity maps mapping $\M\br{A}$ to $\M\br{A}$ and $\M\br{B}$ to $\M\br{B}$, respectively. By \cref{lem:markovtensor}, we have $Q^{(n)}=\T\br{Q}$. By induction it suffices to show that
	\begin{equation}\label{eqn:markovnormeq1}
	\nnorm{Q^{(r)}}_2^2\leq\rho^2 \nnorm{Q^{(r-1)}}_2^2,
	\end{equation}
	and
	\begin{equation}\label{eqn:markovnormeq2}
	Q^{(r)}=Q^{(r)}[[n]],
	\end{equation}
	where $Q^{(r)}[[n]]$ is defined by expanding $Q^{(r)}$ over \[\set{\A_{\sigma_{\leq r}}\otimes\B_{\sigma_{> r}}}_{\sigma\in[m_A^2]_{\geq 0}^r\times[m_B^2]_{\geq 0}^{n-r}},\]
	because $\T=\T^{(n)}\circ\cdots\circ\T^{\br{1}}$.
	Let $\set{\ket{u_i}}_{i\in[m_A]}$ and $\set{\ket{v_i}}_{i\in[m_B]}$ be orthonormal bases of $\complex^{m_A}$ and $\complex^{m_B}$, respectively. For any $s\in[m_A]^{r-1}\times[m_B]^{n-r}$, we define
	\[\ket{w_s}=\ket{u_{s_1}}\otimes\ldots\otimes\ket{u_{s_{r-1}}}\otimes\ket{v_{s_{r+1}}}\otimes\ldots\otimes\ket{v_{s_n}}.\]
	For any $s, t\in[m_A]^{r-1}\times[m_B]^{n-r}$, we define
	\[P^{(r)}_{s,t}=\br{\bra{w_s}\otimes\id}Q^{(r)}\br{\ket{w_t}\otimes\id} ~\mbox{and}~ Q^{(r-1)}_{s,t}=\br{\bra{w_s}\otimes\id}Q^{(r-1)}\br{\ket{w_t}\otimes\id},\]
	where $\ket{w_s}$ and $\ket{w_t}$ lie in the registers $\set{1,\ldots, r-1,r+1,\ldots, n}$. Note that $\T^{\br{r}}$ applies $\T$ to the $r$-th register and leaves other registers unchanged.  Then  $Q^{\br{r}}=\T^{\br{r}}\br{Q^{\br{r-1}}}$ implies that
	\begin{equation}\label{eqn:tpq}
	P^{(r)}_{s,t}=\T\br{Q_{s,t}^{(r-1)}}.
	\end{equation}
	And we also have
	\begin{eqnarray*}
		&&\Tr~Q_{s,t}^{(r-1)}=0,
	\end{eqnarray*}
	by the induction $Q^{(r-1)}=Q^{(r-1)}[n]$.
	By \cref{prop:maximalvariance},  $\nnorm{P^{(r)}_{s,t}}_2\leq\rho\nnorm{Q^{(r-1)}_{s,t}}_2$.
	Consider
	\begin{eqnarray*}
		&&\sum_{s,t}\norm{P^{(r)}_{s,t}}_2^2\\
		&=&\sum_{s,t}\Tr~\br{P^{(r)}_{s,t}}^{\dagger}P^{(r)}_{s,t}\\
		&=&\sum_{s,t}\Tr~\br{\bra{w_t}\otimes \id}\br{Q^{(r)}}^{\dagger}\br{\ket{w_s}\otimes\id}\br{\bra{w_s}\otimes\id}Q^{\br{r}}\br{\ket{w_t}\otimes\id}\\
		&=&\sum_{t}\Tr~\br{\bra{w_t}\otimes \id}\br{Q^{(r)}}^{\dagger}Q^{\br{r}}\br{\ket{w_t}\otimes\id}\hspace{4mm}(\mbox{because $\sum_s\ketbra{w_s}=\id$})\\
		&=&\Tr~\br{Q^{(r)}}^{\dagger}Q^{(r)}	
	\end{eqnarray*}
	Similarly,
	\[\sum_{s,t}\norm{Q^{(r-1)}_{s,t}}_2^2=\Tr~\br{Q^{(r-1)}}^{\dagger}Q^{(r-1)}.\]
	Combining with Eq.~\cref{eqn:tpq} and \cref{prop:maximalvariance}, we conclude Eq.~\cref{eqn:markovnormeq1}.
	
	For Eq.~\cref{eqn:markovnormeq2}, compute
	\begin{eqnarray*}
		&&\widehat{Q^{(r)}}\br{\sigma}=\widehat{\T^{\br{r}}\br{Q^{(r-1)}}}\br{\sigma}=\innerproduct{\A_{\sigma_{\leq r}}\otimes\B_{\sigma_{>r}}}{\T^{\br{r}}\br{Q^{(r-1)}}}\\
		&=&\sum_{\tau:\abs{\tau}=n}\widehat{Q^{(r-1)}}\br{\tau}\innerproduct{\A_{\sigma_{\leq r}}\otimes\B_{\sigma_{>r}}}{\T^{\br{r}}\br{\A_{\tau_{<r}}\otimes\B_{\tau_{\geq r}}}}\\
		&=&\sum_{\tau:\abs{\tau}=n,\tau_{-r}=\sigma_{-r}}\widehat{Q^{(r-1)}}\br{\tau}\innerproduct{\A_{\sigma_r}}{\T\br{\B_{\tau_r}}}.
	\end{eqnarray*}
	Note that
	\[\innerproduct{\A_0}{\T\br{\B_{\tau_r}}}=\Tr\br{\br{\id\otimes\B_{\tau_r}}\psi_{AB}}=0,\]
	as $\abs{\tau}=n$. Therefore, $\widehat{\T^{\br{r}}\br{Q}}\br{\sigma}=0$ if $\abs{\sigma}<n$. We conclude Eq.~\cref{eqn:markovnormeq2}.
\end{proof}
A useful property of (classical) maximal correlation coefficients is {\em tensorization}, which states that the maximal correlation of multiple independent identical copies of a distribution is the same as the one of one copy. The same property also holds for the quantum maximal correlation shown by Beigi~\cite{Beigi:2013}. Here we provide a different proof.
\begin{fact}~\cite{Beigi:2013}\label{prop:maximalcorrelationtensorisation}Given quantum systems $A, B$ with $\dim A=m_A$ and $\dim B= m_B$, a bipartite quantum state $\psi_{AB}$ with $\psi_A=\frac{\id_{m_A}}{m_A}$ and $\psi_B=\frac{\id_{m_B}}{m_B}$,  it holds that
	\[\rho\br{\psi_{AB}^{\otimes n}}=\rho\br{\psi_{AB}}.\]
\end{fact}
\begin{proof}
	Given $Q\in\M\br{\B^n}$ with $\Tr~Q=0$ and $\nnorm{Q}_2=1$, we use the Efron-Stein decomposition $Q=\sum_{S\neq\emptyset}Q[S]$ to obtain
	\begin{eqnarray*}
		&&\nnorm{\T\br{Q}}_2^2\\
		&=&\nnorm{\sum_{S\neq\emptyset}\T\br{Q[S]}}_2^2\quad\quad\mbox{(linearity of $\T$)}\\
		&=&\nnorm{\sum_{S\neq\emptyset}\T\br{Q}[S]}_2^2\quad\quad\mbox{(\cref{prop:markovenfronstein})}\\
		&=&\sum_{S\neq\emptyset}\nnorm{\T\br{Q}[S]}_2^2\quad\quad\mbox{(orthogonality of the Efron-Stein decomposition)}\\
		&\leq&\sum_{S\neq\emptyset}\rho^{2\abs{S}}\nnorm{Q[S]}_2^2\quad\quad\mbox{(\cref{prop:markovoperatornorm})}\\
		&\leq&\rho^2\sum_{S\neq\emptyset}\nnorm{Q[S]}_2^2\\
		&=&\rho^2\nnorm{Q}_2^2\quad\quad\mbox{(orthogonality of the Efron-Stein decomposition)}\\
		&=&\rho^2.
	\end{eqnarray*}
	From \cref{prop:maximalvariance}, $\rho\br{\psi^{\otimes n}_{AB}}\leq\rho\br{\psi_{AB}}$. The other direction trivially follows by definition.
\end{proof}

	\section{Main results}\label{sec:mainresult}
	\begin{theorem}\label{thm:nijs}
		Let $\epsilon\in(0,1)$, integers $n,m,a,b,t>0$, $\psi_{AB}$ be a noisy MES. Namely, $\psi_A=\psi_B=\frac{\id_m}{m}$. Let $\rho=\rho\br{\psi_{AB}}<1$ be the maximal correlation of $\psi_{AB}$ defined in \cref{def:maximalcorrelation}. Then there exists an explicitly computable $D=D\br{\rho,\epsilon,a,b,m,t}$, such that for any sequences of POVMs $$\vec{P_1},\ldots,\vec{P_a}~\mbox{and}~\vec{Q_1},\ldots,\vec{Q_b},$$ where

$$\vec{P_u}=\br{P_{u,1},\ldots, P_{u,t}}~\mbox{and}~\vec{Q_v}=\br{Q_{v,1},\ldots, Q_{v,t}}$$ and $P_{u,i},Q_{v,j}\in\H_m^{\otimes n}$ for $u\in[a],v\in[b],i,j\in[t]$,  there exist sequences of positive semidefinite operators $$\vec{\widetilde{P}_1},\ldots,\vec{\widetilde{P}_a}~\mbox{and}~\vec{\widetilde{Q}_1},\ldots,\vec{\widetilde{Q}_b},$$ where
$$\vec{\widetilde{P}_u}=\br{\widetilde{P}_{u,1}\ldots,\widetilde{P}_{u,t}}~\mbox{and}~\vec{\widetilde{Q}_v}=\br{\widetilde{Q}_{v,1},\ldots,\widetilde{Q}_{v,t}}$$ for $\widetilde{P}_{u,i},\widetilde{Q}_{v,j}\in\H_m^{\otimes D}$, $u\in[a],v\in[b],i,j\in[t]$, such that the following holds.
		
		\begin{enumerate}
			\item For $u\in[a],v\in[b]$, $\vec{\widetilde{P}_u}, \vec{\widetilde{Q}_v}$ are sub-POVMs. Namely,
$$\sum_{j=1}^t\widetilde{P}_{u,j}\leq\id, \sum_{j=1}^t\widetilde{Q}_{v,j}\leq\id~\mbox{and}~\widetilde{P}_{u,j}\geq0, \widetilde{Q}_{v,j}\geq0~\mbox{for}~j\in[t].$$
			\item For any $u\in[a],v\in[b],i,j\in[t]$, $$\abs{\Tr\br{\br{P_{u,i}\otimes Q_{v,j}}\psi_{AB}^{\otimes n}}-\Tr\br{\br{\widetilde{P}_{u,i}\otimes \widetilde{Q}_{v,j}}\psi_{AB}^{\otimes D}}}\leq\epsilon.$$
			
		\end{enumerate}
		In particular, one may choose $$D=\exp\br{\mathrm{poly}\br{a,b,t,\exp\br{\mathrm{poly}\br{\frac{1}{\epsilon},\frac{1}{1-\rho},\log m}}}}.$$
	\end{theorem}
\begin{remark}
Here we assume for simplicity that the sizes of the output sets of the measurements on both sides are same. However, the same argument also holds when they are different by setting $t$ to be the larger one.
\end{remark}
	The proof is deferred to the end of this section. The following is an application of \cref{thm:nijs} to the decidability of nonlocal games.
	
	\begin{theorem}\label{thm:decidable}
		Given parameters $0<\epsilon,\rho<1$, an integer $m\geq 2$, a noisy MES state $\psi_{AB}$, i.e.,  $\psi_A=\psi_B=\frac{\id_m}{m}$ with the maximal correlation $\rho=\rho\br{\psi_{AB}}<1$
		as defined in \cref{def:maximalcorrelation}, let $G$  be a nonlocal game with the question sets $\X,\Y$ and the answer sets $\A,\B$. Suppose the players share arbitrarily many copies of $\psi_{AB}$. Let $\omega^*(G,\psi_{AB})$ be the supremum of the winning probability that the players can achieve. Then there exists an explicitly computable bound $D=D\br{\abs{\X},\abs{\Y},\abs{\A},\abs{\B},m,\epsilon,\rho}$ such that it suffices for the players to share $D$ copies of $\psi_{AB}$ to achieve the winning probability at least $\omega^*(G,\psi_{AB})-\epsilon$. In particular, one may choose
		$$D=\exp\br{\mathrm{poly}\br{\abs{\X},\abs{\Y}, \exp\br{\mathrm{poly}\br{\abs{\A},\abs{\B},\frac{1}{\epsilon},\frac{1}{1-\rho}},\log m}}}.$$
	\end{theorem}
	\begin{proof}
		We assume that $\abs{\A}=\abs{B}=t$ without loss of generality. Suppose the players share $n$ copies of $\psi_{AB}$ and employ the strategies $$\br{\set{\vec{P_x}=\br{P_{x,1},\ldots, P_{x,t}}}_{x\in\X},\set{\vec{Q_y}=\br{Q_{y,1},\ldots, Q_{y,t}}}_{y\in\Y}}$$ with the winning probability $\omega$. For simplicity, we set $\X=\set{1,2,\ldots,\abs{\X}}$ and $\Y=\set{1,2,\ldots,\abs{\Y}}$.
		
		We apply \cref{thm:nijs} to the following two sequences of POVMs
		$$\br{\vec{P_1},\vec{P_2},\ldots\vec{P_{\abs{\X}}}}~\mbox{and}~\br{\vec{Q_1},\vec{Q_2},\ldots, \vec{Q_{\abs{\Y}}}}$$
		with parameters $\epsilon\leftarrow\epsilon/t^2, t\leftarrow t$. Let $\vec{\widetilde{P}_1},\ldots,\vec{\widetilde{P}_{|\X|}},\vec{\widetilde{Q}_1},\ldots,\vec{\widetilde{Q}_{|Y|}}$ be the sequences sub-POVMs induced by \cref{thm:nijs}.
We claim that the strategy $$\br{\set{\vec{\widetilde{P_{x}}}}_{x\in\X},\set{\vec{\widetilde{Q_{y}}}}_{y\in\Y}}$$ wins the game with probability $\widetilde{\omega}\geq \omega-\epsilon$. Moreover, they can be easily converted to valid POVMs without lowering the winning probability.
		
		Let
$$\vec{\widetilde{P}_x}=\set{\widetilde{P}_{x,1},\ldots,\widetilde{P}_{x,t}}, \vec{\widetilde{Q}_y}=\set{\widetilde{Q}_{y,1},\ldots,\widetilde{Q}_{y,t}}$$ and $$\nu_{xy}\br{a,b}=\Tr\br{\br{P_{x,a}\otimes Q_{y,b}}\psi_{AB}^{\otimes n}}$$
and $$\widetilde{\nu}_{xy}\br{a,b}=\Tr\br{\br{\widetilde{P}_{x,a}\otimes \widetilde{Q}_{y,b}}\psi_{AB}^{\otimes D}}.$$
From \cref{thm:nijs}, for any $(x,y,a,b)\in\X\times\Y\times[t]\times[t]$,
		\begin{eqnarray*}
			\abs{\nu_{xy}\br{a,b}-\widetilde{\nu}_{xy}\br{a,b}}\leq\epsilon/t^2
		\end{eqnarray*}
		Thus
		\begin{eqnarray*}
			&&\abs{\omega-\tilde{\omega}}=\abs{\sum_{xy}\mu\br{x,y}\br{\nu_{xy}\br{a,b}-\widetilde{\nu}_{xy}\br{a,b}}V(x,y,a,b)} \\&\leq&\sum_{xy}\mu\br{x,y}\sum_{a,b}\abs{\nu_{xy}\br{a,b}-\widetilde{\nu}_{xy}\br{a,b}}\\
			&\leq&\epsilon.
		\end{eqnarray*}
		
	\end{proof}

\begin{remark}
\cref{thm:nijs} and \cref{thm:decidable} can be directly generalized to the case when $\psi_A=\frac{\id_{m_A}}{m_A}$, $\psi_B=\frac{\id_{m_B}}{m_B}$ and $m_A\neq m_B$.
\end{remark}

To prove \cref{thm:nijs}, we need the following intermediate lemma for technical reasons.

\begin{lemma}\label{lem:nijslem}
With the same setting of \cref{thm:nijs}, the following holds.
  	\begin{enumerate}
			\item For $1\leq u\leq s$, $\vec{\widetilde{P}_u}, \vec{\widetilde{Q}_u}$ are sub-POVMs. Namely,
$$\sum_{j=1}^t\widetilde{P}_{u,j}\leq\id, \sum_{j=1}^t\widetilde{Q}_{u,j}\leq\id~\mbox{and}~\widetilde{P}_{u,j}\geq0, \widetilde{Q}_{u,j}\geq0 ~\mbox{for}~ j\in[t].$$
			\item For any $u\in[s],i,j\in[t]$, $$\abs{\Tr\br{\br{P_{u,i}\otimes Q_{u,j}}\psi_{AB}^{\otimes n}}-\Tr\br{\br{\widetilde{P}_{u,i}\otimes \widetilde{Q}_{u,j}}\psi_{AB}^{\otimes D}}}\leq\epsilon.$$
			
			\item If $\vec{P_u}=\vec{P_v}$, then $\vec{\widetilde{P}_{u}}=\vec{\widetilde{P}_{v}}$. If $\vec{Q_u}=\vec{Q_v}$, then $\vec{\widetilde{Q}_u}=\vec{\widetilde{Q}_v}$.
		\end{enumerate}

\end{lemma}
\begin{proof}[Proof of \cref{thm:nijs}]
   Applying Lemma~\ref{lem:nijslem} to the following two sequences of POVMs
		$$\br{\underbrace{\vec{P_1},\ldots, \vec{P_1}}_{b~\text{times}},\underbrace{\vec{P_2},\ldots, \vec{P_2}}_{b~\text{times}},\ldots, \underbrace{\vec{P_a},\ldots, \vec{P_a}}_{b~\text{times}}}$$
		and
		$$\br{\vec{Q_1},\ldots, \vec{Q_b},\vec{Q_1},\ldots, \vec{Q_b},\ldots,\vec{Q_1},\ldots, \vec{Q_b}}$$
with parameters $n\leftarrow n,\rho\leftarrow \rho,\epsilon\leftarrow \epsilon,a\leftarrow ab,b\leftarrow ab,t\leftarrow t$, we obtain
$$\br{\underbrace{\vec{\widetilde{P_1}},\ldots, \vec{\widetilde{P_1}}}_{b~\text{times}},\underbrace{\vec{\widetilde{P_2}},\ldots, \vec{\widetilde{P_2}}}_{b~\text{times}},\ldots, \underbrace{\vec{\widetilde{P_a}},\ldots, \widetilde{\vec{P_a}}}_{b~\text{times}}}$$
		and
		$$\br{\vec{\widetilde{Q_1}},\ldots, \vec{\widetilde{Q_b}},\vec{\widetilde{Q_1}},\ldots, \vec{\widetilde{Q_b}},\ldots,\vec{\widetilde{Q_1}},\ldots, \vec{\widetilde{Q_b}}}$$
guaranteed by Item 3 of \cref{lem:nijslem}. Item 1 and Item 2 in \cref{thm:nijs} follow by Item 1 and Item 2 in \cref{lem:nijslem}.
\end{proof}
It remains to show \cref{lem:nijslem}.	
	
\begin{proof}[Proof of \cref{lem:nijslem}]
		Let $\delta,\tau $ be parameters which are chosen later. The proof is composed of several steps.
		\begin{itemize}
			\item \textbf{Smooth operators}. For $u\in[s],i\in[t]$, we apply the map $f$ implied by \cref{lem:smoothing of strategies} to $\set{P_{u,i}}_{u\in[s],i\in[t]}$ and $\set{Q_{u,i}}_{u\in[s],i\in[t]}$ to get $\set{P_{u,i}^{(1)}}_{u\in[s],i\in[t]}$ and $\set{Q_{u,i}^{(1)}}_{u\in[s],i\in[t]}$, respectively, and $ d_1=\frac{2\log^2(1/\delta)}{C\br{1-\rho}\delta}$ for some constant $C$ satisfying the following.
			\bigskip
			\begin{enumerate}
				\item $\set{P_{u,i}^{(1)}}_{1\leq i\leq t}$ and $\set{Q_{u,i}^{(1)}}_{1\leq i\leq t}$ are POVMs for $1\leq u\leq s$, which is implied by item 5 in \cref{lem:smoothing of strategies}.
				\item For any $i\in[t],u\in[s]$,$$\nnorm{P_{u,i}^{(1)}}_2\leq\nnorm{P_{u,i}}_2~\mbox{and}~\nnorm{Q_{u,i}^{(1)}}_2\leq\nnorm{Q_{u,i}}_2.$$
				\item For any $i,j\in[t],u\in[s]$,
				$$\abs{\Tr\br{\br{P_{u,i}^{(1)}\otimes\ Q^{(1)}_{u,j}}\psi^{\otimes n}_{AB}}-\Tr\br{\br{P_{u,i}\otimes Q_{u,j}}\psi^{\otimes n}_{AB}}}\leq\delta.$$
				\item For any $i\in[t],u\in[s]$,$$\nnorm{\br{P_{u,i}^{(1)}}^{>d_1}}^2_2\leq\delta~\mbox{and}~\nnorm{\br{Q_{u,i}^{(1)}}^{>d_1}}^2_2\leq\delta.$$
				\item If $P_{u,i}=P_{v,i}$, then $P_{u,i}^{(1)}=P_{v,i}^{(1)}$. If $Q_{u,i}=Q_{v,i}$, then $Q_{u,i}^{(1)}=Q_{v,i}^{(1)}$.
			\end{enumerate}
			\bigskip
			
%
%
%
%
			\item\textbf{Regularization}. For any $u\in[s]$ applying \cref{lem:regular} to $P_{u,i}^{(1)}$ and $Q_{u,i}^{(1)}$ with $\delta\leftarrow\delta, \epsilon\leftarrow\tau, d\leftarrow d_1$, we obtain sets $H_{u,i}\subseteq[n]$ of size $\abs{H_{u,i}}\leq\frac{2d_1}{\tau}$ such that for all $i\in[t]$
			
			\[\br{\forall r\notin H_{u,i}}~\influence_r\br{\br{P_{u,i}^{(1)}}^{\leq d_1}}\leq\tau,~\mbox{and}~\influence_r\br{\br{Q_{u,i}^{(1)}}^{\leq d_1}}\leq\tau.\]
			Set $H=\bigcup_{u\in[s],i\in[t]}H_{u,i}$. Then $h=\abs{H}\leq\frac{2std_1}{\tau}$. Without loss of generality, assume $H=[h]$. It holds that for any $u\in[s]$
			
			\begin{enumerate}
				\item $\br{\forall i\in[t],r\notin H}~\influence_r\br{\br{P_{u,i}^{(1)}}^{\leq d_1}}\leq\tau,~\mbox{and}~\influence_r\br{\br{Q_{u,i}^{(1)}}^{\leq d_1}}\leq\tau;$
				\item $\forall i,j\in[t]$, we have
$$\Tr\br{\br{P_{u,i}^{(1)}\otimes Q_{u,j}^{(1)}}}\psi_{AB}^{\otimes n}=\sum_{\sigma\in[m^2]_{\geq 0}^h}c_{\sigma}\Tr\br{\br{P_{u,i,\sigma}^{(1)}\otimes Q_{u,j,\sigma}^{(1)}}\psi_{AB}^{\otimes (n-h)}},$$

				where
				$\br{c_i}_{i=0}^{m^2-1}$ are the singular values of the matrix $\mathsf{Corr}\br{\psi_{AB}}$ in non-increasing order defined in \cref{def:covariancematrix}, $c_{\sigma}=c_{\sigma_1}\cdot c_{\sigma_2}\cdots c_{\sigma_h}$ and
				\[P_{u,i,\sigma}^{(1)}=\sum_{\tau\in[m^2]_{\geq 0}^n:\tau_H=\sigma}\widehat{P_{u,i}^{(1)}}\br{\tau}\A_{\tau_{H^c}}\]
\[Q_{u,j,\sigma}^{(1)}=\sum_{\tau\in[m^2]_{\geq 0}^n:\tau_H=\sigma}\widehat{Q_{u,j}^{(1)}}\br{\tau}\B_{\tau_{H^c}},\]
			\end{enumerate}
for some standard orthonormal bases $\set{\A_i}_{i=0}^{m^2-1}$ and $\set{\B_i}_{i=0}^{m^2-1}$.

			\item \textbf{Invariance from $\H_m^{\otimes n}$ to $L^2\br{\H_m^{\otimes h},\gamma_{2(m^2-1)\br{n-h}}}$}. ~For any $u\in[s]$, applying \cref{lem:jointinvariance} to $P_{u,i}^{(1)}$'s and $Q_{u,i}^{(1)}$'s and $H$, we obtain joint random operators
\[\br{\mathbf{P}_{u,i}^{(2)},\mathbf{Q}_{u,j}^{(2)}}\in L^2\br{\M_m^{\otimes h},\gamma_{2(m^2-1)\br{n-h}}}\times  L^2\br{\M_m^{\otimes h},\gamma_{2(m^2-1)\br{n-h}}}\]
for $1\leq i,j\leq 2(m^2-1)(n-h)$ with same joint random variables\linebreak $\br{\mathbf{g}_i,\mathbf{h}_i}_{i=1}^{2(m^2-1)\br{n-h}}\sim \G_{\rho}^{\otimes 2(m^2-1)(n-h)}$ such that the followings hold.

			\bigskip
			\begin{enumerate}
\item 	$\mathbf{P}_{u,i}^{(2)}$ and $\mathbf{Q}_{u,i}^{(2)}$ are both multilinear and of degree at most $d_1$ for $u\in[s], i\in[t]$.			
\item $\sum_{i=1}^t\mathbf{P}_{u,i}^{(2)}=\id$ and $\sum_{i=1}^t\mathbf{Q}_{u,i}^{(2)}=\id$ for $u\in[s]$.
				\item For any $u\in[s],i\in[t]$ $$N_2\br{\mathbf{P}_{u,i}^{(2)}}\leq\nnorm{P_{u,i}^{(1)}}_2,N_2\br{\mathbf{Q}_{u,i}^{(2)}}\leq\nnorm{Q_{u,i}^{(1)}}_2$$
				\item For any $u\in[s],i,j\in[t]$ $$\abs{\Tr\br{\br{P_{u,i}^{(1)}\otimes Q_{u,j}^{(1)}}\psi_{AB}^{\otimes n}}-\expec{}{\Tr~\br{\br{\mathbf{P}_{u,i}^{(2)}\otimes\mathbf{Q}_{u,j}^{(2)}}\psi_{AB}^{\otimes h}}}}\leq\delta$$
				\item For any $u,v\in[s]$ $$\expec{}{\frac{1}{m^h}\sum_{i=1}^t\Tr~\zeta\br{\mathbf{P}_{u,i}^{(2)}}}\leq O\br{t\br{\br{3^{d_1}m^{d_1/2}\sqrt{\tau}d_1}^{2/3}+\sqrt{\delta}}}$$
				$$\expec{}{\frac{1}{m^h}\sum_{i=1}^t\Tr~\zeta\br{\mathbf{Q}_{u,i}^{(2)}}}\leq O\br{t\br{\br{3^{d_1}m^{d_1/2}\sqrt{\tau}d_1}^{2/3}+\sqrt{\delta}}}$$
				
				\item If $P_{u,i}^{(1)}=P_{v,i}^{(1)}$, then $\mathbf{P}_{u,i}^{(2)}=\mathbf{P}_{v,i}^{(2)}$. If $Q_{u,i}^{(1)}=Q_{v,i}^{(1)}$, then $\mathbf{Q}_{u,i}^{(2)}=\mathbf{Q}_{v,i}^{(2)}$.
			\end{enumerate}
			\bigskip
			\item\textbf{Dimension reduction}. ~For any $u\in[s],i,j\in[t]$, applying \cref{lem:dimensionreduction} to \linebreak$\br{\mathbf{P}_{u,i}^{(2)},\mathbf{Q}_{v,j}^{(2)}}$ with $\delta\leftarrow\delta/2st^2, d\leftarrow d_1, n\leftarrow 2(m^2-1)(n-h),\alpha\leftarrow 1/9st^2$, if we sample $\mathbf{M}\sim\gamma_{n\times n_0}$, then items 1 to 3 in ~\cref{lem:dimensionreduction} hold for $f_{\mathbf{M}}\br{\mathbf{P}_{u,i}^{(2)}}$, $f_{\mathbf{M}}\br{\mathbf{Q}_{u,i}^{(2)}}$ with probability at least $1-\delta/2st^2-2/9st^2$, where $f_{\mathbf{M}}$ is the function introduced in \cref{lem:dimensionreduction}. Further applying a union bound on $u\in[s],i,j\in[t]$, we obtain that items 1 to 3 in~\cref{lem:dimensionreduction} hold with probability at least $7/9-\delta/2>0$. Thus we have joint random operators $\br{\mathbf{P}_{u,i}^{(3)},\mathbf{Q}_{u,i}^{(3)}}\in L^2\br{\H_m^{\otimes h},\gamma_{n_0}}\times L^2\br{\H_m^{\otimes h},\gamma_{n_0}}$ with the joint random variables drawn from $\G_{\rho}^{\otimes n_0}$ such that the following holds for all $u\in[s]$.
			\bigskip
			\begin{enumerate}
				\item $\sum_{i=1}^t\mathbf{P}_{u,i}^{(3)}=\id$ and $\sum_{i=1}^t\mathbf{Q}_{u,i}^{(3)}=\id$.
				\item For all $i\in[t]$,  $$N_2\br{\mathbf{P}_{u,i}^{(3)}}\leq(1+\delta)N_2\br{\mathbf{P}_{u,i}^{(2)}}~\mbox{and}~N_2\br{\mathbf{Q}_{u,i}^{(3)}}\leq(1+\delta)N_2\br{\mathbf{Q}_{u,i}^{(2)}}.$$
				
				\item $$\expec{}{\frac{1}{m^h}\sum_{i=1}^t\Tr~\zeta\br{\mathbf{P}_{u,i}^{(3)}}}\leq 3t\sqrt{s}\expec{}{\frac{1}{m^h}\sum_{i=1}^t\Tr~\zeta\br{\mathbf{P}_{u,i}^{(2)}}};$$
				$$\expec{}{\frac{1}{m^h}\sum_{i=1}^t\Tr~\zeta\br{\mathbf{Q}_{u,i}^{(3)}}}\leq 3t\sqrt{s}\expec{}{\frac{1}{m^h}\sum_{i=1}^t\Tr~\zeta\br{\mathbf{Q}_{u,i}^{(2)}}}.$$
				\item For all $i,j\in[t]$,
\begin{multline*}
\abs{\expec{}{\Tr\br{\br{\mathbf{P}_{u,i}^{(3)}\otimes \mathbf{Q}_{u,j}^{(3)}}\psi_{AB}^{\otimes h}}}-\expec{}{\Tr\br{\br{\mathbf{P}_{u,i}^{(2)}\otimes \mathbf{Q}_{u,j}^{(2)}}\psi_{AB}^{\otimes h}}}}\\
\leq\delta N_2\br{\mathbf{P}_{u,i}^{(2)}}N_2\br{\mathbf{Q}_{u,j}^{(2)}}.
\end{multline*}

				\item If $\mathbf{P}_{u,i}^{(2)}=\mathbf{P}_{v,i}^{(2)}$, then $\mathbf{P}_{u,i}^{(3)}=\mathbf{P}_{v,i}^{(3)}$.
				If $\mathbf{Q}_{u,i}^{(2)}=\mathbf{Q}_{v,i}^{(2)}$, then $\mathbf{Q}_{u,i}^{(3)}=\mathbf{Q}_{v,i}^{(3)}$.
			\end{enumerate}
			\bigskip
			Here $n_0=\frac{m^{O(h)}d_1^{O\br{d_1}}s^6t^{12}}{\delta^6}$.
			\item\textbf{Smooth random operators}. For any $u\in[s]$, we apply \cref{lem:smoothgaussian} to $\br{\mathbf{P}_{u,i}^{(3)},\mathbf{Q}_{u,i}^{(3)}}$ with $ h\leftarrow h, n\leftarrow n_0$ and obtain joint random operators $\br{\mathbf{P}_{u,i}^{(4)},\mathbf{Q}_{u,i}^{(4)}}\in L^2\br{\H_m^{\otimes h},\gamma_{n_0}}\times L^2\br{\H_m^{\otimes h},\gamma_{n_0}}$ such that the following holds.
			\bigskip
			\begin{enumerate}
				\item $\sum_{i=1}^t\mathbf{P}_{u,i}^{(4)}=\id$ and $\sum_{i=1}^t\mathbf{Q}_{u,i}^{(4)}=\id$.
				\item For all $i\in[t]$, $\deg\br{\mathbf{P}_{u,i}^{(4)}}\leq d_2$ and $\deg\br{\mathbf{Q}_{u,i}^{(4)}}\leq d_2$.
				
				\item For all $i\in[t]$, $N_2\br{\mathbf{P}_{u,i}^{(4)}}\leq N_2\br{\mathbf{P}_{u,i}^{(3)}}$ and $N_2\br{\mathbf{Q}_{u,i}^{(4)}}\leq N_2\br{\mathbf{Q}_{u,i}^{(3)}}$.
				\item
\begin{multline*}
\expec{}{\frac{1}{m^h}\sum_{i=1}^t\Tr~\zeta\br{\mathbf{P}_{u,i}^{(4)}}}\\
\leq2\br{\expec{}{\frac{1}{m^h}\sum_{i=1}^t\Tr~\zeta\br{\mathbf{P}_{u,i}^{(3)}}}+\delta \br{\sum_{i=1}^tN_2\br{\mathbf{P}_{u,i}^{(3)}}^2}};
\end{multline*}
\begin{multline*}\expec{}{\frac{1}{m^h}\sum_{i=1}^t\Tr~\zeta\br{\mathbf{Q}_{u,i}^{(4)}}}\\
\leq2\br{\expec{}{\frac{1}{m^h}\sum_{i=1}^t\Tr~\zeta\br{\mathbf{Q}_{u,i}^{(3)}}}+\delta \br{\sum_{i=1}^tN_2\br{\mathbf{Q}_{u,i}^{(3)}}^2}}.
\end{multline*}
				\item For all $i,j\in[t]$,
\begin{multline*}
\abs{\expec{}{\Tr\br{\br{\mathbf{P}_{u,i}^{(4)}\otimes\mathbf{Q}_{u,j}^{(4)}}\psi_{AB}^{\otimes h}}}-\expec{}{\Tr\br{\br{\mathbf{P}_{u,i}^{(3)}\otimes\mathbf{Q}_{u,j}^{(3)}}\psi_{AB}^{\otimes h}}}}\\
\leq \delta N_2\br{\mathbf{P}_{u,i}^{(3)}}N_2\br{\mathbf{Q}_{u,j}^{(3)}}.
\end{multline*}

				\item If $\mathbf{P}_{u,i}^{(3)}=\mathbf{P}_{v,i}^{(3)}$, then $\mathbf{P}_{u,i}^{(4)}=\mathbf{P}_{v,i}^{(4)}$.
				If $\mathbf{Q}_{u,i}^{(3)}=\mathbf{Q}_{v,i}^{(3)}$, then $\mathbf{Q}_{u,i}^{(4)}=\mathbf{Q}_{v,i}^{(4)}$.
			\end{enumerate}
			\bigskip
			Here $d_2=O\br{\frac{\log^2\frac{1}{\delta}}{\delta\br{1-\rho}}}$.
				
			\item\textbf{Multilinearization}. For any $u\in[s]$, suppose
			\[\br{\mathbf{P}_{u,i}^{(4)},\mathbf{Q}_{u,i}^{(4)}}=\br{\sum_{\sigma\in[m^2]_{\geq 0}^h}p^{(4)}_{u,i,\sigma}\br{\mathbf{g}}\A_{\sigma},\sum_{\sigma\in[m^2]_{\geq 0}^h}q^{(4)}_{u,i,\sigma}\br{\mathbf{h}}\B_{\sigma}},\]
where $\br{\mathbf{g},\mathbf{h}}\sim\G_{\rho}^{\otimes n_0}$.
Applying \cref{lem:multiliniearization} with $d\leftarrow d_2,h\leftarrow h, n\leftarrow n_0,\delta\leftarrow\tau$, we obtain the joint random operators
			\[\br{\mathbf{P}_{u,i}^{(5)},\mathbf{Q}_{u,i}^{(5)}}=\br{\sum_{\sigma\in[m^2]_{\geq 0}^h}p^{(5)}_{u,i,\sigma}\br{\mathbf{x}}\A_{\sigma},\sum_{\sigma\in[m^2]_{\geq 0}^h}q^{(5)}_{u,i,\sigma}\br{\mathbf{y}}\B_{\sigma}}\]
			where ${\br{\mathbf{x},\mathbf{y}}\sim\G_{\rho}^{\otimes n_0n_1}}$, with $n_1=O\br{\frac{d_2^2}{\tau^2}}$ such that the following holds.
			\bigskip
			\begin{enumerate}
\item $\mathbf{P}_{u,i}^{(5)}$ and $\mathbf{Q}_{u,i}^{(5)}$ are both multilinear for $1\leq i\leq t$.
				\item $\sum_{i=1}^t\mathbf{P}_{u,i}^{(5)}=\id$ and $\sum_{i=1}^t\mathbf{Q}_{u,i}^{(5)}=\id$.
				\item For all $i\in[t]$, both $\deg\br{\mathbf{P}_{u,i}^{(5)}}$ and $\deg\br{\mathbf{Q}_{u,i}^{(5)}}$ are at most $d_2$.
				\item For all $\br{i,j,k}\in[t]\times[n_0]\times[n_1]$ and $\sigma\in[m^2]_{\geq 0}^h$,
				
				$$\influence_{(j-1)n_1+k}\br{p^{(5)}_{u,i,\sigma}}\leq\tau\cdot\influence_j\br{p^{(4)}_{u,i,\sigma}},$$$$\influence_{(j-1)n_1+k}\br{q^{(5)}_{u,i,\sigma}}\leq\tau\cdot\influence_j\br{q^{(4)}_{u,i,\sigma}}.$$
				\item For all $i\in[t]$, $N_2\br{\mathbf{P}_{u,i}^{(5)}}\leq N_2\br{\mathbf{P}_{u,i}^{(4)}}~\mbox{and}~N_2\br{\mathbf{Q}_{u,i}^{(5)}}\leq N_2\br{\mathbf{Q}_{u,i}^{(4)}}.$
				\item \begin{multline*}\abs{\expec{}{\frac{1}{m^h}\sum_{i=1}^t\Tr~\zeta\br{\mathbf{P}_{u,i}^{(5)}}}-\expec{}{\frac{1}{m^h}\sum_{i=1}^t\Tr~\zeta\br{\mathbf{P}_{u,i}^{(4)}}}}\\\leq 4\tau \br{\sum_{i=1}^tN_2\br{\mathbf{P}_{u,i}^{(4)}}^2};\end{multline*}
				
				\begin{multline*}\abs{\expec{}{\frac{1}{m^h}\sum_{i=1}^t\Tr~\zeta\br{\mathbf{Q}_{u,i}^{(5)}}}-\expec{}{\frac{1}{m^h}\sum_{i=1}^t\Tr~\zeta\br{\mathbf{Q}_{u,i}^{(4)}}}}\\\leq 4\tau \br{\sum_{i=1}^tN_2\br{\mathbf{Q}_{u,i}^{(4)}}^2}.\end{multline*}
				\item For all $i,j\in[t]$, \begin{multline*}\abs{\expec{}{\Tr\br{\br{\mathbf{P}_{u,i}^{(5)}\otimes\mathbf{Q}_{u,j}^{(5)}}\psi_{AB}^{\otimes h}}}-\expec{}{\Tr\br{\br{\mathbf{P}_{u,i}^{(4)}\otimes\mathbf{Q}_{u,j}^{(4)}}\psi_{AB}^{\otimes h}}}}\\\leq\delta N_2\br{\mathbf{P}_{u,i}^{(4)}}N_2\br{\mathbf{Q}_{u,j}^{(4)}}.\end{multline*}
				\item If $\mathbf{P}_{u,i}^{(4)}=\mathbf{P}_{v,i}^{(4)}$, then $\mathbf{P}_{u,i}^{(5)}=\mathbf{P}_{v,i}^{(5)}$.
				If $\mathbf{Q}_{u,i}^{(4)}=\mathbf{Q}_{v,i}^{(4)}$, then $\mathbf{Q}_{u,i}^{(5)}=\mathbf{Q}_{v,i}^{(5)}$.
			\end{enumerate}
			\bigskip
			\item\textbf{Invariance from $L^2\br{\H_m^{\otimes h},\gamma_{n_0n_1}}$ to $\H_m^{\otimes h+n_0n_1}$}. From  item 1 and item 2 above, and  item 3 of \cref{fac:vecfun} and \cref{lem:randoperator}, we have
			\[\sum_{\sigma}\influence_j\br{p_{u,i,\sigma}^{(5)}}\leq\tau N_2\br{\mathbf{P}_{u,i}^{(4)}}^2.\]
			Similarly, we have
			\[\sum_{\sigma}\influence_i\br{q_{u,i,\sigma}^{(5)}}\leq\tau N_2\br{\mathbf{Q}_{u,i}^{(4)}}^2.\]
			For any $u\in[s]$, we apply \cref{lem:invariancejointgaussian} to $\br{\mathbf{P}_{u,i}^{(5)},\mathbf{Q}_{u,i}^{(5)}}$ with $n\leftarrow n_0n_1,h\leftarrow h,d\leftarrow d_2$,
			$$\tau\leftarrow\tau_0=\max_u\set{\max\set{\tau N_2\br{\mathbf{P}_{u,i}^{(4)}}^2,\tau N_2\br{\mathbf{Q}_{u,i}^{(4)}}^2}~:~u\in[s]}$$
			to get $\br{P_{u,i}^{(6)},Q_{u,i}^{(6)}}\in\H_m^{\otimes h+n_0n_1}\times \H_m^{\otimes h+n_0n_1}$ satisfying that
			
			\bigskip
			\begin{enumerate}
				\item $\sum_{i=1}^t P_{u,i}^{(6)}=\id$ and $\sum_{i=1}^tQ_{u,i}^{(6)}=\id$.
				\item For all $i,j\in[t]$ $$\Tr\br{\br{P_{u,i}^{(6)}\otimes Q_{u,j}^{(6)}}\psi_{AB}^{\otimes (n_0n_1+h)}}=\expec{}{\Tr\br{\br{\mathbf{P}_{u,i}^{(5)}\otimes\mathbf{Q}_{u,j}^{(5)}}\psi_{AB}^{\otimes h}}}.$$
				
				\item For all $i\in[t]$ $$N_2\br{\mathbf{P}_{u,i}^{(5)}}=\nnorm{P_{u,i}^{(6)}}_2,N_2\br{\mathbf{Q}_{u,i}^{(5)}}=\nnorm{Q_{u,i}^{(6)}}_2.$$
				
				\item For all $i\in[t]$ \begin{multline*}\abs{\expec{}{\frac{1}{m^h}\Tr~\zeta\br{\mathbf{P}_{u,i}^{(5)}}-\frac{1}{m^{n_0n_1+h}}\Tr~\zeta\br{P_{u,i}^{(6)}}}}\\\leq O\br{\br{3^{d_2}m^{d_2/2}d_2\sqrt{\tau_0}}^{2/3}};\end{multline*}
							
				\begin{multline*}\abs{\expec{}{\frac{1}{m^h}\Tr~\zeta\br{\mathbf{Q}_{u,i}^{(5)}}-\frac{1}{m^{n_0n_1+h}}\Tr~\zeta\br{Q_{u,i}^{(6)}}}}\\\leq O\br{\br{3^{d_2}m^{d_2/2}d_2\sqrt{\tau_0}}^{2/3}}.\end{multline*}
				\item If $\mathbf{P}_{u,i}^{(5)}=\mathbf{P}_{v,i}^{(5)}$, then $P_{u,i}^{(6)}=P_{v,i}^{(6)}$.
				If $\mathbf{Q}_{u,i}^{(5)}=\mathbf{Q}_{v,i}^{(5)}$, then $Q_{u,i}^{(6)}=Q_{v,i}^{(6)}$.
			\end{enumerate}
			\bigskip
			
			\item\textbf{Rounding to measurement operators}. Let $\pos{\br{P_{u,i}^{(6)}}}$ and $\pos{\br{Q_{u,i}^{(6)}}}$ be the matrices defined in Eq.\cref{eqn:geq0def}. Suppose

\[X_u=\sum_{i=1}^t\pos{\br{P_{u,i}^{(6)}}} ~\mbox{and}~Y_u=\sum_{i=1}^t\pos{\br{Q_{u,i}^{(6)}}}.\]
Define
			$$\br{\widetilde{P_{u,i}}}_{1\leq i\leq t}=\br{\sqrtpinv{X_u}\pos{\br{P_{u,1}^{(6)}}}\sqrtpinv{X_u},\dots,\sqrtpinv{X_u}\pos{\br{P_{u,t}^{(6)}}}\sqrtpinv{X_u}}$$
			$$\br{\widetilde{Q_{u,i}}}_{1\leq i\leq t}=\br{\sqrtpinv{Y_u}\pos{\br{Q_{u,1}^{(6)}}}\sqrtpinv{Y_u},\dots,\sqrtpinv{Y_u}\pos{\br{Q_{u,t}^{(6)}}}\sqrtpinv{Y_u}}$$

			It is easy to verify that $\br{\widetilde{P_{u,i}}}_{1\leq i\leq t}$ and $\br{\widetilde{Q_{u,i}}}_{1\leq i\leq t}$ are both sub-POVMs for $u\in[s]$.  Then
			\begin{eqnarray*}
				&&\abs{\Tr\br{\br{P_{u,i}^{(6)}\otimes Q_{u,j}^{(6)}}\psi_{AB}^{\otimes n_0n_1+h}}-\Tr\br{\br{\widetilde{P_{u,i}}\otimes \widetilde{Q_{u,j}}}\psi_{AB}^{\otimes n_0n_1+h}}}\\
				&\leq&\abs{\Tr\br{\br{P_{u,i}^{(6)}\otimes \br{Q_{u,j}^{(6)}-\widetilde{Q_{u,j}}}}\psi_{AB}^{\otimes n_0n_1+h}}}+\abs{\Tr\br{\br{\br{P_{u,i}^{(6)}-\widetilde{P_{u,i}}}\otimes \widetilde{Q_{u,j}}}\psi_{AB}^{\otimes n_0n_1+h}}}\\
				&\leq&\nnorm{P_{u,i}^{(6)}}_2\nnorm{Q_{u,j}^{(6)}-\widetilde{Q_{u,j}}}_2+\nnorm{\widetilde{Q_{u,j}}}_2\nnorm{P_{u,i}^{(6)}-\widetilde{P_{u,i}}}_2\\
				&\leq&\nnorm{P_{u,i}^{(6)}}_2\br{\sum_{j=1}^t\nnorm{Q_{u,j}^{(6)}-\widetilde{Q_{u,j}}}_2^2}^{\frac{1}{2}}+\br{\sum_{i=1}^t\nnorm{P_{u,i}^{(6)}-\widetilde{P_{u,i}}}_2^2}^{\frac{1}{2}}\\
				&\leq&\nnorm{P_{u,i}^{(6)}}_2\br{ \frac{3(t+1)}{m^{h+n_0n_1}}\sum_{i=1}^t\Tr~\zeta\br{Q_{u,i}^{(6)}} +6\br{ \frac{t}{m^{h+n_0n_1}}\sum_{i=1}^t\Tr~\zeta\br{Q_{u,i}^{(6)}} }^{\frac{1}{2}}}^{\frac{1}{2}}\\
				&&+\br{ \frac{3(t+1)}{m^{h+n_0n_1}}\sum_{i=1}^t\Tr~\zeta\br{P_{u,i}^{(6)}} +6\br{ \frac{t}{m^{h+n_0n_1}}\sum_{i=1}^t\Tr~\zeta\br{P_{u,i}^{(6)}} }^{\frac{1}{2}}}^{\frac{1}{2}},
			\end{eqnarray*}
		\end{itemize}
where the second inequality follows from \cref{fac:cauchyschwartz} and the last inequality follows from \cref{lem:closedelta}.

Keeping track of the parameters in the construction, we are able to upper bound $\sum_{j=1}^t\frac{1}{m^{h+n_0n_1}}\Tr~\zeta\br{Q_{u,j}^{(6)}}$. The dependency of parameters is pictorially described in \cref{fig:dependence}. Choosing

\begin{equation}\label{eqn:paramters}
  \delta=\epsilon^8/\br{10^{16}t^2s}, \tau=\epsilon^{12}/\br{s^2t^3\exp\br{\frac{\log m\log^2\frac{1}{\delta}}{\br{1-\rho}\delta}}}~\mbox{and}~ D=n_0n_1+h,
\end{equation}
 we conclude the desired result.
	\end{proof}

\newcommand{\rectab}[1]{\begin{tabular}{|c|}\hline#1\\\hline\end{tabular}}
\newcommand{\somespace}{5em}
\setlength{\tabcolsep}{1pt}
\begin{figure}
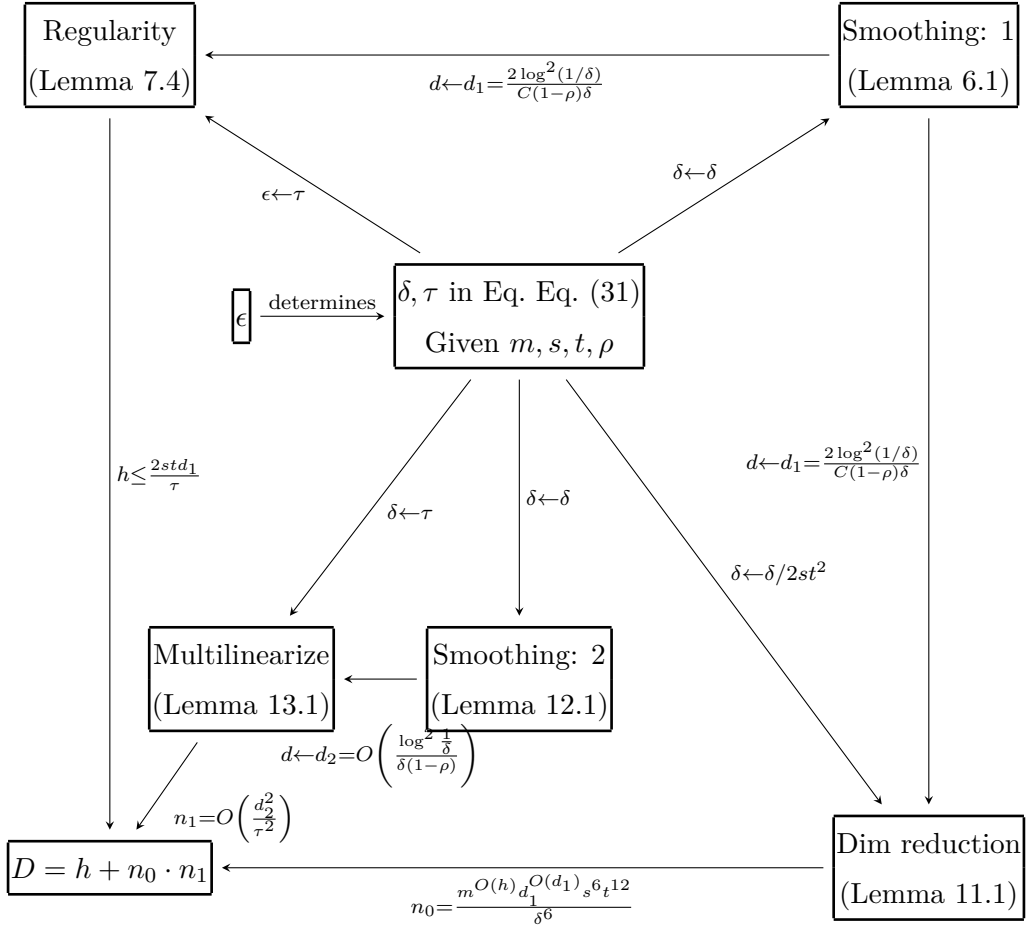
\label{fig:dependence}
\begin{center}
\begin{codi}
\obj{
|(C)|\rectab{Regularity\\(\cref{lem:regular})}&&[\somespace]&[\somespace]&|(B)|\rectab{Smoothing: 1\\
(\cref{lem:smoothing of strategies})}\\
&&&&\\
&|(H)|\rectab{$\epsilon$}&|(A)|\rectab{$\delta,\tau$ in Eq.~\cref{eqn:paramters}\\Given $m,s,t,\rho$}&&\\[8em]
&|(G)|\rectab{Multilinearize\\(\cref{lem:multiliniearization})}&|(F)|\rectab{Smoothing: 2\\(\cref{lem:smoothgaussian})}&&\\[2em]
|(E)|\rectab{$D=h+n_0\cdot n_1$}&&&&|(D)|\rectab{Dim reduction\\(\cref{lem:dimensionreduction})}\\
};
\mor H \mathrm{{determines}}:-> A \delta\leftarrow\delta :-> B;
\mor :B d\leftarrow d_1=\frac{{2\log^2(1/\delta)}}{{C\br{1-\rho}\delta}}:-> C;
\mor :C h\leq{\frac{2std_1}{\tau} }:-> E;
\mor A \epsilon\leftarrow\tau :-> C;
\mor [swap]:B d\leftarrow d_1=\frac{{2\log^2(1/\delta)}}{{C\br{1-\rho}\delta}} :-> D;
\mor A \delta\leftarrow\delta/2st^2:-> D;
\mor :D n_0={\frac{m^{O(h)}d_1^{O\br{d_1}}s^6t^{12}}{\delta^6}}:-> E;
\mor A \delta\leftarrow\delta:-> F {\rule{0pt}{1cm}d\leftarrow d_2=O\br{{\frac{\log^2\frac{1}{\delta}}{\delta\br{1-\rho}}}}}:-> G;
\mor A \delta\leftarrow\tau:-> G n_1=O\br{{\frac{d_2^2}{\tau^2}}}:-> E;
\end{codi}
\end{center}\caption{Dependency of parameters in the proof of \cref{lem:nijslem}}
\end{figure}
	\section{Summary and open problems}\label{sec:openproblems}
	
	In this work, we prove that if the players in a nonlocal game are provided with arbitrarily many copies of noisy MESs, then the game is decidable. We prove it by reducing the problem to the decidability of the quantum non-interactive simulation
of joint distributions and uses the framework in~\cite{7782969}. To this end, we systematically generalize the theory of Boolean analysis to the space of matrices and the space of random matrices.

This work also raises many interesting open questions.

\begin{enumerate}
  \item In this paper, we prove a quantum invariance principle for $C^3$ piecewise polynomial functions. Can we establish a quantum invariance principle for a larger class of functions? More specifically, can we prove a quantum analog of the classical invariance principle in~\cite{IMossel:2012}?

  \item As mentioned above, invariance principles have rich applications in theoretical computer science. Are there other applications of the quantum invariance principle established in this paper?

  \item We essentially use the dimension reduction for polynomials in Gaussian spaces~\cite{Ghazi:2018:DRP:3235586.3235614} as an intermediate step. Is it possible to simplify the proof by proving a quantum dimension reduction directly?

  \item Given a nonlocal game $G$ and a bipartite state $\psi$, we say $\br{G,\psi}$ is a {\em mono-state game} if the players in $G$ only share arbitrarily many copies of a bipartite state $\psi$. This paper proves that the mono-state games $\br{G,\psi}$ are decidable if $\psi$ is a noisy MES. It is known that the games are undecidable if $\psi$ is an EPR state because of Ji et al.'s breakthrough work~\cite{JNVWY'20}. Due to PCP theorem~\cite{Arora:1998:PVH:278298.278306, Arora:1998:PCP:273865.273901} it is $\mathrm{NP}$-complete when $\psi$ is separable. Hence it is interesting to identify the hardness of the mono-state games $\br{G,\psi}$ for any given bipartite state $\psi$.

  \item  Can we use the framework developed in this paper to design algorithms for other "tensored" quantities in quantum information theory and quantum complexity theory such as quantum channel capacities~\cite{8242350}, the regularizations of the various entanglement measures~\cite{Plbnio:2007:IEM:2011706.2011707}, quantum information complexity~\cite{Touchette:2015:QIC:2746539.2746613}, etc.?

\item Is it possible to generalize our result to multiplayer nonlocal games?  Note that the decidability of non-interactive simulations of $k$-partite distributions for $k\geq 3$ is still open~\cite{7782969}.
\end{enumerate}
	
	\section{Smoothing operators}\label{sec:smoothoperators}
The main lemma in this section is the following.

\begin{lemma}\label{lem:smoothing of strategies}
	Given parameters $0\leq\rho<1$, $0<\delta<1$, integer $n>0, m>1$ an $m$-dimensional noisy MES $\psi_{AB}$  with the maximal correlation $\rho=\rho\br{\psi_{AB}}$, there exists $d=d\br{\rho,\delta}$ and a map $f:\H_m^{\otimes n}\rightarrow \H_m^{\otimes n},$ such that for any $P\in\H_m^{\otimes n}, Q\in\H_m^{\otimes n}$ satisfying $0\leq P\leq \id$ and $0\leq Q\leq \id$, the operators $P^{(1)}=f\br{P}$ and $Q^{(1)}=f\br{Q}$ satisfy that
		
		\begin{enumerate}
			\item $0\leq P^{(1)}\leq \id~\mbox{and}~0\leq Q^{(1)}\leq \id;$
			\item $\nnorm{P^{(1)}}_2\leq\nnorm{P}_2~\mbox{and}~\nnorm{Q^{(1)}}_2\leq\nnorm{Q}_2;$
			\item $\abs{\Tr\br{\br{P^{(1)}\otimes Q^{(1)}}\psi^{\otimes n}_{AB}}-\Tr\br{\br{P\otimes Q}\psi^{\otimes n}_{AB}}}\leq\delta,$
			\item $\nnorm{\br{P^{(1)}}^{>d}}^2_2\leq\delta~\mbox{and}~\nnorm{\br{Q^{(1)}}^{>d}}^2_2\leq\delta,$

            \item the map $f$ is linear and unital.
		\end{enumerate}
		In particular, we can take $d=\frac{2\log^2\frac{1}{\delta}}{C\br{1-\rho}\delta}$ for some constant $C$.
\end{lemma}

Before proving \cref{lem:smoothing of strategies}, we need the following lemma, whose classical analog is proved in~\cite{Mossel:2010}.

\begin{lemma}\label{lem:Tsmooth}
	Given a noisy MES $\psi_{AB}$ with the maximal correlation $\rho=\rho(\psi_{AB})<1$, a parameter $0<\epsilon<1$ and operators $P\in\H_m^{\otimes n}, Q\in\H_m^{\otimes n}$, it holds that
	\[\abs{\Tr\br{\br{P\otimes Q}\psi_{AB}^{\otimes n}}-\Tr\br{\br{\Delta_{1-\gamma}\br{P}\otimes\Delta_{1-\gamma}\br{Q}}\psi_{AB}^{\otimes n}}}\leq 2\epsilon\sqrt{\var{P}\var{Q}},\]
	for
\begin{equation}\label{eqn:gammachoice}
  0\leq \gamma\leq1-\br{1-\epsilon}^{\log\rho/\br{\log\epsilon+\log\rho}}.
\end{equation}
In particular, we can take
	\[\gamma=C\frac{\br{1-\rho}\epsilon}{\log\br{1/\epsilon}}.\]
	where $C$ is an absolute constant.
\end{lemma}
\begin{proof}

	We first show that for any $P,Q\in\H_m^{\otimes n}$
	\begin{align}\label{eqn:varpvarq}
	  &\abs{\Tr\br{\br{P\otimes Q}\psi_{AB}^{\otimes n}}-\Tr\br{\br{P\otimes \Delta_{1-\gamma}\br{Q}}\psi_{AB}^{\otimes n}}} \\\nonumber
	  =& \abs{\Tr\br{\br{P\otimes\br{\id-\Delta_{1-\gamma}}\br{Q}}\psi_{AB}^{\otimes n}}}\leq\epsilon\sqrt{\var{P}\var{Q}}.
	\end{align}
	By \cref{lem:bonamibecknerdef} item 1, for all $S\subseteq[n]$, we have
	\[\br{\id-\Delta_{1-\gamma}}\br{Q[S]}=\br{1-\br{1-\gamma}^{\abs{S}}}Q[S].\]
	Using \cref{prop:markovenfronstein} and \cref{prop:markovoperatornorm},
	\begin{equation}\label{eqn:tqs}
	\twonorm{\T\br{Q[S]}}=\twonorm{\T\br{Q}[S]}\leq\rho^{\abs{S}}\twonorm{Q[S]}.
	\end{equation}
	Let $\T'=\T\circ\br{\id_B-\Delta_{1-\gamma}}$.
	From \cref{def:markovoperator},
	\[\Tr\br{\br{P\otimes\br{\id-\Delta_{1-\gamma}}\br{Q}}\psi_{AB}^{\otimes n}}=\innerproduct{P}{\T'\br{Q}}\]	
	Combining \cref{lem:bonamibecknerdef} and Eq. \cref{eqn:tqs}, we have
\begin{eqnarray}
  &\twonorm{\T'\br{Q[S]}}\leq\rho^{\abs{S}}\cdot\br{1-\br{1-\gamma}^{\abs{S}}}\twonorm{Q[S]}\nonumber\\
  \leq&\min\set{\rho^{\abs{S}},\br{1-\br{1-\gamma}^{\abs{S}}}}\twonorm{Q[S]}\leq\epsilon\twonorm{Q[S]},\label{eqn:TprimeQS}
\end{eqnarray}
where the second inequality follows from our choice of $\gamma$ in Eq.~\cref{eqn:gammachoice}.
	From \cref{lem:bonamibecknerdef},  \cref{prop:enfronsteinortho} and  \cref{prop:markovenfronstein}, $\innerproduct{\T'\br{Q[S]}}{P[S']}=0$ if $S\neq S'$. Note that  $\T'\br{\id}=0$. And thus $\T'\br{Q[\emptyset]}=0$.
	Therefore,
	\begin{align*}
		\abs{\innerproduct{P}{\T'\br{Q}}}&=\abs{\sum_{S\neq\emptyset}\innerproduct{P[S]}{\T'\br{Q[S]}}}\\
		&\leq\sqrt{\frac{1}{m^n}\sum_{S\neq\emptyset}\twonorm{P[S]}^2}\cdot\sqrt{\frac{1}{m^n}\sum_{S\neq\emptyset}\twonorm{\T'\br{Q[S]}}^2}\\
		&\leq\epsilon\sqrt{\frac{1}{m^n}\sum_{S\neq\emptyset}\twonorm{P[S]}^2}\cdot\sqrt{\frac{1}{m^n}\sum_{S\neq\emptyset}\twonorm{Q[S]}^2}\\
		&\leq\epsilon\sqrt{\var{P}\var{Q}},
	\end{align*}
	where the second inequality follows from Eq.~\eqref{eqn:TprimeQS} and the last inequality follows from the orthogonality of the Efron-Stein decomposition.

Thus we obtain Eq.~\cref{eqn:varpvarq}. Note that it holds for any $P,Q\in\H_m^{\otimes n}$. We apply the same argument to $P$ and $\Delta_{1-\gamma}\br{Q}$ and get
\begin{equation*}
  \abs{\Tr~\br{\br{P-\Delta_{1-\gamma}\br{P}}\otimes \Delta_{1-\gamma}\br{Q}}\psi_{AB}^{\otimes n}}\leq\epsilon\sqrt{\var{P}\var{\Delta_{1-\gamma}\br{Q}}}.
\end{equation*}
Note that $\var{\Delta_{1-\gamma}\br{Q}}\leq\var{Q}$ by \cref{lem:variance} and \cref{lem:bonamibecknerdef} item 1. Thus
\begin{eqnarray*}
	  &&\abs{\Tr\br{\br{P\otimes Q}\psi_{AB}^{\otimes n}}-\Tr\br{\br{\Delta_{1-\gamma}\br{P}\otimes\Delta_{1-\gamma}\br{Q}}\psi_{AB}^{\otimes n}}} \\
	  &\leq&\abs{\Tr\br{\br{P\otimes Q}\psi_{AB}^{\otimes n}}-\Tr\br{\br{P\otimes \Delta_{1-\gamma}\br{Q}}\psi_{AB}^{\otimes n}}} \\
&&+ \abs{\Tr\br{\br{P\otimes \Delta_{1-\gamma}\br{Q}}\psi_{AB}^{\otimes n}}-\Tr\br{\br{\Delta_{1-\gamma}\br{P}\otimes \Delta_{1-\gamma}\br{Q}}\psi_{AB}^{\otimes n}}}\\
&\leq&2\epsilon\sqrt{\var{P}\var{Q}}.
	\end{eqnarray*}
We conclude the desired result.
\end{proof}	

We are now ready to prove \cref{lem:smoothing of strategies}.

\begin{proof}[Proof of \cref{lem:smoothing of strategies}]
	Given $\rho$ and $\delta$, we choose $\epsilon=\delta/2$ and $\gamma$ in \cref{lem:Tsmooth}. We choose $d$ sufficiently large such that $(1-\gamma)^{2d}\leq\delta$, that is, $d=\frac{\log\frac{1}{\delta}}{2\gamma}$. Given $P,Q$ as in \cref{lem:smoothing of strategies}, we set $$P^{(1)}=\Delta_{1-\gamma}\br{P}~\mbox{and}~ Q^{(1)}=\Delta_{1-\gamma}\br{Q}.$$ From \cref{lem:bonamibecknerdef} item 2 and item 3, $0\leq P^{(1)}\leq 1$ and $0\leq Q^{(1)}\leq 1$. By \cref{lem:bonamibecknerdef} item 2, $\nnorm{P^{(1)}}_2\leq\nnorm{P}_2$ and $ \nnorm{Q^{(1)}}_2\leq\nnorm{Q}_2$. Note that $\var{P}\leq\nnorm{P}_2^2\leq 1~\mbox{and}~ \var{Q}\leq\nnorm{Q}_2^2\leq 1$
due to \cref{lem:variance}. Combing with \cref{lem:Tsmooth}, we get
	\[\abs{\Tr\br{\br{P\otimes Q}\psi_{AB}^{\otimes n}}-\Tr\br{\br{P^{(1)}\otimes Q^{(1)}}\psi_{AB}^{\otimes n}}}\leq 2\epsilon=\delta.\]
	Notice that $$\widehat{P^{(1)}}\br{\sigma}=\br{1-\gamma}^{\abs{\sigma}}\widehat{P}\br{\sigma}~\mbox{and}~ \widehat{Q^{(1)}}\br{\sigma}=\br{1-\gamma}^{\abs{\sigma}}\widehat{Q}\br{\sigma}$$ by \cref{lem:bonamibecknerdef} item 1. Thus, we get that
	\begin{eqnarray*} &&\sum_{\abs{\sigma}>d}\widehat{P^{(1)}}\br{\sigma}^2\leq\br{1-\gamma}^{2d}\sum_{\abs{\sigma}>d}\widehat{P}\br{\sigma}^2\leq\br{1-\gamma}^{2d}\nnorm{P}_2^2\leq\delta;\\		&&\sum_{\abs{\sigma}>d}\widehat{Q^{(1)}}\br{\sigma}^2\leq\br{1-\gamma}^{2d}\sum_{\abs{\sigma}>d}\widehat{Q}\br{\sigma}^2\leq\br{1-\gamma}^{2d}\nnorm{Q}_2^2\leq\delta,
	\end{eqnarray*}
	where the second inequalities in both equations are from \cref{fac:basicfourier} item 3.

Item 5 follows from the fact that $\Delta_{1-\gamma}\br{\cdot}$ is linear and unital by definition.
\end{proof}
	
	\section{Joint regularity lemma}\label{sec:jointregular}
From Lemma~\ref{lem:smoothing of strategies}, $P$ and $Q$ can be approximated by low-degree operators after the smoothing operation. Therefore, the number of registers that have high influences is bounded by Lemma~\ref{lem:partialvariance}. Suppose $P$ and $Q$ are both diagonal. It is proved in~\cite[Lemma 5.1]{Ghazi:2018:DRP:3235586.3235614} that even if we randomly restrict the registers in $H$ to the elements in computational basis, the influences of the registers not in $H$ are still small~\cite{v010a002,doi:10.1137/100783534}. More specifically, we assume that $P$ and $Q$ are diagonal matrices with $P=\sum_{s\in[m^2]_{\geq 0}^{|H|}}\ketbra{s}\otimes P_s$ and
$Q=\sum_{s\in[m^2]_{\geq 0}^{|H|}}\ketbra{s}\otimes Q_s,$ where $\set{\ket{s}}_{s\in[m^2]_{\geq 0}^{|H|}}$ is the computational basis of the registers in $H$.
 Then it is shown that all the registers of $P_s$ and $Q_s$ have low influences with high probability over sampling $s$. This fact is crucial to all the followup works~\cite{7782969,doi:10.1137/1.9781611975031.174}. However, random restrictions cannot be applied to $P$ and $Q$ if they are not diagonal. If the off-diagonal entries of $P$ and $Q$ are nonzero, then the information about the off-diagonal entries might be lost if we restrict the registers in $H$ to a computational basis. Instead of random restrictions, we expand the operators in a properly chosen standard orthonormal basis. Before getting into the details, we introduce the notion of {\em correlation matrices}.

\begin{definition}\label{def:covariancematrix}
	Given quantum systems $A$ and $B$ with dimensions $m_A$ and $m_B$, respectively, and a bipartite quantum state $\psi_{AB}$, let $\A=\set{\A_i}_{i\in[m_A^2]_{\geq 0}}$  and $\B=\set{\B_i}_{i\in[m_B^2]_{\geq 0}}$ be standard orthonormal bases in the spaces $\M\br{A}$ and \linebreak$\M\br{B}$, respectively.
	The correlation matrix of $\br{\psi_{AB}, \A,\B}$ is an $m_A^2\times m_B^2$ matrix, where
	\[\mathsf{Corr}\br{\psi_{AB},\A,\B}_{i,j}=\Tr\br{\br{\A_i\otimes\B_j}\psi_{AB}}.\]
	for $i\in[m_A^2]_{\geq 0}, j\in[m_B^2]_{\geq 0}$. The correlation matrix of $\br{\psi_{AB}^{\otimes n},\A,\B}$ is an $m_A^{2n}\times m_B^{2n}$ matrix, where
	\[\mathsf{Corr}\br{\psi_{AB}^{\otimes n},\A,\B}_{\sigma,\tau}=\Tr\br{\br{\A_{\sigma}\otimes\B_{\tau}}\psi_{AB}^{\otimes n}},\]
	for $\sigma\in[m_A^2]^n_{\geq 0}$ and $\tau\in[m_B^2]^n_{\geq 0}$.
\end{definition}
The lemma below follows by the definition.
\begin{lemma}\label{lem:convariancetensor}
	Let  $A$ and $B$ be two quantum systems of dimensions $m_A$ and $m_B$, respectively. For any standard orthonormal bases in $\M\br{A}$ and $\M\br{B}$ and  bipartite quantum state $\psi_{AB}$,  it holds for all positive integer $n$ that
\[\mathsf{Corr}\br{\psi_{AB}^{\otimes n},\A,\B}=\mathsf{Corr}\br{\psi_{AB},\A,\B}^{\otimes n}.\]
\end{lemma}

\begin{lemma}\label{lem:normofM}
	Given quantum systems $A$ and $B$ with dimensions $m_A$ and $m_B$, respectively, and a noisy MES $\psi_{AB}$, for any standard orthonormal bases \linebreak$\A=\set{\A_i}_{i\in[m_A^2]_{\geq 0}}$  and $\B=\set{\B_i}_{i\in[m_B^2]_{\geq 0}}$ in $\M\br{A}$ and $\M\br{B}$, respectively, it holds that $s_1\br{\mathsf{Corr}\br{\psi_{AB},\A,\B}}=1$ and $s_2\br{\mathsf{Corr}\br{\psi_{AB},\A,\B}}=\rho$, where $\rho=\rho\br{\psi_{AB}}$ and $s_i\br{\cdot}$ is the $i$-th largest singular value.
	
	Moreover, there exist standard orthonormal bases $\A=\set{\A_i}_{i\in[m_A^2]_{\geq 0}}$  and $\B=\set{\B_i}_{i\in[m_B^2]_{\geq 0}}$ in $\br{\M\br{A},\psi_A}$ and $\br{\M\br{B},\psi_B}$, respectively, such that
	\[\mathsf{Corr}\br{\psi_{AB},\A,\B}_{i,j}=\begin{cases}c_i~&\mbox{if $i=j$}\\0~&\mbox{otherwise},\end{cases}\]
	where $c_1=1, c_2=\rho\br{\psi_{AB}}$ and $c_1\geq c_2\geq c_3\geq\ldots$.
\end{lemma}
\begin{proof}
	We assume that the dimensions of $A$ and $B$ are both $m$. Similar arguments apply when the dimensions of the two systems are different. By \cref{fac:unitarybasis} and \cref{lem:convariancetensor}, the correlation matrices of $\br{\psi_{AB},\A,\B}$ for different orthonormal bases $\br{\A,\B}$ are equivalent up to left and right orthonormal transformations. Thus the singular values of the correlation matrix of $\br{\psi_{AB},\A,\B}$ are independent of the choices of $\A$ and $\B$.

 Let $M=\mathsf{Corr}\br{\psi_{AB},\A,\B}$. Recall that $\A_0=\B_0=\id$, by \cref{fac:cauchyschwartz} item 1 we have that $M_{0,0}=1$, $M_{0,i}=M_{i,0}=0$ for all $i\in[m^2-1]$. Note that $M$ is a real matrix. Thus we set $M=U^{\dagger}DV$ to be a singular value decomposition of $M$ where $U$ and $V$ are orthogonal matrices and $U_{0,0}=V_{0,0}=1, U_{0,i}=U_{i,0}=V_{0,i}=V_{i,0}=0$ for $1\leq i\leq m^2-1$. Let $$\P_i=\sum_{j=1}^{m^2-1}U_{ij}\A_j~\mbox{and}~\Q_i=\sum_{j=1}^{m^2-1}V_{ji}^{\dagger}\B_j$$ for $1\leq i\leq m^2-1$. Then
	$\Tr~\P_i=\Tr~\Q_i=0$ since $\Tr~\A_j=\Tr~\B_j=0$ for $1\leq j\leq m^2-1$. And thus
$$\var{\P_i}=\nnorm{\P_i}_2^2=\sum_j\abs{U_{ij}}^2\nnorm{\A_j}_2^2=1.$$
Similarly, $\var{\Q_i}=1$. Thus by the definition of quantum maximal correlations,
	\begin{eqnarray*}
		\rho&\geq&\Tr\br{\br{\P_i\otimes \Q_{i'}}\psi_{AB}}\\
		&=&\sum_{j,k=1}^{m^2-1}U_{ij}V^{\dagger}_{ki'}\Tr\br{\br{\A_j\otimes\B_k}\psi_{AB}}\\
		&=&\sum_{j,k=1}^{m^2-1}U_{ij}V_{ki'}^{\dagger}M_{jk}\\
        &=&\br{UMV^{\dagger}}_{ii'}\\
		&=&\delta_{i,i'}D_{ii}.
	\end{eqnarray*}
	Hence $\norm{M}=1$ and $s_2\br{M}\leq\rho$.

	From \cref{prop:maximalvariance}, we assume that $P, Q$ are two Hermitian operators with $\Tr~P=\Tr~Q=0$ in $\H\br{A}, \H\br{B}$ which achieve $\rho$. Then from \cref{def:maximalcorrelation}, both $\set{\id_{m_A}, P}$ and $\set{\id_{m_B}, Q}$ can be extended to orthonormal bases in the real Hilbert space consisting of all Hermitian matrices of dimension $m$ over $\reals$, say $\set{\P_i}_{i=0}^{m^2-1}$ and $\set{\Q_i}_{i=0}^{m^2-1}$ where $\P_0=\Q_0=\id_m, \P_1=P$ and $\Q_1=Q$. Since $\set{\P_i}_{i=0}^{m^2-1}$ and $\set{\Q_i}_{i=0}^{m^2-1}$ are also orthonormal in $\M_m$, respectively, they are both standard orthonormal bases in $\M_m$ as well. Let $M'$ be the corresponding correlation matrix. Again by \cref{fac:cauchyschwartz}, $M'_{0,0}=1, M'_{1,1}=\rho$ and $M'_{0,i}=M'_{i,0}=0$ for $1\leq i\leq m^2-1$. Since the largest singular value of a matrix is not less than the largest diagonal element, we have $s_2\br{M'}\geq\rho$. As argued above, $M$ and $M'$ have the same singular values. Therefore, $s_2(M)=s_2(M')=\rho$.

The second part of the lemma follows by \cref{fac:unitarybasis} and \cref{lem:jointbasis}.
\end{proof}

We reach the main lemma in this section.
\begin{lemma}\label{lem:regular}
	Given parameters $0\leq\delta,\epsilon<1$, integers $d\geq 0, n>0, m>1$, a noisy $m$-dimensional MES $\psi_{AB}$ and operators $P\in\H_m^{\otimes n}, Q\in\H_m^{\otimes n}$ satisfying that $\nnorm{P}_2\leq 1, \nnorm{Q}_2\leq 1, \nnorm{P^{>d}}^2_2\leq\delta$ and $\nnorm{Q^{>d}}^2_2\leq\delta$, let $\A=\set{\A_i}_{i=0}^{m^2-1}$ and $\B=\set{\B_i}_{i=0}^{m^2-1}$ be the standard orthonormal bases induced by \cref{lem:normofM} and $\br{c_i}_{i=0}^{m^2-1}$ be the singular values of $\mathsf{Corr}\br{\psi_{AB},\A,\B}$ in non-increasing order. Suppose $P$ and $Q$ have Fourier expansions
	\[P=\sum_{\sigma\in[m^2]_{\geq 0}^n}\widehat{P}\br{\sigma}\A_{\sigma}~\mbox{and}~Q=\sum_{\sigma\in[m^2]_{\geq 0}^n}\widehat{Q}\br{\sigma}\B_{\sigma}.\]
	Then there exists a subset $H\subseteq[n]$ of size $h=\abs{H}\leq\frac{2d}{\epsilon}$ such that for any $i\notin H$,
$$\influence_i\br{P^{\leq d}}\leq\epsilon~\mbox{and}~\influence_i\br{Q^{\leq d}}\leq\epsilon$$ Without loss of generality, we assume $H=[h]$, then
	\[\Tr\br{\br{P\otimes Q}\psi_{AB}^{\otimes n}}=\sum_{\sigma\in[m^2]_{\geq 0}^h}c_{\sigma}\Tr\br{\br{P_{\sigma}\otimes Q_{\sigma}}\psi_{AB}^{\otimes(n-h)}},\]
	where $c_{\sigma}=\prod_{i=1}^hc_{\sigma_i}$ for any $\sigma\in[m^2]_{\geq 0}^h$;
	and
	
	$$P_{\sigma}=\sum_{\tau\in[m^2]_{\geq 0}^n:\tau_H=\sigma}\widehat{P}\br{\tau}\A_{\tau_{H^c}}$$
	$$Q_{\sigma}=\sum_{\tau\in[m^2]_{\geq 0}^n:\tau_H=\sigma}\widehat{Q}\br{\tau}\B_{\tau_{H^c}}.$$
\end{lemma}
\begin{proof}
	Set $H=\set{i:\influence_i\br{P^{\leq d}}\geq\epsilon~\mbox{or}~\influence_i\br{Q^{\leq d}}\geq\epsilon}$. From \cref{lem:partialvariance} item 4 and the fact that $\nnorm{P^{\leq d}}_2\leq\nnorm{P}_2$ and $\nnorm{Q^{\leq d}}_2\leq\nnorm{Q}_2$ we have $\abs{H}\leq\frac{2d}{\epsilon}$. Without loss of generality, we assume that $H=[h]$. By the definitions of $P_{\sigma}$ and $Q_{\sigma}$, we have
\[P=\sum_{\sigma\in[m^2]_{\geq 0}^h}P_{\sigma}\otimes\A_{\sigma}~\mbox{and}~Q=\sum_{\sigma\in[m^2]_{\geq 0}^h}Q_{\sigma}\otimes\B_{\sigma}.\]
Thus
\begin{align*}
&\Tr~\br{\br{P\otimes Q}\psi_{AB}^{\otimes n}}\\
=&\sum_{\sigma,\sigma'\in[m^2]_{\geq 0}^h}\br{\Tr~\br{\br{P_{\sigma}\otimes Q_{\sigma'}}\psi_{AB}^{\otimes \br{n-h}}}}\br{\Tr\br{\br{\A_{\sigma}\otimes\B_{\sigma'}}\psi_{AB}^{\otimes h}}}\\
=&\sum_{\sigma\in[m^2]_{\geq 0}^h}c_{\sigma}\Tr\br{\br{P_{\sigma}\otimes Q_{\sigma}}\psi_{AB}^{\otimes(n-h)}},
\end{align*}
where the second equality follows from \cref{lem:convariancetensor} and \cref{lem:normofM}.

\end{proof}

	\section{Hypercontractive inequality for high dimensional random operators}\label{sec:hypercontractive}

Recall that $\nnorm{\cdot}_p$ represents the normalized $p$-norm defined in \cref{eqn:nnormdef}. For any linear map $L:\M_k\rightarrow\M_k$, we define two families of norms as follows.
\begin{definition}\label{def:ptoqnorm}
Given $1\leq p,q\leq\infty$, integers $m,n\geq 1$, a linear map $L:\M_m\rightarrow\M_n$, the $p$-to-$q$ norm of $L$ is defined to be
$$\norm{L}_{p\rightarrow q}=\sup_{M\neq0}\frac{\norm{L(M)}_q}{\norm{M}_p}.$$
The normalized $p$-to-$q$ norm of $L$ is
$$\nnorm{L}_{p\rightarrow q}=\sup_{M\neq0}\frac{\nnorm{L(M)}_q}{\nnorm{M}_p}$$
\end{definition}

	The purpose of this section is to establish a hypercontractive inequality for random operators. A hypercontractive inequality on space $L^2\br{\complex^k,\gamma_n}$ asserts that the normalized $p$-to-$q$ norm of the Ornstein-Uhlenbeck operator $U_{\rho}$ is upper bounded by $1$ if $0\leq\rho\leq \sqrt{\frac{p-1}{q-1}}$ for any $0<p<q\leq\infty$. For matrix spaces, King~\cite{King2003} proved a hypercontractive inequality for all unital channels, which immediately implies hypercontractivity of the noise operators $\Delta_{\rho}$ in $\M_2$.  However, his proof can not be extended to higher dimensions. In this section, we first provide a hypercontractive inequality for $\Delta_{\rho}$ in an arbitrary dimension. Then we introduce a noise operator acting on $L^2\br{\M^{\otimes h}_m,\gamma_n}$, which is a hybrid of Ornstein-Uhlenbeck operators $U_{\rho}$ and noise operators $\Delta_{\rho}$. Finally, we establish a hypercontractive inequality for random operators.
	
	\begin{lemma}\label{lem:multihypercontractivity}
		For any integer $m\geq 2$ and $0\leq\rho\leq\sqrt{\frac{1}{3\sqrt{m}}}$, let $\Delta_\rho:\M_m\rightarrow\M_m$ be the noise operator defined in \cref{def:bonamibeckner}. For any integer $n\geq1$, it holds that$$\nnorm{\Delta_\rho^{\otimes n}}_{2\rightarrow4}\leq1.$$
	\end{lemma}
	The proof is via induction on $n$. The following lemma is for the base case.
	
	\begin{lemma}\label{lem:singlehypercontractivity}
		Let $m>1, n\geq 1$ be integers and $M\in\M_m$. For any $0\leq\rho\leq\sqrt{\frac{1}{3\sqrt{m}}}$, it holds that
		$$\nnorm{\Delta_\rho(M)}_{4}\leq\nnorm{M}_{2}.$$
	\end{lemma}
	\begin{proof}
		From \cite[Theorem 1]{10.5555/2011608.2011614}, it suffices to prove the lemma for $M\in\H_m$. Note that $\ell_p$-norms are invariant under unitary transformations. The noise operator $\Delta_{\rho}$ commutes with any unitary operation. Namely, $\Delta_{\rho}\br{UMU^{\dagger}}=U\br{\Delta_{\rho}\br{M}}U^{\dagger}$ for any unitary $U$. Thus we may assume that $M=\mathrm{Diag}\br{d_1,\ldots, d_m}$ is diagonal without loss of generality. Note that the set of  $m\times m$ diagonal matrices forms an $m$-dimensional Hilbert space, denoted by $\D$. Consider the Hilbert space $\F=\br{\set{f:[m]\rightarrow\reals},\innerproduct{\cdot}{\cdot}}$, where $\innerproduct{f}{g}=\frac{1}{m}\sum_{i=1}^mf(i)g(i)$, which is isomorphic to $\D$. The noise operator $\Delta_{\rho}$ in $\D$ is isomorphic to the operator $\Delta'_{\rho}$ in $\F$, where $\Delta'_{\rho}f\br{x}=\expec{\mathbf{y}\sim_{\rho}x}{f(\mathbf{y})}$ and $\mathbf{y}\sim_{\rho}x$ represents that
\[\mathbf{y}=\begin{cases}
    x & \mbox{ with probability $\rho$};  \\
    \mbox{drawn from $[m]$ uniformly}, & \mbox{otherwise}.
  \end{cases}\]
It is known in~\cite[Corollary 3.1]{PawelWolff2007}\cite[Page 288, Corollary 10.20]{Odonnell08} that $\norm{\Delta'_{\rho}f}_4\leq\twonorm{f}$ for $0\leq\rho\leq\sqrt{\frac{1}{3\sqrt{m}}}$. From the isomorphism between $\D$ and $\F$, we conclude the result.
	\end{proof}
	
	%

Note that $\Delta_{\rho}^{\otimes n}$ only commutes with the product-unitary operations. But the operators $M\in\M_m^{\otimes n}$ that achieve $\nnorm{\Delta_{\rho}^{\otimes n}}_{2\rightarrow 4}$ are not necessarily product states, which thus cannot be diagonalized by product-unitary operators in general. Hence, we cannot apply the hypercontractivity on function spaces directly as we did in \cref{lem:singlehypercontractivity}. Instead, we use the same strategy as that in most of the proofs of hypercontractivity on function spaces by the induction on $n$.

	\begin{lemma}\label{lem:norminequality}
		Let $k>0$ be an integer and
		$$M=\begin{pmatrix}
		M_{11}&\cdots&M_{1m}\\
		\vdots&\ddots&\vdots\\
		M_{m1}&\cdots&M_{mm}
		\end{pmatrix}$$
		
		where $M_{ij}\in\M_k$ for $i,j\in[m]$. And let
		$$M'=\begin{pmatrix}
		\norm{M_{11}}_p&\cdots&\norm{M_{1m}}_p\\
		\vdots&\ddots&\vdots\\
		\norm{M_{m1}}_p&\cdots&\norm{M_{mm}}_p
		\end{pmatrix}$$

		Then it holds that
		\begin{enumerate}
			\item $\norm{M}_p\leq\norm{M'}_p$ if $p=4$;
			\item $\norm{M}_p=\norm{M'}_p$ if $p=2$.
		\end{enumerate}
	\end{lemma}
	\begin{proof}
		\begin{align*}
		\norm{M}_4^4&=\Tr M^{\dagger}M M^{\dagger}M\\
		&=\Tr~\sum_{ijkl}\br{M_{ji}}^{\dagger}M_{jk}\br{M_{lk}}^{\dagger}M_{li}\\
		&\leq\sum_{ijkl}\abs{\Tr \br{\br{M_{ji}}^{\dagger}M_{jk}\br{M_{lk}}^{\dagger}M_{li}}}\\
		&\leq\sum_{ijkl}\norm{\br{M_{ji}}^{\dagger}}_4\norm{M_{jk}}_4\norm{\br{M_{lk}}^{\dagger}}_4\norm{M_{li}}_4\\
		&=\Tr~\br{M'}^{\dagger}M'\br{M'}^{\dagger}M'\\
		&=\norm{M'}_4^4.
		\end{align*}
		The last inequality is by
\begin{eqnarray*}
  &&\abs{\Tr~X_1X_2X_3X_4}\\
  &\leq&\twonorm{X_1X_2}\twonorm{X_3X_4}\\
  &=&\sqrt{\Tr~\br{X_2X_2^{\dagger}X_1^{\dagger}X_1}}\cdot\sqrt{\Tr~\br{X_4X_4^{\dagger}X_3^{\dagger}X_3}}\\
  &\leq&\norm{X_1}_4\norm{X_2}_4\norm{X_3}_4\norm{X_4}_4,
\end{eqnarray*}
where both inequalities follow from Cauchy-Schwarz inequality.

For item 2,		
		\begin{align*}
		\norm{M}_2^2&=\Tr M^{\dagger}M=\Tr~\sum_{ij}\br{M_{ji}}^{\dagger}M_{ji}=\sum_{ij}\norm{M_{ji}}_2^2=\norm{M'}_2^2.
		\end{align*}
	\end{proof}
	
	\begin{lemma}\label{lem:tensorinequality}
		For any integer $n\geq 1$, it holds that
		$$\normsub{\Delta_\rho\otimes\Delta_\rho^{\otimes n}}{2\rightarrow4}\leq\normsub{\Delta_\rho^{\otimes n}}{2\rightarrow4}\cdot \normsub{\Delta_\rho}{2\rightarrow4}$$
	\end{lemma}
	
	\begin{proof}
		Given $A\in\M_{m^{n+1}}$, it can be written as an $m\times m$ block matrix:
		$$A=\begin{pmatrix}
		A_{11}&\cdots&A_{1m}\\
		\vdots&\ddots&\vdots\\
		A_{m1}&\cdots&A_{mm}
		\end{pmatrix},$$
		where $A_{ij}\in\M_{m^n}$ for $1\leq i,j\leq m$. It is sufficient to prove that
		$$\norm{\fn{(\Delta_\rho\otimes\Delta_\rho^{\otimes n})}{A}}_4\leq\normsub{\Delta_\rho^{\otimes n}}{2\rightarrow4}\normsub{\Delta_\rho}{2\rightarrow4}\norm{A}_2$$
		
		Define $\displaystyle B_{ij}=\Delta_\rho^{\otimes n}\br{A_{ij}},C_{ii}=\rho B_{ii}+\frac{1-\rho}{m}\sum_{k=1}^mB_{kk}$ and $ C_{ij}=\rho B_{ij}$ for $i\ne j$. Then
		\begin{align*}
		\fn{(\Delta_\rho\otimes\Delta_\rho^{\otimes n})}{A}&=(\Delta_\rho\otimes\id_{d^n})
		\begin{pmatrix}
		B_{11}&\cdots&B_{1m}\\
		\vdots&\ddots&\vdots\\
		B_{m1}&\cdots&B_{mm}
		\end{pmatrix}\\
		&=\begin{pmatrix}
		C_{11}&\cdots&C_{1m}\\
		\vdots&\ddots&\vdots\\
		C_{m1}&\cdots&C_{mm}
		\end{pmatrix}.
		\end{align*}
		
		By \cref{lem:norminequality}$$\norm{\fn{(\Delta_\rho\otimes\Delta_\rho^{\otimes n})}{A}}_4\leq\norm{
			\begin{pmatrix}
			\norm{C_{11}}_4&\cdots&\norm{C_{1m}}_4\\
			\vdots&\ddots&\vdots\\
			\norm{C_{m1}}_4&\cdots&\norm{C_{mm}}_4
			\end{pmatrix}.
		}_4$$
		
		Let $ d_i=\rho\norm{B_{ii}}_4+\frac{1-\rho}{m}\sum_{k=1}^m\norm{B_{kk}}_4.$ By the triangle inequality, $\norm{C_{ii}}_4\leq d_i$.
		
		Hence
		\begin{align*}
		&\norm{\fn{(\Delta_\rho\otimes\Delta_\rho^{\otimes n})}{A}}_4\\
		\leq&\norm{
			\begin{pmatrix}
			d_1&\cdots&\norm{C_{1m}}_4\\
			\vdots&\ddots&\vdots\\
			\norm{C_{m1}}_4&\cdots& d_m
			\end{pmatrix}
		}_4\\
		=&\norm{\Delta_\rho
			\begin{pmatrix}
			\norm{B_{11}}_4&\cdots&\norm{B_{1m}}_4~\hspace{3mm}\\
			\vdots&\ddots&\vdots\\
			\norm{B_{m1}}_4&\cdots&\norm{B_{mm}}_4
			\end{pmatrix}
		}_4\mbox{(By the definition of $\Delta_{\rho}$)}\\
		\leq&\normsub{\Delta_\rho}{2\rightarrow4}\norm{
			\begin{pmatrix}
			\norm{B_{11}}_4&\cdots&\norm{B_{1m}}_4\\
			\vdots&\ddots&\vdots\\
			\norm{B_{m1}}_4&\cdots&\norm{B_{mm}}_4
			\end{pmatrix}
		}_2\\
		\leq&\normsub{\Delta_\rho^{\otimes n}}{2\rightarrow4}\normsub{\Delta_\rho}{2\rightarrow4}\norm{
			\begin{pmatrix}
			\norm{A_{11}}_2&\cdots&\norm{A_{1m}}_2\\
			\vdots&\ddots&\vdots\\
			\norm{A_{m1}}_2&\cdots&\norm{A_{mm}}_2
			\end{pmatrix}
		}_2~\hspace{3mm}\mbox{(by induction)}\\
		=&\normsub{\Delta_\rho^{\otimes n}}{2\rightarrow4}\normsub{\Delta_\rho}{2\rightarrow4}\norm{A}_2\quad\quad\mbox{(\cref{lem:norminequality} item 2)}
		\end{align*}
	\end{proof}
	Now we are ready to prove \cref{lem:multihypercontractivity}.
	\begin{proof}[Proof of \cref{lem:multihypercontractivity}]
		\begin{align*}
		&\nnorm{\Delta_\rho^{\otimes n}}_{2\rightarrow4}\\
		&=m^{\frac{n}{4}}\normsub{\Delta_\rho^{\otimes n}}{2\rightarrow4}\\
		&\leq \left(m^{\frac{1}{4}}\normsub{\Delta_\rho}{2\rightarrow4}\right)^n\quad\quad\mbox{(\cref{lem:tensorinequality})}\\
		&=\nnorm{\Delta_{\rho}}_{2\rightarrow 4}^n\\
		&\leq1\quad\quad\mbox{(\cref{lem:singlehypercontractivity})}
		\end{align*}
	\end{proof}

The following fact is a well known hypercontractive inequality on Gaussian spaces.
\begin{fact}\footnote{The results in~\cite{PawelWolff2007,MosselOdonnell:2010} are for $f\in L^2\br{\reals, \gamma_n}$. But they can be extended to $L^2\br{\complex,\gamma_n}$ using the same argument in the proof of \cref{lem:gausshyper}.}\label{fac:gaussianhypercontractivity}~\cite{PawelWolff2007}\cite[Page 332, Theorem 11.23]{Odonnell08}
	For any $0\leq\rho\leq\frac{1}{\sqrt{3}}$, $f\in L^2\br{\complex, \gamma_n}$, it holds that
	\[\norm{U_{\rho}f}_4\leq\twonorm{f},\]
	where $U_{\rho}$ is an Ornstein-Uhlenbeck operator given in \cref{def:ornstein}.
\end{fact}

We need to generalize \cref{fac:gaussianhypercontractivity} for multiple functions for technical reasons.

\begin{lemma}\label{lem:gausshyper}
	Given functions $p_1,\ldots p_n\in L^2\br{\complex,\gamma_n}$, it holds that
	\[\br{\expec{\mathbf{x}\sim \gamma_n}{\br{\sum_{i=1}^n \abs{\br{U_{\rho}p_i}\br{\mathbf{x}}}^2}^2}}^{\frac{1}{4}}\leq\br{\expec{\mathbf{x}\sim \gamma_n}{\sum_{i=1}^n\abs{p_i\br{\mathbf{x}}}^2}}^{\frac{1}{2}}.\]
\end{lemma}

\begin{proof}
	Let $q_i= U_{\rho}p_i$. Then
	\begin{eqnarray*}
		&&\br{\expec{\mathbf{x}\sim \gamma_n}{\br{\sum_{i=1}^n\abs{\br{U_{\rho}p_i}\br{\mathbf{x}}}^2}^2}}^{\frac{1}{4}}\\
		&=&\br{\sum_{i=1}^n\expec{\mathbf{x}\sim \gamma_n}{\abs{q_i\br{\mathbf{x}}}^4}+\sum_{i\neq j}\expec{\mathbf{x}}{\abs{q_i\br{\mathbf{x}}^2q_j\br{\mathbf{x}}^2}}}^{\frac{1}{4}}\\
		&\leq&\br{\sum_{i=1}^n\norm{q_i}_4^4+\sum_{i\neq j}\norm{q_i}_4^2\norm{q_j}_4^2}^{\frac{1}{4}}\quad\quad\mbox{(Cauchy-Schwarz inequality)}\\
		&\leq&\br{\sum_{i=1}^n\twonorm{p_i}^4+\sum_{i\neq j}\twonorm{p_i}^2\twonorm{p_j}^2}^{\frac{1}{4}}\quad\quad\mbox{(\cref{fac:gaussianhypercontractivity})}\\
		&=&\br{\sum_i\twonorm{p_i}^2}^{\frac{1}{2}}\\
		&=&\br{\expec{\mathbf{x}\sim \gamma_n}{\sum_{i=1}^n\abs{p_i\br{\mathbf{x}}}^2}}^{\frac{1}{2}}.
	\end{eqnarray*}
\end{proof}

We introduce a noise operator $\Gamma_{\rho}$ acting on $L^2\br{\M_m^{\otimes h},\gamma_n}$, which is a hybrid of the Ornstein-Uhlenbeck operator $U_{\rho}$ in \cref{def:ornstein} and the noise operator $\Delta_{\rho}$ in \cref{def:bonamibeckner}.
	\begin{definition}\label{def:gamma}
	Given $0\leq\rho\leq 1$ and integers $h,n\geq 0, m\geq 2$, let $\mathbf{P}\in L^2\br{\M_m^{\otimes h},\gamma_n}$ with an expansion
 \begin{equation*}
	\mathbf{P}=\sum_{\sigma\in[m^2]_{\geq 0}^h}p_{\sigma}\br{\mathbf{g}}\B_{\sigma},
	\end{equation*}
	where $\set{\B_i}_{i=0}^{m^2-1}$ is a standard orthonormal basis in $\M_m$, $p_{\sigma}\in L^2\br{\complex,\gamma_n}$ and $\mathbf{g}\sim \gamma_n.$
 The noise operator $\Gamma_{\rho}:L^2\br{\M_m^{\otimes h},\gamma_n}\rightarrow L^2\br{\M_m^{\otimes h},\gamma_n}$ is defined to be	\[\Gamma_{\rho}\br{\mathbf{P}}=\sum_{\sigma\in[m^2]_{\geq 0}^h}\br{U_{\rho}p_{\sigma}}\br{\mathbf{g}}\Delta_{\rho}\br{\B_{\sigma}},\]
	where $\set{\B_i}_{i=0}^{m^2-1}$ is a standard orthonormal basis in $\M_m$.
	\begin{remark}
	  The definition is independent of the choice of bases. For simplicity, let's assume that $h=1$. Let $\set{\A_i}_{i\in[m^2]_{\geq 0}}$ be another standard orthonormal basis in $\M_m$ and $\mathbf{P}=\sum_{\sigma\in[m^2]_{\geq 0}}q_{\sigma}\br{\mathbf{g}}\A_{\sigma}$. From \cref{fac:unitarybasis}, there exists an orthogonal matrix $V=\br{V_{ij}}_{0\leq i,j\leq m^2-1}$ with $V_{0,i}=V_{i,0}=\delta_{0,i}$ satisfying that $\A_{\sigma}=\sum_{\sigma'=0}^{m^2-1}V_{\sigma,\sigma'}\B_{\sigma'}$ for $\sigma\in[m^2]_{\geq 0}$. Hence,  $p_{\sigma'}\br{\mathbf{g}}=\sum_{\sigma=0}^{m^2-1}V_{\sigma,\sigma'}q_{\sigma}\br{\mathbf{g}}$ for $\sigma'\in[m^2]_{\geq 0}$. Then
\begin{align*}
\sum_{\sigma\in[m^2]_{\geq 0}}\br{U_{\rho}q_{\sigma}}\br{\mathbf{g}}\Delta_{\rho}\br{\A_{\sigma}}&=\sum_{\sigma,\sigma'\in[m^2]_{\geq 0}}V_{\sigma,\sigma'}\br{U_{\rho}q_{\sigma}}\br{\mathbf{g}}\Delta_{\rho}\br{\B_{\sigma'}}\\
&=\sum_{\sigma'\in[m^2]_{\geq 0}}\br{U_{\rho}p_{\sigma'}}\br{\mathbf{g}}\Delta_{\rho}\br{\B_{\sigma'}},
\end{align*}
where the second equality follows from the linearity of Ornstein-Uhlenbeck operators.
	
	\end{remark}
	
\end{definition}

Recall that $\abs{\sigma}=\abs{\set{i:\sigma_i\neq 0}}$ and $\wt{\tau}=\sum_i\tau_i$. The lemma below directly follows from  \cref{fac:vecfun} and \cref{lem:bonamibecknerdef} item 1.

\begin{lemma}\label{lem:gammaoperator}
	Given $0\leq\rho\leq 1$, integers $n,h>0, m>1$ and a random operator $\mathbf{P}\in L^2\br{\M_m^{\otimes h},\gamma_n}$ with an expansion
\begin{equation*}
	\mathbf{P}=\sum_{\sigma\in[m^2]_{\geq 0}^h}p_{\sigma}\br{\mathbf{g}}\B_{\sigma},
\end{equation*}
	where $\set{\B_i}_{i=0}^{m^2-1}$ is a standard orthonormal basis in $\M_m$, $p_{\sigma}\in L^2\br{\complex,\gamma_n}$ and $\mathbf{g}\sim \gamma_n,$ it holds that
	\begin{equation}\label{eqn:gamma}
	\Gamma_{\rho}\br{\mathbf{P}}=\sum_{\sigma\in[m^2]_{\geq0}^h}\sum_{\tau\in\mathbb{Z}_{\geq 0}^n}\rho^{\abs{\sigma}+\wt{\tau}}\widehat{p_{\sigma}}\br{\tau}H_{\tau}\br{\mathbf{g}}\B_{\sigma}.
	\end{equation}
	where $H_{\tau}$'s are the Hermite polynomials defined in Eqs.~\cref{eqn:hermitebasis}\cref{eqn:hermite} and\\
$p_{\sigma}\br{\cdot}=\sum_{\tau\in\mathbb{Z}_{\geq 0}^n}\widehat{p_{\sigma}}\br{\tau}H_{\tau}\br{\cdot}$ for $\sigma\in[m^2]_{\geq 0}^h$.
\end{lemma}

The main result in this section is the hypercontractivity of random operators, which is stated as follows.

	\begin{lemma}\label{lem:hypercontractivity}
		Given integers $n,h>0,m>1$ and $0\leq\rho\leq \frac{1}{\sqrt{3\sqrt{m}}}$, for any random operator $\mathbf{P}\in L^2\br{\M_m^{\otimes h},\gamma_n}$, it holds that
		\[N_{4}\br{\Gamma_{\rho}\br{\mathbf{P}}}\leq N_2\br{\mathbf{P}},\]
		where $\Gamma_{\rho}$ is a noise operator acting on $L^2\br{\M_m^{\otimes h},\gamma_n}$ defined in \cref{def:gamma} and $N_p$ is a normalized $p$-norm of a random operator in \cref{def:randop}.
	\end{lemma}	
	
	\begin{proof}
		Let $\mathbf{P}=\sum_{\sigma\in[m^2]^h_{\geq 0}}p_{\sigma}\br{\mathbf{g}}\B_{\sigma}$, where
		$\set{\B_i}_{i=0}^{m^2-1}$ is a standard orthonormal basis.
		Set $\mathbf{Q}=\sum_{\sigma\in[m^2]^h_{\geq 0}}\br{U_{\rho}p_{\sigma}}\br{\mathbf{g}}\B_{\sigma}$. Then by the definition of $\Gamma_{\rho}$,
		\[\Gamma_{\rho}\br{\mathbf{P}}=\Delta_{\rho}\br{\mathbf{Q}}.\]
		\noindent Using \cref{lem:multihypercontractivity},
		\begin{equation}\label{eqn:np}
 N_4\br{\Gamma_{\rho}\br{\mathbf{P}}}=\br{\expec{}{\nnorm{\Delta_{\rho}\br{\mathbf{Q}}}_4^4}}^{\frac{1}{4}}\leq\br{\expec{}{\nnorm{\mathbf{Q}}_2^4}}^{\frac{1}{4}}.
		\end{equation}
		Let $p_{ij}\in L^2\br{\complex,\gamma_n}$ and $q_{ij}\in L^2\br{\complex,\gamma_n}$ be the entries of $\mathbf{P}$ and $\mathbf{Q}$, respectively, for $1\leq i,j\leq m^h$. Then $q_{ij}=U_{\rho}p_{ij}$.
		Notice that
\begin{align}
\br{\expec{}{\nnorm{\mathbf{Q}}_2^4}}^{\frac{1}{4}}&=\frac{1}{m^{h/2}}\br{\expec{\mathbf{x}\sim \gamma_n}{\br{\sum_{ij}\abs{q_{ij}\br{\mathbf{x}}}^2}^2}}^{\frac{1}{4}}\label{eqn:Q}\\
&\leq\frac{1}{m^{h/2}}\br{\expec{\mathbf{x}\sim \gamma_n}{\sum_{ij}\abs{p_{ij}\br{\mathbf{x}}}^2}}^{\frac{1}{2}}\nonumber\\&=N_2\br{\mathbf{P}}\nonumber,
		\end{align}
		where the inequality follows from \cref{lem:gausshyper}. Combining Eqs.~\cref{eqn:np}~\cref{eqn:Q}, we conclude the result.
		
	\end{proof}
	
	The following is an application of \cref{lem:hypercontractivity}.
	\begin{lemma}\label{lem:Xhypercontractivity}
		Given integers $h,n\geq 0$, for any multilinear random operator $\mathbf{P}\in L^2\br{\M_m^{\otimes h},\gamma_n}$,	it holds that
		\[N_4\br{\mathbf{P}}\leq 3^{d/2}m^{d/4}N_2\br{\mathbf{P}},\]
		where $d=\max_{\sigma\in[m^2]_{\geq 0}^h}\br{\deg\br{p_{\sigma}}+\abs{\sigma}}$.
	\end{lemma}
	\begin{proof}
		Suppose
		$\mathbf{P}=\sum_{\sigma\in[m^2]_{\geq 0}^h}p_{\sigma}\br{\mathbf{g}}\B_{\sigma},$
		where $\set{\B_i}_{i=0}^{m^2-1}$ is a standard orthonormal basis in $\M_m$, $p_{\sigma}\in L^2\br{\complex,\gamma_n}$ is multilinear and $\mathbf{g}\sim \gamma_n$. Set $$\mathbf{P}^{=i}=\sum_{\br{\sigma,\tau}\in[m^2]_{\geq0}^h\times\mathbb{Z}_{\geq 0}^n:\atop\abs{\sigma}+\wt{\tau}=i}\widehat{p_{\sigma}}\br{\tau}H_{\tau}\br{\mathbf{g}}\B_{\sigma}.$$
		Notice that $p_{\sigma}$'s are multilinear as $\mathbf{P}$ is multilinear. Thus $\abs{\tau}=\wt{\tau}\leq\deg\br{p_{\sigma}}$ whenever $\widehat{p_{\sigma}}\br{\tau}\neq 0$. Applying \cref{lem:gammaoperator} and \cref{lem:hypercontractivity},
		\begin{align*} N_4\br{\mathbf{P}}&= N_4\br{\Gamma_{\frac{1}{\sqrt{3\sqrt{m}}}}\br{\sum_{i=1}^{d}\br{\sqrt{3\sqrt{m}}}^i\mathbf{P}^{=i}}}\\
		&\leq N_2\br{\sum_{i=1}^{d}\br{\sqrt{3\sqrt{m}}}^i\mathbf{P}^{=i}}\\
		&=\br{\expec{}{\nnorm{\sum_{i=1}^{d}\br{\sqrt{3\sqrt{m}}}^i\mathbf{P}^{=i}}^2_2}}^{1/2}
		\end{align*}
		Note that
		\[\expec{}{\Tr~\br{\mathbf{P}^{=i}}^{\dagger}\mathbf{P}^{=j}}=0,\]
		whenever $i\neq j$.
		Therefore,
\begin{align*}
N_4\br{\mathbf{P}}&\leq\br{\sum_{i=1}^{d}\br{\sqrt{3\sqrt{m}}}^{2i}\expec{}{\nnorm{\mathbf{P}^{=i}}^2_2}}^{\frac{1}{2}}\\
&\leq3^{d/2}m^{d/4}\br{\sum_{i=1}^{d}\expec{}{\nnorm{\mathbf{P}^{=i}}^2_2}}^{\frac{1}{2}}\\
&=3^{d/2}m^{d/4}N_2\br{\mathbf{P}}.
\end{align*}
		
	\end{proof}

	\section{Reduction from POVMs to single operators}\label{sec:reduction}

Let's define the function $\zeta:\reals\rightarrow\reals$ as follows.

\begin{eqnarray}
	&&\zeta\br{x}=\begin{cases}x^2~&\mbox{if $x\leq 0$}\\ 0~&\mbox{otherwise}\end{cases}.\label{eqn:zeta}
\end{eqnarray}

\begin{lemma}\label{lem:closedelta1}
	Given an integer $m>0$, $M\in\H_m$, $\Delta=\set{X\in\H_m:X\geq0}$, let \[\R\br{M}=\arg\min\set{\twonorm{M-X}:X\in\Delta}\]
	 be a rounding map of $\Delta$ with respect to the distance $\twonorm{\cdot}$. It holds that
	\[\Tr~\zeta\br{M}=\twonorm{M-\R\br{M}}^2.\]
\end{lemma}
\begin{proof}
	Without loss of generality, we assume that $M$ is diagonal. Let $$X_0=\arg\min\set{\twonorm{M-X}:X\in\Delta}.$$
	The lemma is easy to verify if $X_0$ is also diagonal. We now show that $X_0$ is indeed diagonal. Note that	\[\twonorm{M-X_0}^2=\Tr~X_0^2+\Tr~M^2-2\sum_i\lambda_i\br{M}X_0\br{i,i},\]
where $X_0\br{i,i}$ is the $\br{i,i}$-th entry of $X_0$. It is known that $\br{\lambda_1\br{X},\ldots,\lambda_n\br{X}}$ majorizes $\br{X_0\br{1,1},\ldots, X_0\br{d,d}}$ by Schur's theorem ~\cite[Page 35, Exercise II.1.12]{Bhatia}. Namely, $\sum_{j=1}^i\lambda_j\br{X}\geq\sum_{j=1}^iX_0\br{j,j}$ for $1\leq i\leq d$. Note that $X_0\geq 0$. It is easy to verify that
	\[\sum_i\lambda_i\br{M}X_0\br{i,i}\leq\sum_i\lambda_i\br{M}\lambda_i\br{X}.\]
	The equality is achieved only if $X_0$ is also diagonal.

\end{proof}

The lemma below is the main result in this section, which states that for any $\vec{X}$ satisfying $\sum_i X_i=\id$, the $\ell_2$ distance between $\vec{X}$ and the set of sub-POVMs can be upper bounded in terms of $\sum_i\Tr~\zeta\br{X_i}$.
\begin{lemma}\label{lem:closedelta}
	Given $\vec{X}\in\br{\H_m^{\otimes n}}^t$ satisfying that $\sum_{i=1}^tX_i=\id$, define
	$$\R\br{\vec{X}}=\arg\min\set{\nnorm{\vec{X}-\vec{P}}_2^2:\vec{P}~\mbox{is a sub-POVM}}$$
	
	It holds that
	
	$$\nnorm{\R\br{\vec{X}}-\vec{X}}_2^2\leq\frac{3\br{t+1}}{m^n}\sum_{i=1}^t\Tr~\zeta(X_i)+6\br{\frac{t}{m^n}\sum_{i=1}^t\Tr~\zeta(X_i)}^{\frac{1}{2}}$$
\end{lemma}

\begin{proof}Recall the definitions of $\pos{X}$ in Eq.~ \cref{eqn:geq0def} and the Moore-Penrose inverse $X^+$ in Eq.~\cref{eqn:penrosedef}. Define
	$$\L\br{\vec{X}}=\br{\sqrtpinv{Y}\pos{X_1}\sqrtpinv{Y},\dots,\sqrtpinv{Y}\pos{X_t}\sqrtpinv{Y}}$$
where $Y=\sum_{i=1}^t\pos{X_i}$.
	
	Then by the definition of $\R\br{\vec{X}}$,
	$$\nnorm{\R\br{\vec{X}}-\vec{X}}_2^2\leq\nnorm{\L\br{\vec{X}}-\vec{X}}_2^2.$$
And	
	\begin{align*}
	&\nnorm{\L\br{\vec{X}}-\vec{X}}_2^2\\
	=&\sum_{i=1}^t\nnorm{\sqrtpinv{Y}\pos{X_i}\sqrtpinv{Y}-X_i}_2^2\\
	\leq&3\sum_{i=1}^t\left(\nnorm{\sqrtpinv{Y}\pos{X_i}\sqrtpinv{Y}-\sqrtpinv{Y}\pos{X_i}}_2^2\right.\\
&\left.+\nnorm{\sqrtpinv{Y}\pos{X_i}-\pos{X_i}}_2^2+\nnorm{\pos{X_i}-X_i}_2^2\right)
	\end{align*}
Note that
	\[\sum_{i=1}^t\nnorm{X_i-\pos{X_i}}_2^2=\sum_{i=1}^t\frac{1}{m^n}\Tr~\zeta(X_i)\]
Combining \cref{claim:term1} and \cref{claim:term2}, we have \[\nnorm{\R\br{\vec{X}}-\vec{X}}_2^2\leq\frac{3(t+1)}{m^n}\sum_{i=1}^t\Tr~\zeta(X_i)+6\br{\frac{t}{m^n}\sum_{i=1}^t\Tr~\zeta(X_i)}^{\frac{1}{2}}\]
\end{proof}

\begin{claim}\label{claim:term1}
\[\sum_{i=1}^t\nnorm{\sqrtpinv{Y}\pos{X_i}\sqrtpinv{Y}-\sqrtpinv{Y}\pos{X_i}}_2^2\leq\br{t\sum_{i=1}^t\nnorm{X_i-\pos{X_i}}_2^2}^{\frac{1}{2}}.\]
\end{claim}
\begin{claim}\label{claim:term2}
\begin{multline*}
\sum_{i=1}^t\nnorm{\sqrtpinv{Y}\pos{X_i}-\pos{X_i}}_2^2\\\leq\br{1+\br{t\sum_{i=1}^t\nnorm{X_i-\pos{X_i}}_2^2}^{\frac{1}{2}}}\br{t\sum_{i=1}^t\nnorm{X_i-\pos{X_i}}_2^2}^{\frac{1}{2}}.\end{multline*}
\end{claim}
Before proving these two claims, we need the following lemma, whose proof is deferred to the end of the section.
\begin{lemma}\label{lem:pinvinequality}
Given Hermitian matrices $A$ and $B$ the following holds.
\begin{enumerate}
  \item If $A\geq B\geq0$, then $B\pinv{A}B\leq B$.
  \item If $A\geq0$, then $(\id-A)^2\leq\abs{\id-A^2}$.
\end{enumerate}
\end{lemma}

\begin{proof}[Proof of \cref{claim:term1}]
Let $\Pi$ be a projector onto the support of $Y$.  We have
  \begin{align*}
	&\sum_{i=1}^t\nnorm{\sqrtpinv{Y}\pos{X_i}\sqrtpinv{Y}-\sqrtpinv{Y}\pos{X_i}}_2^2\\
	&=\sum_{i=1}^t\nnorm{\sqrtpinv{Y}\pos{X_i}\br{\id-\sqrtpinv{Y}}}_2^2\\
	&=\sum_{i=1}^t\frac{1}{m^n}\Tr~\br{\id-\sqrtpinv{Y}}\pos{X_i}\pinv{Y}\pos{X_i}\br{\id-\sqrtpinv{Y}}\\
	&\leq\sum_{i=1}^t\frac{1}{m^n}\Tr~\br{\id-\sqrtpinv{Y}}\pos{X_i}\br{\id-\sqrtpinv{Y}}\quad\quad\mbox{(\cref{lem:pinvinequality} item 1)}\\
	&=\frac{1}{m^n}\Tr~\br{\id-\sqrtpinv{Y}}Y\br{\id-\sqrtpinv{Y}}\mbox{\quad(by definition  $Y=\sum_{i=1}^t\pos{X_i}$)}\\
	&=\frac{1}{m^n}\Tr~Y^{\frac{1}{2}}\br{\id-\sqrtpinv{Y}}^2Y^{\frac{1}{2}}\\
	&\leq\frac{1}{m^n}\Tr Y^{\frac{1}{2}}\abs{\id-\pinv{Y}}Y^{\frac{1}{2}}\quad\quad\mbox{(\cref{lem:pinvinequality} item 2)}\\
&=\frac{1}{m^n}\Tr \abs{Y-\Pi}\\
	&\leq\frac{1}{m^n}\Tr \abs{Y-\id}\\
	&\leq\nnorm{\id-Y}_2\\
	&=\nnorm{\sum_{i=1}^tX_i-\sum_{i=1}^t\pos{X_i}}_2\\
	&\leq \sum_{i=1}^t\nnorm{X_i-\pos{X_i}}_2\\
	&\leq \br{t\sum_{i=1}^t\nnorm{X_i-\pos{X_i}}_2^2}^{\frac{1}{2}}
	\end{align*}
\end{proof}

\begin{proof}[Proof of \cref{claim:term2}]
  Let $\Pi$ be a projector onto the support of $Y$. Note that $0\leq\pos{X_i}\leq Y$ for all $i$. We have $\Pi\pos{X_i}=\pos{X_i}$. Thus	
	\begin{align}
	&\sum_{i=1}^t\nnorm{\sqrtpinv{Y}\pos{X_i}-\pos{X_i}}_2^2\nonumber\\
	=&\frac{1}{m^n}\sum_{i=1}^t\Tr \br{\pos{X_i}}^2\br{\Pi-\sqrtpinv{Y}}^2\nonumber\\
	\leq&\frac{1}{m^n}\sum_{i=1}^t\Tr \br{\pos{X_i}}^2\abs{\Pi-\pinv{Y}}\quad\quad\mbox{(\cref{lem:pinvinequality} item 2)}\nonumber\\
	=&\frac{1}{m^n}\sum_{i=1}^t\Tr \br{\pos{X_i}}^2\sqrtpinv{Y}\abs{\Pi-Y}\sqrtpinv{Y}\nonumber\\
    =&\frac{1}{m^n}\sum_{i=1}^t\Tr \sqrtpinv{Y}\br{\pos{X_i}}^2\sqrtpinv{Y}\abs{\Pi-Y}.\label{eqn:YXX}
    \end{align}
  Set
  \begin{eqnarray*}
    &D=\begin{pmatrix}
        \pos{X_1}&& \\
        &\ddots&\\
        &&\pos{X_t}
      \end{pmatrix}~\mbox{and}~ V=\begin{pmatrix}
        \br{\pos{X_1}}^{\frac{1}{2}} \sqrtpinv{Y}\\
        \vdots\\
        \br{\pos{X_t}}^{\frac{1}{2}} \sqrtpinv{Y}
      \end{pmatrix}\label{eqn:W}.
  \end{eqnarray*}
    From Eq. \cref{eqn:YXX},

    \begin{align}
    &\sum_{i=1}^t\nnorm{\sqrtpinv{Y}\pos{X_i}-\pos{X_i}}_2^2\nonumber\\
    \leq&\frac{1}{m^n}\Tr~ \br{V^{\dagger}DV\abs{\Pi-Y}}\nonumber\\
    \leq&\frac{1}{m^n}\twonorm{V^{\dagger}DV}\cdot\norm{\Pi-Y}_2\nonumber\\
    \leq&\frac{1}{m^n}\norm{V}^2\cdot\norm{D}_2\cdot\norm{\Pi-Y}_2,\label{eqn:YXX2}
    \end{align}
    where the second inequality follows from the fact that $\norm{ABC}_2\leq\norm{A}\norm{B}_2\norm{C}$.

Note that
    \[\norm{V}^2=\norm{V^{\dagger}V}=\norm{\sqrtpinv{Y}\br{\sum_i\pos{X_i}}\sqrtpinv{Y}}=1,\]
    and
    \[\twonorm{D}=\br{\sum_{i=1}^t\Tr~\br{\pos{X_i}}^2}^{1/2}\leq\br{\Tr~\br{\sum_{i=1}^t\pos{X_i}}^2}^{1/2}=\twonorm{Y}.\]
    Then from Eq. \cref{eqn:YXX2},
    \begin{align*}
    &\sum_{i=1}^t\nnorm{\sqrtpinv{Y}\pos{X_i}-\pos{X_i}}_2^2\nonumber\\
    \leq&\frac{1}{m^n}\twonorm{Y}\cdot\twonorm{\Pi-Y}\\
    =&\nnorm{Y}_2\nnorm{\id-Y}_2\\
	\leq&\br{\nnorm{\id-Y}_2+\nnorm{\id}_2}\nnorm{\id-Y}_2\\ \leq&\br{1+\br{t\sum_{i=1}^t\nnorm{X_i-\pos{X_i}}_2^2}^{\frac{1}{2}}}\br{t\sum_{i=1}^t\nnorm{X_i-\pos{X_i}}_2^2}^{\frac{1}{2}}.
	\end{align*}
We conclude the result.
\end{proof}

	\begin{proof}[Proof of \cref{lem:pinvinequality}]
		For item 1, without loss of generality, we may asume that $A$ is diagonal with the following form
		$$A=\begin{pmatrix}
		\Lambda&\mathbf{0}\\
		\mathbf{0}&\mathbf{0}
		\end{pmatrix}$$
		where $\Lambda>0$ is diagonal.
		
		Because $A\geq B\geq 0$, $B$ must have the same block structure. Namely,
		$$B=\begin{pmatrix}
		\Lambda'&\mathbf{0}\\
		\mathbf{0}&\mathbf{0}
		\end{pmatrix}$$
		where $\Lambda\geq\Lambda'\geq 0$.
		If $\Lambda'$ is invertible, then $\Lambda^{-1}\leq\br{\Lambda'}^{-1}$ by~\cite[Proposition V.1.6]{Bhatia}. Thus $\Lambda'\Lambda^{-1}\Lambda'\leq\Lambda'$. By continuity, it also holds when $\Lambda'$ is not invertible. Thus
\[B\pinv{A}B=\begin{pmatrix}
		\Lambda'\Lambda^{-1}\Lambda'&\mathbf{0}\\
		\mathbf{0}&\mathbf{0}
		\end{pmatrix}\leq\begin{pmatrix}
		\Lambda'&\mathbf{0}\\
		\mathbf{0}&\mathbf{0}
		\end{pmatrix}=B.\]

	For item 2, we again assume without loss of generality that $A$ is diagonal of the following form
		$$A=\begin{pmatrix}
	\Lambda&\mathbf{0}\\
	\mathbf{0}&\mathbf{0}
	\end{pmatrix}$$
	where $\Lambda>0$ is diagonal. Then
	\[\br{\id-A}^2=\begin{pmatrix}
	\br{\id-\Lambda}^2 & \mathbf{0}\\
	\mathbf{0} & \id
	\end{pmatrix},\]
	and
	\[\abs{\id-A^2}=\begin{pmatrix}
	\abs{\id-\Lambda^2} & \mathbf{0}\\
	\mathbf{0} & \id
	\end{pmatrix}.\]
	It is not hard to verify that
	$\br{\id-\Lambda}^2\leq\abs{\id-\Lambda^2}$. Thus the result follows.
\end{proof}


\section{Quantum invariance principle}\label{sec:invariance}

In this section, we prove a quantum invariance principle with respect to the function $\zeta\br{\cdot}$ given in Eq.~\eqref{eqn:zeta}. Before proving the main result, we need to investigate the analytical properties of $\zeta$.

\subsection{Analytical properties of $\zeta$}
We introduce the Taylor expansions of matrix functions, for which we adopt Fr\'echet derivatives. The Fr\'echet derivatives are derivatives defined on Banach spaces. In this paper, we only concern ourselves with Fr\'echet derivatives on matrix spaces. Readers may refer to~\cite{Coleman} for a more thorough treatment.

\begin{definition}\label{def:frechetderivative}
	Given a map $f:\M_m\rightarrow\M_m$ and $P, Q\in\M_m$, the Fr\'echet derivative of $f$ at $P$ with direction $Q$ is defined to be
	\[Df\br{P}\Br{Q}=\frac{d}{dt}f\br{P+tQ}|_{t=0}.\]
	The $k$-th order Fr\'echet derivative of $f$ at $P$ with direction $\br{Q_1,\ldots, Q_k}$ is defined to be
	\[D^kf\br{P}\Br{Q_1,\ldots, Q_k}=\frac{d}{dt}\br{D^{k-1}f\br{P+tQ_k}\Br{Q_1,\ldots, Q_{k-1}}}|_{t=0}.\]
\end{definition}

Fr\'echet derivatives share many common properties with the derivatives in Euclidean spaces, such as linearity, composition rules, etc. The most basic properties are summarized in \cref{sec:frechet}.

Note that the function $\zeta$ is in $\C^1$ but not in $\C^2$. We define a $\C^2$-approximation of $\zeta$ in the following, whose Fr\'echet derivatives are easier to calculate comparing with the  $\C^{\infty}$-approximation considered in~\cite{Mossel:2010,MosselOdonnell:2010}.
For any $\lambda>0$, define $\zeta_{\lambda}:\reals\rightarrow\reals$ to be\footnote{The definition of $\zeta_{\lambda}$ is derived from the following construction.

		\[\psi\br{x}=\begin{cases}
		\frac{1}{2}~&\mbox{if $-1\leq x\leq 1$}\\
		0~&\mbox{otherwise.}
		\end{cases}\]
		\[\psi_{\lambda}\br{x}=\psi\br{x/\lambda}/\lambda\]
		\[\zeta_{\lambda}\br{x}= \zeta*\psi_{\lambda}\]
	}

\begin{eqnarray}
	&&\zeta_{\lambda}\br{x}=\begin{cases}
	x^2+\frac{1}{3}\lambda^2~&\mbox{if $x\leq-\lambda$}\\
	\frac{\br{\lambda-x}^3}{6\lambda}~&\mbox{if $-\lambda\leq x\leq \lambda$}\\
	0&\mbox{if $x\geq\lambda$}.
	\end{cases}\label{eqn:zetalambda}
	\end{eqnarray}

The following lemma can be verified by elementary calculus.

\begin{lemma}\label{lem:zeta}
	For any $\lambda>0$, it holds that
	\begin{enumerate}
		\item $\norm{\zeta_{\lambda}-\zeta}_{\infty}\leq \frac{\lambda^2}{2}.$
		\item $\zeta_{\lambda}\in\mathcal{C}^2$. $\zeta_{\lambda}''$ is a continuous piecewise linear continuous function. $\zeta'''_{\lambda}\br{\cdot}$ exists in $\reals$ except for finite number of points. And $\abs{\zeta'''_{\lambda}\br{x}}\leq \frac{1}{\lambda}$ for any $x$ that $\zeta'''_{\lambda}\br{x}$ exists. 	
	\end{enumerate}
\end{lemma}

\begin{lemma}\label{lem:zetataylor}
	For any Hermitian matrices $P, Q$ and $\lambda>0$, it holds that
	\begin{equation}\label{eqn:zetataylor}
	\Tr~\zeta_{\lambda}\br{P+ Q}=\Tr~\zeta_{\lambda}\br{P}+\Tr~ D\zeta_{\lambda}\br{P}\Br{Q}+\frac{1}{2}\Tr~D^2\zeta_{\lambda}\br{P}\Br{Q}+O\br{\frac{\twonorm{Q}\norm{Q}_4^2}{\lambda}}.
	\end{equation}
\end{lemma}
We prove Lemma~\ref{lem:zetataylor} by calculating each order of the Fr\'echet derivatives combining with several techniques in matrix analysis, which is deferred to \cref{sec:zetataylor}.

The following lemma enables us to remove the part of an operator with low $2$-norm without changing the value of $\Tr~\zeta\br{\cdot}$ much. The proof is also deferred to \cref{sec:zetataylor}.
\begin{lemma}\label{lem:zetaadditivity}
	For any Hermitian matrices $P$ and $Q$, it holds that \[\abs{\Tr~\br{\zeta\br{P+Q}-\zeta\br{P}}}\leq2\br{\twonorm{P}\twonorm{Q}+\twonorm{Q}^2}\].
\end{lemma}

\subsection{Quantum invariance principle for $\zeta\br{\cdot}$}

The following lemma is the main result in this section.

\begin{lemma}\label{lem:jointinvariance}
	Given $0<\tau,\delta,\rho<1$, integers $n>h\geq 0, d>0, m>1$, $H\subseteq[n]$ of size $\abs{H}=h$, a noisy MES state $\psi_{AB}$ with the maximal correlation $\rho=\rho\br{\psi_{AB}}$, there exists a map $f:\H_m^{\otimes n}\times\reals^{2(m^2-1)(n-h)}\rightarrow L^2\br{\H_m^{\otimes h},\gamma_{2(m^2-1)(n-h)}}$ such that the following holds.
	
	For any $P,Q\in\H_m^{\otimes n}, 0\leq P,Q\leq \id$, satisfying that $\influence_i\br{P}\leq\tau, \influence_i\br{Q}\leq\tau$ for all $i\notin H$ and $\nnorm{P^{>d}}_2^2\leq\delta,\nnorm{Q^{>d}}_2^2\leq\delta$,
\begin{multline*}
\br{\mathbf{P},\mathbf{Q}}=\br{f\br{P,\mathbf{g}},f\br{Q,\mathbf{h}}}_{\br{\mathbf{g},\mathbf{h}}\sim\G_{\rho}^{\otimes2(m^2-1)(n-h)}}\\\in L^2\br{\H_m^{\otimes h},\gamma_{2(m^2-1)(n-h)}}\times  L^2\br{\H_m^{\otimes h},\gamma_{2(m^2-1)(n-h)}}
\end{multline*}
	are  degree-$d$ multilinear joint random operators with the joint random variables drawn from $\G_{\rho}^{\otimes 2(m^2-1)(n-h)}$. And
	\begin{enumerate}
		\item $N_2\br{\mathbf{P}}\leq\nnorm{P}_2$ and $N_2\br{\mathbf{Q}}\leq\nnorm{Q}_2$;
		\item $\abs{\Tr~\br{P\otimes Q}\psi_{AB}^{\otimes n}-\expec{}{\Tr~\br{\br{\mathbf{P}\otimes\mathbf{Q}}\psi_{AB}^{\otimes h}}}}\leq\delta$.
		\item $\expec{}{\Tr~\zeta\br{\mathbf{P}}}\leq O\br{m^h\br{\br{3^dm^{d/2}\sqrt{\tau}d}^{2/3}+\sqrt{\delta}}}$ and

$\expec{}{\Tr~\zeta\br{\mathbf{Q}}}\leq O\br{m^h\br{\br{3^dm^{d/2}\sqrt{\tau}d}^{2/3}+\sqrt{\delta}}}$.
\item The map $f\br{\cdot,\mathbf{g}}:\H_m^{\otimes n}\rightarrow L^2\br{\H_m^{\otimes h},\gamma_{2\br{m^2-1}\br{n-h}}}$ is linear and unital.
	\end{enumerate}
\end{lemma}
The main difficulty is to prove item 3, which is obtained via several lemmas. The logical flow of the proof is illustrated in~\cref{fig:logicalflow}.

\newcommand{\tabspace}{8em}
\newcommand{\spaceone}{\hspace*{0.2em}}

\bigskip

\begin{figure}[h]
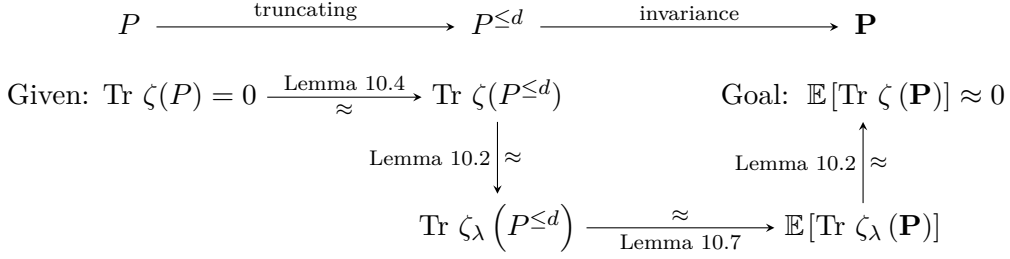
\label{fig:logicalflow}

 \begin{center}
\begin{codi}
\obj{
|(A)| P\spaceone&[\tabspace]|(B)| \spaceone P^{\leq d}\spaceone&[\tabspace]|(C)|\spaceone \mathbf{P}\\[-8mm]
|(D)| \text{Given: }\Tr~\zeta(P)=0&|(F)| \Tr~\zeta(P^{\leq d})&|(E)| \text{Goal: }\expec{}{\Tr~\zeta\br{\mathbf{P}}}\approx0\\
&|(G)| \Tr~\zeta_\lambda\br{P^{\leq d}}&|(H)| \expec{}{\Tr~\zeta_\lambda\br{\mathbf{P}}}\\
};
\mor A \mathrm{{truncating}}:-> B \mathrm{{invariance}}:-> C;
\mor D ["\approx", swap]["\text{{\cref{lem:zetaadditivity}}}"]:-> F ["\approx"]["\text{{\cref{lem:zeta}}}", swap]:-> G ["\approx"]["\text{{\cref{lem:hybrid}}}", swap]:-> H ["\approx", swap]["\text{{\cref{lem:zeta}}}"]:-> E;
\end{codi}
\end{center}

  \caption{Logical flow of the proof of \cref{lem:jointinvariance}}

\end{figure}

\bigskip

We define joint random variables $\set{\mathbf{g}_{i,j}}_{1\leq i\leq n, 0\leq j\leq m^2-1}$, where
$\g_{i,0}=1$ for all $i\in[n]$ and $\set{\mathbf{g}_{i,j}}_{1\leq i\leq n, 1\leq j\leq m^2-1}$ are independent identical distributions of $\gamma_1$.

	Given a standard orthonormal basis $\B=\set{\B_i}_{0\leq i\leq m^2-1}$, and $M\in\M_m^{\otimes n}$ with a Fourier expansion\[M=\sum_{\sigma\in[m^2]^n_{\geq0}}\widehat{M}\br{\sigma}\B_{\sigma}.\]

For any $0\leq i\leq n$, define the hybrid basis elements and the hybrid random operators as follows.
	\begin{align}
	&\X_{\sigma}^{\br{i}}=\prod_{j=1}^i\mathbf{g}_{j,\sigma_j }\B_{\sigma_{>i}}~\mbox{for $\sigma\in[m^2]_{\geq 0}^n$};\label{eqn:hybridxi}\\
	&\mathbf{M}^{\br{i}}=\sum_{\sigma\in[m^2]^n_{\geq0}}\widehat{M}\br{\sigma}\X^{\br{i}}_{\sigma}.\label{eqn:hybridmi}
	\end{align}
	
Note that $\mathbf{M}^{(i)}\in L^2\br{\M_m^{\otimes\br{n-i}},\gamma_{\br{m^2-1}i}}$ is a random operator of dimension $m^{n-i}$.
\begin{lemma}\label{lem:mg}
	$\mathbf{M}^{(i)}$ is independent of the choice of bases. Namely, for any standard orthonormal basis $\set{\A_i}_{i=0}^{m^2-1}$ in $\M_m$ and $M=\sum_{\sigma\in[m^2]_{\geq 0}^n}\lambda_{\sigma}\A_{\sigma}$, set $\mathbf{N}=\sum_{\sigma\in[m^2]_{\geq 0}^n}\lambda_{\sigma}\br{\prod_{j=1}^i\g_{j,\sigma_j}}\A_{\sigma>i}$. Then $\mathbf{N}$ and $\mathbf{M}^{\br{i}}$ have the same distribution.
\end{lemma}
\begin{proof}
	From \cref{fac:unitarybasis}, all orthonormal bases are equivalent up to orthonormal transformations. The lemma follows from the well known fact that the Gaussian distribution $\gamma_n$ is invariant under any orthonormal transformation.		
\end{proof}
\begin{lemma}\label{lem:hybrid}
	For any integers $n>0, m>1$ and $0\leq i\leq n-1$ and $M\in\H_m^{\otimes n}$, it holds that
	\begin{align*}
	\abs{\expec{}{\frac{1}{m^{n-i-1}}\Tr~\zeta_{\lambda}\br{\mathbf{M}^{\br{i+1}}}-\frac{1}{m^{n-i}}\Tr~\zeta_{\lambda}\br{\mathbf{M}^{\br{i}}}}}
	\leq O\br{\frac{m^{d/2}3^d}{\lambda}\influence_{i+1}\br{M}^{3/2}},
	\end{align*}
	where $d=\deg M$.
\end{lemma}
\begin{proof} Note that

	\begin{align*}
	&\mathbf{M}^{\br{i}}=\sum_{\sigma:\sigma_{i+1}=0}\widehat{M}\br{\sigma}\X_{\sigma}^{\br{i}}+\sum_{\sigma:\sigma_{i+1}\neq 0}\widehat{M}\br{\sigma}\X_{\sigma}^{\br{i}},\\
	&\mathbf{M}^{\br{i+1}}=\sum_{\sigma:\sigma_{i+1}= 0}\widehat{M}\br{\sigma}\X_{\sigma}^{\br{i+1}}+\sum_{\sigma:\sigma_{i+1}\neq 0}\widehat{M}\br{\sigma}\X_{\sigma}^{\br{i+1}},
	\end{align*}
	Set
	\begin{align*}
	&\mathbf{A}=\sum_{\sigma:\sigma_{i+1}=0}\widehat{M}\br{\sigma}\X_{\sigma}^{\br{i}}\\
	&\mathbf{B}=\sum_{\sigma:\sigma_{i+1}\neq 0}\widehat{M}\br{\sigma}\X_{\sigma}^{\br{i}},\\
	&\mathbf{C}=\sum_{\sigma:\sigma_{i+1}=0}\widehat{M}\br{\sigma}\X_{\sigma}^{\br{i+1}}\\
	&\mathbf{D}=\sum_{\sigma:\sigma_{i+1}\neq 0}\widehat{M}\br{\sigma}\X_{\sigma}^{\br{i+1}}.
	\end{align*}
	Then we have
	\begin{align*}
	\mathbf{M}^{\br{i}}=\mathbf{A}+\mathbf{B};~\mathbf{M}^{\br{i+1}}=\mathbf{C}+\mathbf{D}.
	\end{align*}
Notice that $\mathbf{A}=\id_m\otimes\mathbf{C}$, where $\id_m$ is placed in the $(i+1)$-th register. Thus

\begin{equation}\label{eqn:AC}
\Tr~\zeta_{\lambda}\br{\mathbf{A}}=m\cdot\Tr~\zeta_{\lambda}\br{\mathbf{C}}.
\end{equation}
From Eq.~\cref{eqn:AC} and \cref{lem:zetataylor},
	\begin{align*} &\abs{\expec{}{\frac{1}{m^{n-i-1}}\Tr~\zeta_{\lambda}\br{\mathbf{M}^{\br{i+1}}}-\frac{1}{m^{n-i}}\Tr~\zeta_{\lambda}\br{\mathbf{M}^{\br{i}}}}}\\
	\leq&\abs{\expec{}{\frac{1}{m^{n-i-1}}\br{\Tr~D\zeta_{\lambda}\br{\mathbf{C}}\Br{\mathbf{D}}+\frac{1}{2}\Tr~D^2\zeta_{\lambda}\br{\mathbf{C}}\Br{\mathbf{D}}}-\atop\frac{1}{m^{n-i}}\br{\Tr~D\zeta_{\lambda}\br{\mathbf{A}}\Br{\mathbf{B}}+\frac{1}{2}\Tr~D^2\zeta_{\lambda}\br{\mathbf{A}}\Br{\mathbf{B}}}}}\\
	&+O\br{\frac{1}{m^{n-i-1}}\expec{}{\frac{\twonorm{\mathbf{D}}\norm{\mathbf{D}}_4^2}{\lambda}}}+O\br{\frac{1}{m^{n-i}}\expec{}{\frac{\twonorm{\mathbf{B}}\norm{\mathbf{B}}_4^2}{\lambda}}}\quad\quad\mbox{(\cref{lem:zetataylor})}\\
	&\leq\abs{\expec{}{\frac{1}{m^{n-i-1}}\br{\Tr~D\zeta_{\lambda}\br{\mathbf{C}}\Br{\mathbf{D}}+\frac{1}{2}\Tr~D^2\zeta_{\lambda}\br{\mathbf{C}}\Br{\mathbf{D}}}-\atop\frac{1}{m^{n-i}}\br{\Tr~D\zeta_{\lambda}\br{\mathbf{A}}\Br{\mathbf{B}}+\frac{1}{2}\Tr~D^2\zeta_{\lambda}\br{\mathbf{A}}\Br{\mathbf{B}}}}}\\
	&+O\br{\frac{1}{\lambda}\br{N_2\br{\mathbf{D}}N_4\br{\mathbf{D}}^2+N_2\br{\mathbf{B}}N_4\br{\mathbf{B}}^2}}\quad\quad\mbox{(Cauchy-Schwarz inequality)}
\end{align*}

We show below in \cref{claim:1} that the first term is zero. Thus
\begin{align*}
&\abs{\expec{}{\frac{1}{m^{n-i-1}}\Tr~\zeta_{\lambda}\br{\mathbf{M}^{\br{i+1}}}-\frac{1}{m^{n-i}}\Tr~\zeta_{\lambda}\br{\mathbf{M}^{\br{i}}}}}\\
\leq &~O\br{\frac{1}{\lambda}\br{N_2\br{\mathbf{D}}N_4\br{\mathbf{D}}^2+N_2\br{\mathbf{B}}N_4\br{\mathbf{B}}^2}}.	
\end{align*}

Applying \cref{lem:Xhypercontractivity}, we have
\begin{align*}
&\abs{\expec{}{\frac{1}{m^{n-i-1}}\Tr~\zeta_{\lambda}\br{\mathbf{M}^{\br{i+1}}}-\frac{1}{m^{n-i}}\Tr~\zeta_{\lambda}\br{\mathbf{M}^{\br{i}}}}}\\
\leq &~O\br{\frac{3^dm^{d/2}}{\lambda}\br{N_2\br{\mathbf{B}}^3+N_2\br{\mathbf{D}}^3}}.
\end{align*}
	
	Notice that
	\[N_2\br{\mathbf{B}}=N_2\br{\mathbf{D}}=\br{\sum_{\sigma:\sigma_{i+1}\neq 0}\abs{\widehat{M}\br{\sigma}^2}}^{1/2}=\influence_{i+1}\br{M}^{1/2}.\]
	Therefore,	\[\abs{\expec{}{\frac{1}{m^{n-i-1}}\Tr~\zeta_{\lambda}\br{\mathbf{M}^{\br{i+1}}}-\frac{1}{m^{n-i}}\Tr~\zeta_{\lambda}\br{\mathbf{M}^{\br{i}}}}}\leq O\br{\frac{m^{d/2}3^d}{\lambda}\influence_{i+1}\br{M}^{3/2}}.\]
\end{proof}

\begin{claim}\label{claim:1}
	It holds that
	\[\expec{}{\Tr~D\zeta_{\lambda}\br{\mathbf{C}}\Br{\mathbf{D}}}=m\expec{}{\Tr~D\zeta_{\lambda}\br{\mathbf{A}}\Br{\mathbf{B}}};\]	\[\expec{}{\Tr~D^2\zeta_{\lambda}\br{\mathbf{C}}\Br{\mathbf{D}}}=m\expec{}{\Tr~D^2\zeta_{\lambda}\br{\mathbf{A}}\Br{\mathbf{B}}}.\]		
\end{claim}
The proofs of both claims above are deferred to \cref{sec:appinvariance}.

Combining~\cref{lem:hybrid} and \cref{lem:zeta}, we have the following lemma.

\begin{lemma}\label{lem:invariance}
	Given $M\in\H_m^{\otimes n}$, let $\set{\B_i}_{i=0}^{m^2-1}$ be a standard orthonormal basis in $\M_m$. Then for any $\lambda>0$ and $H\subseteq[n]$, it holds that
\begin{eqnarray*}
  && \abs{\expec{}{\frac{1}{m^h}\Tr~\zeta\br{\sum_{\sigma\in[m^2]_{\geq 0}^n}\widehat{M}\br{\sigma}\prod_{i\notin H}\mathbf{g}_{i,\sigma_i}\br{\bigotimes_{i\in H}\B_{\sigma_i}}}}-\frac{1}{m^n}\Tr~\zeta\br{M}} \\
  &\leq& O\br{\lambda^2+\frac{3^d m^{d/2}}{\lambda}\sum_{i\notin H}\influence_i\br{M}^{3/2}},
\end{eqnarray*}
	where $d=\deg M$ and $h=\abs{H}$.
\end{lemma}
\begin{proof}
  Note that we can replace the quantum registers by Gaussian random variables in any order. Thus, we may assume without loss of generality that $H=[\abs{H}]$. The conclusion follows from combining~\cref{lem:hybrid} and \cref{lem:zeta}.
\end{proof}

\begin{lemma}\label{lem:invarianceH}
	Given $0<\tau,\delta<1$, $M\in\H_m^{\otimes n}$, $H\subseteq[n]$ of size $\abs{H}=h$, an integer $d>0$ and standard orthonormal basis $\B=\set{\B_i}_{i=0}^{m^2-1}$, suppose $\nnorm{M}_2\leq 1$ and $\influence_i\br{M}\leq \tau$ for all $i\notin H$. Set
	\[\mathbf{M}=\sum_{\sigma\in[m^2]_{\geq 0}^n:\abs{\sigma}\leq d}\widehat{M}\br{\sigma}\prod_{i\notin H}\mathbf{g}_{i,\sigma_i}\br{\bigotimes_{i\in H}\B_{\sigma_i}}.\]
Then it holds that
\[\abs{\frac{1}{m^h}\expec{}{\Tr~\zeta\br{\mathbf{M}}}-\frac{1}{m^n}\Tr~\zeta\br{M^{\leq d}}}\leq O\br{\br{3^dm^{d/2}\sqrt{\tau} d}^{2/3}}\]
where $M^{\leq d}$ is defined in \cref{def:lowdegreehighdegree}.
\end{lemma}

\begin{proof}
	Applying \cref{lem:invariance},
	\begin{eqnarray}
	&&\abs{\frac{1}{m^h}\expec{}{\Tr~\zeta\br{\mathbf{M}}}-\frac{1}{m^n}\Tr~\zeta\br{M^{\leq d}}}\nonumber\\
	&\leq&O\br{\lambda^2+\frac{3^dm^{d/2}}{\lambda}\sum_{i\notin H}\influence_i\br{M^{\leq d}}^{3/2}}\nonumber\\
	&\leq&O\br{\lambda^2+\frac{3^dm^{d/2}\sqrt{\tau}}{\lambda}\influence\br{M^{\leq d}}}\nonumber\\
	&\leq&O\br{\lambda^2+\frac{3^dm^{d/2}\sqrt{\tau} d}{\lambda}},\label{eqn:etamlowdiff}
	\end{eqnarray}
	where the last inequality comes from \cref{lem:partialvariance} item 4.

	Choosing $\lambda=\br{3^dm^{d/2}\sqrt{\tau} d}^{1/3}$, we conclude the result.
\end{proof}

%

\begin{proof}[Proof of \cref{lem:jointinvariance}]
	From \cref{lem:normofM}, we may let $\set{\A_i}_{i=0}^{m^2-1}$ and $\set{\B_i}_{i=0}^{m^2-1}$ be standard orthonormal bases in $\M_m$ satisfying that

\[\Tr~\br{\A_i\otimes\B_j}\psi_{AB}=c_i\delta_{i,j}~\mbox{for}~ 0\leq i,j\leq m^2-1,\]
where $1=c_0>c_1=\rho\geq c_2\geq \dots\geq c_{m^2-1}\geq 0$. Recall that $\G_{\rho}$ represents a two-dimensional Gaussian distribution $N\br{\begin{pmatrix}
       0 \\
       0
     \end{pmatrix},\begin{pmatrix}
                     1 & \rho \\
                     \rho & 1
                   \end{pmatrix}}$. We introduce random variables $\br{\br{\mathbf{g}_{i,j}^{(0)},\mathbf{h}_{i,j}^{(0)}}}_{(i,j)\in[n-h]\times[m^2]_{\geq 0}}$ as follows.

	\begin{align}
&\mathbf{g}_{i,0}^{(0)}=\mathbf{h}_{i,0}^{(0)}=1~\mbox{for $1\leq i\leq n-h $};\nonumber\\
&\br{\mathbf{g}_{i,j}^{(0)},\mathbf{h}_{i,j}^{(0)}}_{j\in[m^2-1]}\sim\G_{c_1}\otimes\ldots\otimes\G_{c_{m^2-1}}~\mbox{for $1\leq i\leq n-h$}\label{eqn:ghcor};
\end{align}

$\br{\br{\mathbf{g}_{i,j}^{(0)},\mathbf{h}_{i,j}^{(0)}}_{j\in[m^2-1]}}_{1\leq i\leq n-h}$ are independent across the indices $i$'s.
	
	Define
	\[\mathbf{P}^{(0)}=\sum_{\sigma\in[m^2]_{\geq 0}^n:\abs{\sigma}\leq d}\widehat{P}\br{\sigma}\br{\prod_{i\notin H}\mathbf{g}^{(0)}_{i,\sigma_i}}\A_{\sigma_H},\]
	and
	\[\mathbf{Q}^{(0)}=\sum_{\sigma\in[m^2]_{\geq 0}^n:\abs{\sigma}\leq d}\widehat{Q}\br{\sigma}\br{\prod_{i\notin H}\mathbf{h}^{(0)}_{i,\sigma_i}}\B_{\sigma_H}.\]
	Then
	\[N_2\br{\mathbf{P}^{(0)}}^2\leq\sum_{\sigma}\abs{\widehat{P}\br{\sigma}}^2=\nnorm{P}_2^2, ~N_2\br{\mathbf{Q}^{(0)}}^2\leq\sum_{\sigma}\abs{\widehat{Q}\br{\sigma}}^2=\nnorm{Q}_2^2,\]	and
	\begin{align*}
	&\abs{\Tr\br{\br{P\otimes Q}\psi_{AB}^{\otimes n}}-\expec{}{\Tr~\br{\br{\mathbf{P}^{(0)}\otimes\mathbf{Q}^{(0)}}\psi_{AB}^{\otimes h}}}}^2\\
&=\abs{\sum_{\sigma:\abs{\sigma}>d}c_{\sigma}\widehat{P}\br{\sigma}\widehat{Q}\br{\sigma}}^2\\
&\leq\sum_{\sigma:\abs{\sigma}>d}c_{\sigma}\abs{\widehat{P}\br{\sigma}}^2\sum_{\sigma:\abs{\sigma}>d}c_{\sigma}\abs{\widehat{Q}\br{\sigma}}^2\quad\quad\mbox{(Cauchy-Schwartz)}\\
&\leq\nnorm{P^{>d}}_2^2\nnorm{Q^{>d}}_2^2\quad\quad\mbox{($c_{\sigma}\leq 1$ due to \cref{lem:normofM})}\\
&\leq\delta^2.
\end{align*}
From \cref{lem:invarianceH},
	\begin{equation}\label{eqn:roundp}
\abs{\frac{1}{m^h}\expec{}{\Tr~\zeta\br{\mathbf{P}^{(0)}}}-\frac{1}{m^n}\Tr~\zeta\br{P^{\leq d}}}\leq O\br{\br{3^dm^{d/2}\sqrt{\tau}d}^{2/3}}
\end{equation}
\[\abs{\frac{1}{m^h}\expec{}{\Tr~\zeta\br{\mathbf{Q}^{(0)}}}-\frac{1}{m^n}\Tr~\zeta\br{Q^{\leq d}}}\leq O\br{\br{3^dm^{d/2}\sqrt{\tau}d}^{2/3}}\]
Note that $\zeta\br{P}=0$ since $0\leq P\leq 1$. Applying \cref{lem:zetaadditivity},
	\begin{eqnarray}
	&&\abs{\Tr~\zeta\br{P^{\leq d}}}=\abs{\Tr~\zeta\br{P-P^{> d}}-\Tr~\zeta\br{P}}\nonumber\\
	&\leq&2\br{\twonorm{P}\twonorm{P^{>d}}+\twonorm{P^{>d}}^2}\leq4\sqrt{\delta}m^n.\label{eqn:msigmalow}
	\end{eqnarray}
	
Combined with Eq.~\cref{eqn:roundp} we conclude that

\[\expec{}{\Tr~\zeta\br{\mathbf{P}^{(0)}}}\leq O\br{m^h\br{\br{3^dm^{d/2}\sqrt{\tau}d}^{2/3}+\sqrt{\delta}}}.\]
Similarly, we have
\[\expec{}{\Tr~\zeta\br{\mathbf{Q}^{(0)}}}\leq O\br{m^h\br{\br{3^dm^{d/2}\sqrt{\tau}d}^{2/3}+\sqrt{\delta}}}.\]

It remains to show that the distribution in Eq.~\cref{eqn:ghcor}  can be obtained from \linebreak$\G_{\rho}^{\otimes 2(m^2-1)(n-h)}$.
	Given $\br{\mathbf{g}_i,\mathbf{h}_i}_{1\leq i\leq 2(m^2-1)(n-h)}\sim \G_{\rho}^{\otimes 2(m^2-1)(n-h)}$, we perform the following substitutions in $\mathbf{P}^{(0)}$ and $\mathbf{Q}^{(0)}$
	\[\mathbf{g}_{i,b}^{(0)}\leftarrow\begin{cases}
	1, & \mbox{if $b=0$} \\
	\mathbf{g}_{(m^2-1)(i-1)+b}, & \mbox{otherwise};
	\end{cases}\]\[
	\mathbf{h}_{i,b}^{(0)}\leftarrow\begin{cases}
	1, & \mbox{if $b=0$} \\
	\frac{c_i}{\rho}\mathbf{h}_{(m^2-1)(i-1)+b}\,+\,\sqrt{1-\br{\frac{c_i}{\rho}}^2}\mathbf{h}_{(m^2-1)(n-h+i-1)+b}, & \mbox{otherwise}
	\end{cases}\]
	to get $\mathbf{P}$ and $\mathbf{Q}$, respectively. Then the distributions $\br{\mathbf{g}_{i,b}^{(0)},\mathbf{h}_{i,b}^{(0)}}_{1\leq i\leq n-h, 0\leq b\leq 3}$ are independent over the indices $\br{i,b}$'s. $\br{\mathbf{g}_{i,b}^{(0)}}_{1\leq i\leq n-h, 1\leq b\leq 3}$ and $\br{\mathbf{h}_{i,b}^{(0)}}_{1\leq i\leq n-h, 1\leq b\leq 3}$ are both i.i.d. standard normal distributions. For all $i\in[n-h],j\in[m^2-1]$, the correlation between $\mathbf{g}_{i,j}^{(0)}$ and $\mathbf{h}_{i,j}^{(0)}$ equals to
	\begin{align*}
&\expec{}{\mathbf{g}_{i,j}^{(0)}\mathbf{h}_{i,j}^{(0)}}\\
=&\expec{}{\mathbf{g}_{(m^2-1)(i-1)+j}\br{\frac{c_i}{\rho}\mathbf{h}_{(m^2-1)(i-1)+j}\,+\,\sqrt{1-\br{\frac{c_i}{\rho}}^2}\mathbf{h}_{(m^2-1)(n-h+i-1)+j}}}\\
=&\frac{c_i}{\rho}\expec{}{\mathbf{g}_{(m^2-1)(i-1)+j}\mathbf{h}_{(m^2-1)(i-1)+j}}\\
=&c_i.
\end{align*}

Thus, $\mathbf{P}^{(0)}$ and $\mathbf{P}$ have the same distribution. Same for $\mathbf{Q}^{(0)}$ and $\mathbf{Q}$. Items 1 to 3 follow.

Item 4 immediately follows from the construction of the map.
\end{proof}

Reversely, the following lemma converts joint random operators to operators.

\begin{lemma}\label{lem:invariancejointgaussian}
	Given $0<\tau,\delta,\rho<1$, integers $n>h\geq 0, d>0, m>1$, a noisy MES state $\psi_{AB}$ with the maximal correlation $\rho=\rho\br{\psi_{AB}}$, there exist maps $f,g:L^2\br{\H_m^{\otimes h},\gamma_n}\rightarrow\H_m^{\otimes n+h}$ such that for any degree-$d$ multilinear joint random operators
	\begin{multline*}
\br{\mathbf{P},\mathbf{Q}}=\br{\sum_{\sigma\in[m^2]_{\geq 0}^h}p_{\sigma}\br{\mathbf{g}}\A_{\sigma},\sum_{\sigma\in[m^2]_{\geq 0}^h}q_{\sigma}\br{\mathbf{h}}\B_{\sigma}}_{\br{\mathbf{g},\mathbf{h}}\sim\G_{\rho}^{\otimes n}}\\\in L^2\br{\H_m^{\otimes h},\gamma_n}\times L^2\br{\H_m^{\otimes h},\gamma_n},
\end{multline*}
	satisfying that $N_2\br{\mathbf{P}}\leq 1$, $N_2\br{\mathbf{Q}}\leq 1$, and
	\[(\forall i\in[n]):~\sum_{\sigma\in[m^2]_{\geq 0}^h}\influence_i\br{p_{\sigma}}\leq\tau~\mbox{and}~\sum_{\sigma\in[m^2]_{\geq 0}^h}\influence_i\br{q_{\sigma}}\leq\tau,\]
	where $\set{\A_i}_{i=0}^{m^2-1}$ and $\set{\B_i}_{i=0}^{m^2-1}$ are standard orthonormal bases in $\M_m$ satisfying
\[\Tr~\br{\A_i\otimes\B_j}\psi_{AB}=c_i\delta_{i,j}~\mbox{for}~0\leq i,j\leq m^2-1\]
and $1=c_0>c_1=\rho\geq c_2\geq \dots\geq c_{m^2-1}\geq 0$.
	Let $\br{P,Q}=\br{f\br{\mathbf{P}},g\br{\mathbf{Q}}}$. The following holds.
	\begin{enumerate}

		\item $\Tr\br{\br{P\otimes Q}\psi_{AB}^{\otimes (n+h)}}=\expec{}{\Tr\br{\br{\mathbf{P}\otimes\mathbf{Q}}\psi_{AB}^{\otimes h}}};$
		
		\item $N_2\br{\mathbf{P}}=\nnorm{P}_2~\mbox{and}~N_2\br{\mathbf{Q}}=\nnorm{Q}_2;$
		
		\item $\abs{\expec{}{\frac{1}{m^h}\Tr~\zeta\br{\mathbf{P}}}-\frac{1}{m^n}\Tr~\zeta\br{P}}\leq O\br{\br{3^dm^{d/2}d\sqrt{\tau}}^{2/3}}$
		
		and
		
		$\abs{\expec{}{\frac{1}{m^h}\Tr~\zeta\br{\mathbf{Q}}}-\frac{1}{m^n}\Tr~\zeta\br{Q}}\leq O\br{\br{3^dm^{d/2}d\sqrt{\tau}}^{2/3}}$
		\item The maps $f$ and $g$ are linear and unital.
		
	\end{enumerate}
\end{lemma}

\begin{proof}
Since $\mathbf{P},\mathbf{Q}$ are multilinear random Hermitian operators, we can assume that $$p_{\sigma}\br{\mathbf{g}}=\sum_{\mu\in\{0,1\}^n}p_{\sigma}(\mu)\prod_{j=1}^n\mathbf{g}_j^{\mu_j};$$
$$q_{\sigma}\br{\mathbf{h}}=\sum_{\mu\in\{0,1\}^n}q_{\sigma}(\mu)\prod_{j=1}^n\mathbf{h}_j^{\mu_j},$$
where $p_{\sigma}(\mu),q_{\sigma}(\mu)\in\reals$. Then $\mathbf{P}$ and $\mathbf{Q}$ can be expressed as
$$\mathbf{P}=\sum_{\sigma\in[m^2]_{\geq 0}^h}\sum_{\mu\in\{0,1\}^n}p_{\sigma}(\mu)\prod_{j=1}^n\mathbf{g}_j^{\mu_j}\A_{\sigma};$$
$$\mathbf{Q}=\sum_{\sigma\in[m^2]_{\geq 0}^h}\sum_{\mu\in\{0,1\}^n}q_{\sigma}(\mu)\prod_{j=1}^n\mathbf{h}_j^{\mu_j}\B_{\sigma}.$$
Define
$$P=\sum_{\sigma\in[m^2]_{\geq 0}^h}\sum_{\mu\in\{0,1\}^n}p_{\sigma}(\mu)\br{\bigotimes_{j=1}^n\A_{\mu_j}}\otimes\A_{\sigma};$$
$$Q=\sum_{\sigma\in[m^2]_{\geq 0}^h}\sum_{\mu\in\{0,1\}^n}q_{\sigma}(\mu)\br{\bigotimes_{j=1}^n\B_{\mu_j}}\otimes\B_{\sigma}.$$
Then
$$\Tr\br{\br{P\otimes Q}\psi_{AB}^{\otimes (n+h)}}=\expec{}{\Tr\br{\br{\mathbf{P}\otimes\mathbf{Q}}\psi_{AB}^{\otimes h}}}=\sum_{\sigma\in[m^2]_{\geq 0}^h}\sum_{\mu\in\{0,1\}^n}p_{\sigma}(\mu)q_{\sigma}(\mu)\rho^{\abs{\mu}}c_\sigma.$$
\begin{align*}
N_2\br{\mathbf{P}}^2&=\expec{}{\sum_{\sigma\in[m^2]_{\geq 0}^h}\abs{\sum_{\mu\in\{0,1\}^n}p_{\sigma}(\mu)\prod_{j=1}^n\mathbf{g}_j^{\mu_j}}^2}\\
&=\sum_{\sigma\in[m^2]_{\geq 0}^h}\expec{}{\abs{\sum_{\mu\in\{0,1\}^n}p_{\sigma}(\mu)\prod_{j=1}^n\mathbf{g}_j^{\mu_j}}^2}\\
&=\sum_{\sigma\in[m^2]_{\geq 0}^h}\sum_{\mu\in\{0,1\}^n}\abs{p_{\sigma}(\mu)}^2\\
&=\nnorm{P}_2^2.
\end{align*}

Similarly, \[N_2\br{\mathbf{Q}}^2=\nnorm{Q}_2^2.\]

To see item 3, it is critical to observe that for all $i\notin H$
\[
\influence_i(P)=\sum_{\sigma\in[m^2]_{\geq 0}^h}\sum_{\mu:\mu_i=1}\abs{p_{\sigma}(\mu)}^2=\sum_{\sigma\in[m^2]_{\geq 0}^h}\influence_i\br{p_{\sigma}}\leq\tau.
\]

Similarly, $\influence_i\br{Q}\leq\tau$.

Then from \cref{lem:invarianceH}, item 3 holds.

Item 4 follows immediately from the constructions of $f$ and $g$.
\end{proof}
\section{Dimension reduction for random operators}\label{sec:dimensionreduction}
The following is the main lemma in this section.

\begin{lemma}\label{lem:dimensionreduction}
Given parameters $\rho\in[0,1], \delta,\alpha>0$, integers $d,n,h>0, m>1$, an $m$-dimensional noisy MES $\psi_{AB}$ with the maximal correlation $\rho=\rho\br{\psi_{AB}}$, and degree-$d$ multilinear joint random operators
	\begin{multline*}\br{\mathbf{P},\mathbf{Q}}=\br{\sum_{\sigma\in[m^2]_{\geq 0}^h}p_{\sigma}\br{\mathbf{g}}\A_{\sigma},\sum_{\sigma\in[m^2]_{\geq 0}^h}q_{\sigma}\br{\mathbf{h}}\B_{\sigma}}_{\br{\mathbf{g},\mathbf{h}}\sim\G_{\rho}^{\otimes n}}\\\in L^2\br{\H_m^{\otimes h},\gamma_n}\times L^2\br{\H_m^{\otimes h},\gamma_n},\end{multline*}
	where $\set{\A_i}_{i=0}^{m^2-1},\set{\B_i}_{i=0}^{m^2-1}$ are both standard orthonormal bases in $\M_m$ satisfying
\[\Tr~\br{\A_i\otimes\B_j}\psi_{AB}=c_i\delta_{i,j}~\mbox{for}~0\leq i,j\leq m^2-1\]
and $1=c_0>c_1=\rho\geq c_2\geq \dots\geq c_{m^2-1}\geq 0$, there exists an explicitly computable $n_0=n_0\br{d,h,\delta,m}$, maps
$f_M, g_M:L^2\br{\H_m^{\otimes h},\gamma_n}\rightarrow L^2\br{\H_m^{\otimes h},\gamma_{n_0}}$
for $M\in\reals^{n\times n_0}$, and joint random operators $\br{\mathbf{P}_M^{\br{1}},\mathbf{Q}_M^{\br{1}}}=\br{f_M(\mathbf{P}),g_M(\mathbf{Q})}:$
	\[\br{\mathbf{P}_M^{\br{1}},\mathbf{Q}_M^{\br{1}}}=\br{\sum_{\sigma\in[m^2]_{\geq 0}^h}p^{\br{1}}_{\sigma,M}\br{\mathbf{x}}\A_{\sigma},\sum_{\sigma\in[m^2]_{\geq 0}^h}q^{\br{1}}_{\sigma,M}\br{\mathbf{y}}\B_{\sigma}}_{\br{\mathbf{x},\mathbf{y}}\sim\G_{\rho}^{\otimes n_0}},\]
	such that if we sample $\mathbf{M}\sim \gamma_{n\times n_0}$, then with probability at least  $1-\delta-2\alpha$, it holds that
		\begin{enumerate}
			\item $N_2\br{\mathbf{P}_{\mathbf{M}}^{(1)}}\leq\br{1+\delta}N_2\br{\mathbf{P}}$ and $N_2\br{\mathbf{Q}_{\mathbf{M}}^{(1)}}\leq\br{1+\delta}N_2\br{\mathbf{Q}}$;
			\item \[\expec{\mathbf{P}}{\Tr~\zeta\br{\mathbf{P}_{\mathbf{M}}^{\br{1}}}}\leq \frac{1}{\sqrt{\alpha}}\expec{\mathbf{P}}{\Tr~\zeta\br{\mathbf{P}}}\] and \[\expec{\mathbf{Q}}{\Tr~\zeta\br{\mathbf{Q}_{\mathbf{M}}^{\br{1}}}}\leq \frac{1}{\sqrt{\alpha}}\expec{\mathbf{Q}}{\Tr~\zeta\br{\mathbf{Q}}};\]
			\item
\begin{multline*}\abs{\expec{\mathbf{P},\mathbf{Q}}{\Tr\br{\br{\mathbf{P}_{\mathbf{M}}^{\br{1}}\otimes \mathbf{Q}_\mathbf{M}^{\br{1}}}\psi_{AB}^{\otimes h}}}-\expec{\mathbf{P},\mathbf{Q}}{\Tr\br{\br{\mathbf{P}\otimes \mathbf{Q}}\psi_{AB}^{\otimes h}}}}\\\leq\delta N_2\br{\mathbf{P}}N_2\br{\mathbf{Q}};
\end{multline*}
\item the maps $f_M, g_M$ are linear and unital for any nonzero $M\in\reals^{n\times n_0}$.
		\end{enumerate}
		
		In particular, one may take $n_0=\frac{m^{O(h)}d^{O\br{d}}}{\delta^6}$.
\end{lemma}

To prove \cref{lem:dimensionreduction}, we make use of a recent result about the dimension reduction for low-degree polynomials in Gaussian spaces due to Ghazi, Kamath and Raghavendra~\cite{Ghazi:2018:DRP:3235586.3235614}.

\begin{fact}~\cite[Theorem 3.1]{Ghazi:2018:DRP:3235586.3235614}\label{fac:dimensionreduction}
	Given parameters $n,d\in\mathbb{Z}_{>0}$, $\rho\in[0,1]$ and $\delta>0$, there exists an explicitly computable $D=D\br{d,\delta}$ such that the following holds.
	
	For any $n$ and any degree-$d$ multilinear polynomials $\alpha,\beta:\reals^n\rightarrow\reals$, and $M\in\reals^{n\times D}$, define functions $\alpha_M, \beta_M:\reals^{D}\rightarrow\reals$ as
	\begin{equation}\label{eqn:fmgm}
	\alpha_M\br{x}= \alpha\br{\frac{Mx}{\twonorm{x}}}~\mbox{and}~\beta_M\br{x}= \beta\br{\frac{Mx}{\twonorm{x}}}.
	\end{equation}
	Then
	\[\Pr_{\mathbf{M}\sim \gamma_{n\times D}}\Br{\abs{\innerproduct{\alpha_{\mathbf{M}}}{\beta_{\mathbf{M}}}_{\G_{\rho}^{\otimes D}}-\innerproduct{\alpha}{\beta}_{\G_{\rho}^{\otimes n}}}<\delta\twonorm{f}\twonorm{g}}\geq 1-\delta.\]

If $\alpha$ and $\beta$ are identical and $\rho=1$, we have

\begin{eqnarray}
&&\Pr_{\mathbf{M}\sim \gamma_{n\times D}}\Br{\abs{\twonorm{\alpha_{\mathbf{M}}}^2-\twonorm{\alpha}^2}\leq\delta\twonorm{\alpha}^2}\geq1-\delta; \\
&&\Pr_{\mathbf{M}\sim \gamma_{n\times D}}\Br{\abs{\twonorm{\beta_{\mathbf{M}}}^2-\twonorm{\beta}^2}\leq\delta\twonorm{\beta}^2}\geq1-\delta.
\end{eqnarray}

	In particular, one may take $D=\frac{d^{O\br{d}}}{\delta^6}$.
\end{fact}


\begin{fact}~\cite[Proposition 3.2]{Ghazi:2018:DRP:3235586.3235614}\label{fac:reductionrounding}
	Given integers $n, k, D>0$, let $\alpha\in L^2\br{\reals^k,\gamma_n}$, and $\Lambda$ be a closed convex set \footnote{In \cite{Ghazi:2018:DRP:3235586.3235614} $\Lambda$ is a simplex of probability distributions.  It is not hard to verify that it also holds for closed convex sets.} in $\reals^k$ with the rounding map $\R$ defined in \cref{subsec:misc}. Let $\alpha_M:\reals^D\rightarrow\reals^k$ be defined analogously to Eq.~\cref{eqn:fmgm}. It holds that,
	\[\Pr_{\mathbf{M}\sim \gamma_{n\times D}}\Br{\twonorm{\R\circ \alpha_{\mathbf{M}}-\alpha_{\mathbf{M}}}\leq\frac{1}{\delta}\twonorm{\R\circ \alpha-\alpha}}\geq1-\delta^2,\]
	for any $0<\delta<1$.
\end{fact}

Before proving \cref{lem:dimensionreduction}, we introduce a closed convex set in $\reals^{m^{2h}}$ for any given standard orthonormal basis $\set{\A_{\sigma}}_{\sigma\in[m^2]_{\geq 0}^h}$.

\begin{definition}\label{def:convexsetlambda}
For any integers $m,h>0$, let $\set{\A_{\sigma}}_{\sigma\in[m^2]_{\geq 0}}$ be a standard orthonormal basis of $\H_m$. We define
\begin{equation}\label{eqn:lambda}
  \Lambda\br{\set{\A_{\sigma}}_{\sigma\in[m^2]_{\geq 0}^h}}=\set{x\in\reals^{m^{2h}}:\sum_{\sigma\in[m^2]_{\geq 0}^h}x_{\sigma}\A_{\sigma}\geq0}.
\end{equation}
\end{definition}

We are now ready to prove \cref{lem:dimensionreduction}.

\begin{proof}[Proof of \cref{lem:dimensionreduction}]
	From \cref{lem:convariancetensor}
	\begin{eqnarray*}
		\expec{\br{\mathbf{g},\mathbf{h}}\sim\G_{\rho}^{\otimes n}}{\Tr\br{\br{P\br{\mathbf{g}}\otimes Q\br{\mathbf{h}}}\psi_{AB}^{\otimes n}}}=\sum_{\sigma\in[m^2]_{\geq 0}^h}c_{\sigma}\innerproduct{p_{\sigma}}{g_{\sigma}}_{\G_{\rho}^{\otimes n}},
	\end{eqnarray*}
	where $c_{\sigma}= c_{\sigma_1}\cdots c_{\sigma_h}$.
	
	Let $n_0=\frac{m^{12h+12}d^{O\br{d}}}{\delta^6}$. Applying \cref{fac:dimensionreduction} by setting the parameters $\delta\leftarrow\frac{\delta}{ m^{2h+2}}$, $D\leftarrow n_0$ , and a union bound on $\sigma\in[m^2]_{\geq 0}^h$, it holds that
	
	\begin{multline}\label{eqn:mpq}
	\Pr_{\mathbf{M}\sim \gamma_{n\times n_0}}\Br{\br{\forall \sigma\in[m^2]_{\geq 0}^h}~\abs{\innerproduct{p_{\sigma,\mathbf{M}}}{q_{\sigma,\mathbf{M}}}_{\G_{\rho}^{\otimes n_0}}-\innerproduct{p_{\sigma}}{q_{\sigma}}_{\G_{\rho}^{\otimes n}}}\leq \delta\twonorm{p_{\sigma}}\twonorm{q_{\sigma}}}\\\geq1-\delta/m^2,
	\end{multline}
	and
	\begin{equation}\label{eqn:mpq5}
	\Pr_{\mathbf{M}\sim \gamma_{n\times n_0}}\Br{\br{\forall\sigma\in[m^2]_{\geq 0}^h}:~\abs{\twonorm{p_{\sigma,\mathbf{M}}}^2-\twonorm{p_{\sigma}}^2}\leq\delta\twonorm{p_{\sigma}}^2}\geq1-\delta/m^2,
	\end{equation}
	and
	\begin{equation}
	\Pr_{\mathbf{M}\sim \gamma_{n\times n_0}}\Br{\br{\forall\sigma\in[m^2]_{\geq 0}^h}:~\abs{\twonorm{q_{\sigma,\mathbf{M}}}^2-\twonorm{q_{\sigma}}^2}\leq\delta\twonorm{q_{\sigma}}^2}\geq1-\delta/m^2.\label{eqn:mpq6}
	\end{equation}
	Define
	\begin{align}
  &\br{f_M\br{\mathbf{P}},g_M\br{\mathbf{Q}}}=\br{\mathbf{P}^{\br{1}}_M,\mathbf{Q}^{\br{1}}_M}\label{eqn:fmgmmm}\\
	 &=\br{\sum_{\sigma\in[m^2]_{\geq 0}^h}p_{\sigma,M}\br{\mathbf{g}}\A_{\sigma},\sum_{\sigma\in[m^2]_{\geq 0}^h}q_{\sigma,M}\br{\mathbf{h}}\B_{\sigma}}_{\br{\mathbf{g},\mathbf{h}}\sim\G_{\rho}^{\otimes n_0}}.\nonumber
	\end{align}
	For any $M$ satisfying Eq.~\cref{eqn:mpq}, we have
	\begin{eqnarray*}
		&&\abs{\expec{}{\Tr\br{\br{\mathbf{P}_M^{\br{1}}\otimes \mathbf{Q}_M^{\br{1}}}\psi_{AB}^{\otimes h}}}-\expec{}{\Tr\br{\br{\mathbf{P}\otimes \mathbf{Q}}\psi_{AB}^{\otimes h}}}} \\
		&=&\abs{\sum_{\sigma\in[m^2]_{\geq 0}^h}c_{\sigma}\br{\innerproduct{p^{\br{1}}_{\sigma,M}}{q^{\br{1}}_{\sigma,M}}_{\G_{\rho}^{\otimes n_0}}-\innerproduct{p_{\sigma}}{q_{\sigma}}_{\G_{\rho}^{\otimes n}}}}\quad\quad\mbox{(\cref{lem:convariancetensor})}\\
		&\leq&\delta\sum_{\sigma\in[m^2]_{\geq 0}^h}\twonorm{p_{\sigma}}\twonorm{q_{\sigma}}\quad\quad\mbox{(Eq.~\cref{eqn:mpq} and $c_{\sigma}\leq 1$ due to \cref{lem:normofM})}\\
		&\leq&\delta\br{\sum_{\sigma\in[m^2]_{\geq 0}^h}\twonorm{p_{\sigma}}^2}^{1/2}\br{\sum_{\sigma\in[m^2]_{\geq 0}^h}\twonorm{q_{\sigma}}^2}^{1/2}\\
		&=&\delta N_2\br{\mathbf{P}}N_2\br{\mathbf{Q}}\quad\quad\mbox{(\cref{lem:randoperator})}.
	\end{eqnarray*}
	Thus
	\begin{multline}\label{eqn:m1}
	\Pr_{\mathbf{M}}\Br{\abs{\expec{\mathbf{P},\mathbf{Q}}{\Tr\br{\br{\mathbf{P}_{\mathbf{M}}^{\br{1}}\otimes \mathbf{Q}_{\mathbf{M}}^{\br{1}}}\psi_{AB}^{\otimes h}}}-\expec{\mathbf{P},\mathbf{Q}}{\Tr\br{\br{\mathbf{P}\otimes \mathbf{Q}}\psi_{AB}^{\otimes h}}}}\leq\delta N_2\br{\mathbf{P}}N_2\br{\mathbf{Q}}}\\\geq1-\delta/m^2.
	\end{multline}
	For any $M$ satisfying Eq.~\cref{eqn:mpq5},
	\[N_2\br{\mathbf{P}_M^{(1)}}^2=\sum_{\sigma\in[m^2]_{\geq 0}^h}\twonorm{p_{\sigma,M}}^2\leq\br{1+\delta}\sum_{\sigma\in[m^2]_{\geq 0}^h}\twonorm{p_{\sigma}}^2=\br{1+\delta}N_2\br{\mathbf{P}}^2,\]
	where both equalities are from \cref{lem:randoperator}. Hence
	\begin{equation}\label{eqn:expecnormp}
	\Pr_{\mathbf{M}\sim \gamma_{n\times n_0}}\Br{N_2\br{\mathbf{P}^{(1)}_{\mathbf{M}}}\leq\br{1+\delta}N_2\br{\mathbf{P}}}\geq1-\delta/m^2.
	\end{equation}
	Using the same argument and Eq.~\eqref{eqn:mpq5}, we have
	\begin{equation}\label{eqn:expecnormq}
	\Pr_{\mathbf{M}\sim \gamma_{n\times n_0}}\Br{N_2\br{\mathbf{Q}^{(1)}_{\mathbf{M}}}\leq\br{1+\delta}N_2\br{\mathbf{Q}}}\geq1-\delta/m^2.
	\end{equation}
	Set $\Lambda=\Lambda\br{\set{\A_{\sigma}}_{\sigma\in[m^2]_{\geq 0}^h}}$ as defined in \cref{eqn:lambda}, which is a closed convex set. Let $\R$ be a rounding map of $\Lambda$ defined in \cref{subsec:misc}. For any random operator $\mathbf{P}\in L^2\br{\H_m^{\otimes h},\gamma_n}$, let $p$ be the associated vector-valued function under the basis $\set{\A_{\sigma}}_{\sigma\in[m^2]_{\geq 0}}$ defined in \cref{def:randomoperators}. Then we have
\begin{align*}
\twonorm{\R\circ p-p}^2&=\sum_{\sigma\in[m^2]_{\geq 0}^h}\expec{\mathbf{g}\sim\gamma_n}{\br{\R\circ p_\sigma-p_\sigma}^2(\mathbf{g})}\\
&=\br{N_2\br{\theta(\mathbf{P})-\mathbf{P}}}^2\\
&=\frac{1}{m^h}\expec{\mathbf{P}}{\Tr~\zeta\br{\mathbf{P}}},
\end{align*}
where for $M\in\H_m^{\otimes n}$, \[\theta(M)=\arg\min\set{\twonorm{M-X}:X\geq0},\]
and the equality follows from \cref{lem:closedelta1}.
	Hence \cref{fac:reductionrounding} implies that
	\begin{equation}\label{eqn:rpq}
	\Pr_{\mathbf{M}\sim \gamma_{n\times n_0}}\Br{\expec{\mathbf{P}}{\Tr~\zeta\br{\mathbf{P}^{\br{1}}}}\leq \frac{1}{\sqrt{\alpha}}\expec{\mathbf{P}}{\Tr~\zeta\br{\mathbf{P}}}}\geq1-\alpha.
	\end{equation}
	Applying the same argument to $\mathbf{Q}$ and $\mathbf{Q}^{\br{1}}$, we have
	\begin{equation}\label{eqn:rpq2}
	\Pr_{\mathbf{M}\sim \gamma_{n\times n_0}}\Br{\expec{\mathbf{Q}}{\Tr~\zeta\br{\mathbf{Q}^{\br{1}}}}\leq \frac{1}{\sqrt{\alpha}}\expec{\mathbf{Q}}{\Tr~\zeta\br{\mathbf{Q}}}}\geq1-\alpha.
	\end{equation}

	Again applying a union bound on Eqs.~\cref{eqn:m1}\cref{eqn:expecnormp}\cref{eqn:expecnormq}\cref{eqn:rpq}\cref{eqn:rpq2}, all the events in Eqs.~\cref{eqn:m1}\cref{eqn:expecnormp}\cref{eqn:expecnormq}\cref{eqn:rpq}\cref{eqn:rpq2} occur with probability at least $1-\delta-2\alpha$ over $\mathbf{M}\sim \gamma_{n\times D}$. Setting $p^{\br{1}}_{\sigma,\mathbf{M}}=p_{\sigma,\mathbf{M}}$ and $q^{\br{1}}_{\sigma,\mathbf{M}}=q_{\sigma,\mathbf{M}}$, we conclude item 1 to 3.
	
For item 4, let
\[\mathbf{P}=\sum_{\sigma\in[m^2]_{\geq 0}^h}p_{\sigma}\br{\mathbf{g}}\A_{\sigma}~\mbox{and}~\mathbf{P}'=\sum_{\sigma\in[m^2]_{\geq 0}^h}p_{\sigma}'\br{\mathbf{g}}\A_{\sigma}.\]
By the definition of map $f_M$ in Eq.~\cref{eqn:fmgmmm},
\begin{eqnarray*}
  &&f_M\br{c\mathbf{P}+c'\mathbf{P}'} \\
  &=& \sum_{\sigma\in[m^2]_{\geq 0}^h}\br{cp_{\sigma,M}\br{\mathbf{g}}+c'p'_{\sigma,M}\br{\mathbf{g}}}\A_{\sigma}\\
  &=& c\sum_{\sigma\in[m^2]_{\geq 0}^h}p_{\sigma,M}\br{\mathbf{g}}\A_{\sigma}+ c'\sum_{\sigma\in[m^2]_{\geq 0}^h}p'_{\sigma,M}\br{\mathbf{g}}\A_{\sigma}\\
  &=&cf_M\br{\mathbf{P}}+c'f_M\br{\mathbf{P}'},
\end{eqnarray*}
for any constants $c$ and $c'$. It is easy to verify that $f_M$ is unital.
\end{proof}

\section{Smoothing random operators}\label{sec:smoothrandom}

The main result in this section is the following, which is a generalization of \cref{lem:smoothing of strategies} to random operators.

\begin{lemma}\label{lem:smoothgaussian}
	Given integers $n,h>0, m>1$, an $m$-dimensional noisy MES $\psi_{AB}$ with the maximal correlation $\rho= \rho\br{\psi_{AB}}<1$, and joint random operators
\begin{multline*}\br{\mathbf{P},\mathbf{Q}}=\br{\sum_{\sigma\in[m^2]_{\geq 0}^h}p_{\sigma}\br{\mathbf{g}}\A_{\sigma},\sum_{\sigma\in[m^2]_{\geq 0}^h}q_{\sigma}\br{\mathbf{h}}\B_{\sigma}}_{\br{\mathbf{g},\mathbf{h}}\sim\G_{\rho}^{\otimes n}}\\\in L^2\br{\H_m^{\otimes h},\gamma_n}\times L^2\br{\H_m^{\otimes h},\gamma_n},\end{multline*}
where $\set{\A_i}_{i=0}^{m^2-1},\set{\B_i}_{i=0}^{m^2-1}$ are both standard orthonormal bases in $\M_m$ satisfying
\[\Tr~\br{\A_i\otimes\B_j}\psi_{AB}=c_i\delta_{i,j}~\mbox{and}~0\leq i,j\leq m^2-1\]
and $1=c_0>c_1=\rho\geq c_2\geq \dots\geq c_{m^2-1}\geq 0$, there exists an explicitly computable $d=d\br{\rho, \delta}$ and a map $f:L^2\br{\H_m^{\otimes h},\gamma_n}\rightarrow L^2\br{\H_m^{\otimes h},\gamma_n}$ such that $$\br{\mathbf{P}^{\br{1}},\mathbf{Q}^{\br{1}}}=\br{f(\mathbf{P}),f(\mathbf{Q})}\in L^2\br{\H_m^{\otimes h},\gamma_n}\times L^2\br{\H_m^{\otimes h},\gamma_n}$$
 satisfies the following.
	\begin{enumerate}
		\item $\deg\br{\mathbf{P}^{(1)}}\leq d$ and $\deg\br{\mathbf{Q}^{(1)}}\leq d$;
		
		\item $N_2\br{\mathbf{P}^{(1)}}\leq N_2\br{\mathbf{P}}$ and $N_2\br{\mathbf{Q}^{(1)}}\leq N_2\br{\mathbf{Q}}$;
		\item $\expec{}{\Tr~\zeta\br{\mathbf{P}^{(1)}}}\leq2\br{\expec{}{\Tr~\zeta\br{\mathbf{P}}}+\delta m^hN_2\br{\mathbf{P}}^2}$ and~\\ $\expec{}{\Tr~\zeta\br{\mathbf{Q}^{(1)}}}\leq2\br{\expec{}{\Tr~\zeta\br{\mathbf{Q}}}+\delta m^hN_2\br{\mathbf{Q}}^2}$;
		\item $\abs{\expec{}{\Tr\br{\br{\mathbf{P}\otimes\mathbf{Q}}\psi_{AB}^{\otimes h}}}-\expec{}{\Tr\br{\br{\mathbf{P}^{(1)}\otimes\mathbf{Q}^{(1)}}\psi_{AB}^{\otimes h}}}}\leq \delta N_2\br{\mathbf{P}}N_2\br{\mathbf{Q}}$;

\item the map $f$ is linear and unital.
	\end{enumerate}
	In particular, one may take $d=O\br{\frac{\log^2\frac{1}{\delta}}{\delta\br{1-\rho}}}$.
\end{lemma}

To prove the main result, we again employ the following lemma about smoothing functions on Gaussian spaces in~\cite{Ghazi:2018:DRP:3235586.3235614}.
\newpage
\begin{fact}~\cite[Lemma 4.1]{Ghazi:2018:DRP:3235586.3235614}\label{fac:smoothgaussian}\footnote{
There are several differences between the statement here and Lemma 4.1 in~\cite{Ghazi:2018:DRP:3235586.3235614}, which are listed below.
\begin{enumerate}
\item In~\cite{Ghazi:2018:DRP:3235586.3235614}, $\Lambda_1=\Lambda_2$ is a simplex of probability distributions. The same proof also works for any closed convex sets $\Lambda_1, \Lambda_2$
\item In~\cite{Ghazi:2018:DRP:3235586.3235614}, it is proved that $\var{f_i^{(1)}}\leq\var{f_i}$ and $\var{g_i^{(1)}}\leq\var{g_i}$. The exactly same proof also works for $\twonorm{\cdot}$.
\item From the proof in~\cite[Page 44, arxiv version]{Ghazi:2018:DRP:3235586.3235614}, you can see that
\[\twonorm{\R_1\circ f^{(1)}-f^{(1)}}\leq\twonorm{\R_1\circ f-f}\leq\twonorm{U_{\rho}f^{>d}}\leq\delta\twonorm{f}.\]
Same for $g$ and $g^{(1)}$.
\item From the proof in~\cite[Page 44, arxiv version]{Ghazi:2018:DRP:3235586.3235614} and the fact that $$\twonorm{f^{(1)}_i-U_{\rho}f_i}=\twonorm{U_{\rho}f^{>d}}\leq\frac{\delta}{2\sqrt{k}}\twonorm{f_i}$$
    and
    $$\twonorm{g^{(1)}_i-U_{\rho}g_i}=\twonorm{U_{\rho}g^{>d}}\leq\frac{\delta}{2\sqrt{k}}\twonorm{g_i},$$
    we can conclude item 4.
\item The function $f^{(1)}$ is obtained by applying the Ornstein-Uhlenbeck operator in Definition~\ref{def:ornstein} to $f_i$ and truncating the high-degree part. Hence, the operations are linear and keep constant functions unchanged. Same for $g_i^{(1)}$ and $g_i$.
\end{enumerate}
}
	Let $\rho\in[0,1),\delta>0,k,n\in\mathbb{Z}_{>0}$ be any given constant parameters, $f,g\in L^2\br{\reals^k,\gamma_n}$; $\Lambda_1,\Lambda_2\subseteq\reals^k$ be closed convex sets. Set $\R_1$ and $\R_2$ be rounding maps of $\Lambda_1$ and $\Lambda_2$, respectively. Then there exists an explicitly computable $d=d\br{\rho,\delta}$ and functions $f^{(1)}, g^{(1)}\in L^2\br{\reals^k,\gamma_n}$,where $f^{(1)}$ only depends on $f$ and $g^{(1)}$ only depends on $g$, such that the following holds.
	\begin{enumerate}
		\item Both $f^{(1)}$ and $g^{(1)}$ are of degree at most $d$.
		\item For any $i\in[k]$, it holds that $\twonorm{f_i^{(1)}}\leq\twonorm{f_i}$ and $\twonorm{g_i^{(1)}}\leq\twonorm{g_i}$.
		\item $$\twonorm{\R_1\circ f^{(1)}-f^{(1)}}\leq\twonorm{\R_1\circ f-f}+\delta\twonorm{f}$$
		and $$\twonorm{\R_2\circ g^{(1)}-g^{(1)}}\leq\twonorm{\R_2\circ g-g}+\delta\twonorm{g}.$$
		\item For every $i\in [k]$,
		\[\abs{\innerproduct{f_i\br{\mathbf{x}}}{g_i\br{\mathbf{y}}}_{\G_{\rho}^{\otimes n}}-\innerproduct{f_i^{(1)}\br{\mathbf{x}}}{g_i^{(1)}\br{\mathbf{y}}}_{\G_{\rho}^{\otimes n}}}\leq\delta\twonorm{f_i}\twonorm{g_i};\]
        \item For any constants $c_1, c_2$, it holds that $\br{c_1f+c_2g}^{(1)}=c_1f^{(1)}+c_2g^{(1)}$ and $f^{(1)}=f$ if $f$ is a constant function.
	\end{enumerate}
	In particular, one may take $d=O\br{\frac{\log^2\frac{1}{\delta}}{\delta\br{1-\rho}}}$.
\end{fact}

Let's proceed to the proof of \cref{lem:smoothgaussian}.

\begin{proof}[Proof of \cref{lem:smoothgaussian}]
Recall that $p$ and $q$ are the associated vector-valued functions of $\mathbf{P}$ and $\mathbf{Q}$ under the bases $\set{\A_i}_{i=0}^{m^2-1}$ and $\set{\B_i}_{i=0}^{m^2-1}$, respectively.
	
	Set \[\Lambda_1=\Lambda\br{\set{\A_{\sigma}}_{\sigma\in[m^2]_{\geq 0}^h}}~\mbox{and}~\Lambda_2=\Lambda\br{\set{\B_{\sigma}}_{\sigma\in[m^2]_{\geq 0}^h}},\]
as defined in \cref{eqn:lambda}.
	
	Applying \cref{fac:smoothgaussian} to $\br{p,q}$, we obtain
	$\br{p^{(1)},q^{(1)}}$.
	Define
	\[\mathbf{P}^{(1)}=\sum_{\sigma\in[m^2]_{\geq 0}^h}p^{(1)}_{\sigma}\br{\mathbf{g}}\A_{\sigma} ~\mbox{and}~\mathbf{Q}^{(1)}=\sum_{\sigma\in[m^2]_{\geq 0}^h}q^{(1)}_{\sigma}\br{\mathbf{h}}\B_{\sigma}.\]
	Item 1 follows directly.
	
	
	To prove item 2, we have that $$N_2\br{\mathbf{P}^{(1)}}=\twonorm{p^{(1)}}\leq\twonorm{p}=N_2\br{\mathbf{P}},$$
	where both equalities are from \cref{lem:randoperator}; the inequality follows from \cref{fac:smoothgaussian} item 2. The second part of item 2 in \cref{lem:smoothgaussian} follows by the same argument.
	
	For item 3, Note that both $\Lambda_1$ and $\Lambda_2$ are closed convex sets. Let $\R_1$ and $\R_2$ be rounding maps of $\Lambda_1$ and $\Lambda_2$, respectively. From \cref{lem:closedelta},
	\begin{align*}
		\twonorm{\R_1\circ p-p}^2=&\frac{1}{m^h}\expec{}{\Tr~\zeta\br{\mathbf{P}}},\\
		\twonorm{\R_2\circ q-q}^2=&\frac{1}{m^h}\expec{}{\Tr~\zeta\br{\mathbf{Q}}}, \\
		\twonorm{\R_1\circ p^{(1)}-p^{(1)}}^2=&\frac{1}{m^h}\expec{}{\Tr~\zeta\br{\mathbf{P}^{(1)}}},\\
		\twonorm{\R_2\circ q^{(1)}-q^{(1)}}^2=&\frac{1}{m^h}\expec{}{\Tr~\zeta\br{\mathbf{Q}^{(1)}}}.
	\end{align*}
	The  \cref{fac:smoothgaussian} item 3 implies that
	\[\br{\frac{1}{m^h}\expec{}{\Tr~\zeta\br{\mathbf{P}^{(1)}}}}^{1/2}\leq\br{\frac{1}{m^h}\expec{}{\Tr~\zeta\br{\mathbf{P}}}}^{1/2}+\delta\twonorm{p}.\]
	Note that $\twonorm{p}=N_2\br{\mathbf{P}}$ by \cref{lem:randoperator}. Squaring both sides of the inequality above, we conclude the first inequality in \cref{lem:smoothgaussian} item 3. The second inequality follows in the same way.
	
	To prove item 4, consider
	\begin{align*}
		&\abs{\expec{}{\Tr\br{\br{\mathbf{P}\otimes \mathbf{Q}}\psi_{AB}^{\otimes h}}}-\expec{}{\Tr\br{\br{\mathbf{P}^{\br{1}}\otimes \mathbf{Q}^{\br{1}}}\psi_{AB}^{\otimes h}}}}\\
		=&\abs{\sum_{\sigma\in[m^2]_{\geq 0}^ h}c_{\sigma}\br{\innerproduct{p_{\sigma}}{q_{\sigma}}_{\G_{\rho}^{\otimes n}}-\innerproduct{p^{(1)}_{\sigma}}{q^{(1)}_{\sigma}}_{\G_{\rho}^{\otimes n}}}}\\
		\leq&~\delta\sum_{\sigma\in[m^2]_{\geq 0}^h}\twonorm{p_{\sigma}}\twonorm{q_{\sigma}}\quad\quad\mbox{(\cref{fac:smoothgaussian} item 4 and $c_{\sigma}\leq 1$ due to \cref{lem:normofM})}\\
		\leq&~\delta\br{\sum_{\sigma\in[m^2]_{\geq 0}^h}\twonorm{p_{\sigma}}^2}^{1/2}\br{\sum_{\sigma\in[m^2]_{\geq 0}^h}\twonorm{q_{\sigma}}^2}^{1/2}\\
		=&~\delta N_2\br{\mathbf{P}}N_2\br{\mathbf{Q}}\quad\quad\mbox{(\cref{lem:randoperator})}.
	\end{align*}
	
	Item 5 is implied by Item 5 in Fact~\ref{fac:smoothgaussian}.
\end{proof}

\section{Multilinearization of random operators}\label{sec:multilinear}

The following is the main lemma in this section, which turns joint random operators $\br{\mathbf{P},\mathbf{Q}}$ to multilinear joint random operators.

\begin{lemma}\label{lem:multiliniearization}Given $0\leq \rho<1$, $\delta>0$, integers $d,h,n>0, m>1$, an $m$-dimensional noisy MES $\psi_{AB}$ with the maximal correlation $\rho= \rho\br{\psi_{AB}}$, there exists $t=O\br{\frac{d^2}{\delta^2}}$ and a map $f:L^2\br{\H_m^{\otimes h},\gamma_n}\rightarrow L^2\br{\H_m^{\otimes h},\gamma_{n\cdot t}}$ such that, for any degree-$d$ joint random operators
\begin{multline*}
		\br{\mathbf{P},\mathbf{Q}}=\br{\sum_{\sigma\in[m^2]_{\geq 0}^h}p_{\sigma}\br{\mathbf{g}}\A_{\sigma},\sum_{\sigma\in[m^2]_{\geq 0}^h}q_{\sigma}\br{\mathbf{h}}\B_{\sigma}}_{\br{\mathbf{g},\mathbf{h}}\sim\G_{\rho}^{\otimes n}}\\\in L^2\br{\H_m^{\otimes h},\gamma_n}\times L^2\br{\H_m^{\otimes h},\gamma_n},
\end{multline*}
		where $\set{\A_i}_{i=0}^{m^2-1}$ and $\set{\B_i}_{i=0}^{m^2-1}$ are standard orthonormal bases in $\M_m$,
		\begin{eqnarray*}
			&&\br{\mathbf{P}^{(1)},\mathbf{Q}^{(1)}}=\br{f\br{\mathbf{P}},f\br{\mathbf{Q}}}\\
			&=&\br{\sum_{\sigma\in[m^2]_{\geq 0}^h}p^{(1)}_{\sigma}\br{\mathbf{x}}\A_{\sigma},\sum_{\sigma\in[m^2]_{\geq 0}^h}q^{(1)}_{\sigma}\br{\mathbf{y}}\B_{\sigma}}_{\br{\mathbf{x},\mathbf{y}}\sim\G_{\rho}^{\otimes n\cdot t}}\in L^2\br{\H_m^{\otimes h},\gamma_{n\cdot t}}\times L^2\br{\H_m^{\otimes h},\gamma_{n\cdot t}}
		\end{eqnarray*}
		are multilinear joint random operators. It further holds that
		\begin{enumerate}
			\item Both $\deg\br{\mathbf{P}^{(1)}}$ and $\deg\br{\mathbf{Q}^{(1)}}$ are at most $d$.
			\item For all $\br{i,j}\in[n]\times[t]$ and $\sigma\in[m^2]_{\geq 0}^h$,
			
			$\influence_{(i-1)t+j}\br{p^{(1)}_{\sigma}}\leq\delta\cdot\influence_i\br{p_{\sigma}}~\mbox{and}~\influence_{(i-1)t+j}\br{q^{(1)}_{\sigma}}\leq\delta\cdot\influence_i\br{q_{\sigma}};$
			\item $N_2\br{\mathbf{P}^{(1)}}\leq N_2\br{\mathbf{P}}~\mbox{and}~N_2\br{\mathbf{Q}^{(1)}}\leq N_2\br{\mathbf{Q}};$
			\item $\abs{\expec{}{\Tr~\zeta\br{\mathbf{P}^{(1)}}}-\expec{}{\Tr~\zeta\br{\mathbf{P}}}}\leq 4\delta m^{h}N_2\br{\mathbf{P}}^2$  and
			
			$\abs{\expec{}{\Tr~\zeta\br{\mathbf{Q}^{(1)}}}-\expec{}{\Tr~\zeta\br{\mathbf{Q}}}}\leq 4\delta m^{h}N_2\br{\mathbf{Q}}^2;$
			\item $\abs{\expec{}{\Tr\br{\br{\mathbf{P}^{(1)}\otimes\mathbf{Q}^{(1)}}\psi_{AB}^{\otimes h}}}-\expec{}{\Tr\br{\br{\mathbf{P}\otimes\mathbf{Q}}\psi_{AB}^{\otimes h}}}}\leq\delta N_2\br{\mathbf{P}}N_2\br{\mathbf{Q}}.$
\item The map $f$ is linear and unital.
		\end{enumerate}
\end{lemma}

\begin{definition}\label{def:linear}
	Suppose $f\in L^2\br{\reals,\gamma_n}$ is given with a Hermite expansion $f=\sum_{\mathbf{\sigma}\in\mathbb{Z}_{\geq 0}^n}\widehat{f}\br{\mathbf{\sigma}}H_{\mathbf{\sigma}}$. The {\em multilinear truncation} of $f$ is defined to be the function $f^{\mathpzc{ml}}\in L^2\br{\reals,\gamma_n}$ given by
	\[f^{\mathpzc{ml}}=\sum_{\mathbf{\sigma}\in\set{0,1}^n}\widehat{f}\br{\mathbf{\sigma}}H_{\mathbf{\sigma}}\br{x}.\]
\end{definition}

\begin{fact}~\cite[Lemma 5.1]{Ghazi:2018:DRP:3235586.3235614}\label{fac:mulilinear}\footnote{There are several differences between the statement here and Lemma 4.1 in~\cite{Ghazi:2018:DRP:3235586.3235614}, which are listed below.
\begin{enumerate}
  \item The function $\bar{f}$ is obtained by replacing each Gaussian random variable $\mathbf{x}$ by $\frac{1}{\sqrt{t}}\br{\mathbf{x}^{(1)}+\cdots+\mathbf{x}^{(t)}}$ for sufficiently large $t$, where $\mathbf{x}^{(1)},\ldots, \mathbf{x}^{(t)}$ are $t$ i.i.d. Gaussian variables. Thus, $f\br{\mathbf{x}}$ and $\bar{f}\br{\bar{\mathbf{x}}}$ have the same distribution. Same for $g\br{\mathbf{y}}$ and $\bar{g}\br{\bar{\mathbf{y}}}$. This implies item 4, item 8 and the equalities in item 3.
  \item Item 5-7 follow from the analysis in ~\cite[page 47, arxiv version ]{Ghazi:2018:DRP:3235586.3235614}.
\end{enumerate}
}	
	Given parameters $\rho\in[0,1], \delta>0$ and $d\in\mathbb{Z}_{\geq 0}$, there exists $t=t\br{d,\delta}$ such that the following holds:
	
	Let $f,g\in L^2\br{\reals,\gamma_n}$ be degree-$d$ polynomials. There exist polynomials $\bar{f},\bar{g}\in L^2\br{\reals,\gamma_{nt}}$ over variables \[\bar{x}=\set{x_j^{\br{i}}:\br{i,j}\in[n]\times[t]}~\mbox{and}~\bar{y}=\set{y_j^{\br{i}}:\br{i,j}\in[n]\times[t]}\]
satisfying the following.
	\begin{enumerate}
		\item $\bar{f}^{\mathpzc{ml}}$ and $\bar{g}^{\mathpzc{ml}}$ are multilinear with degree $d$.
		\item $\var{\bar{f}^{\mathpzc{ml}}}\leq\var{f}$ and $\var{\bar{g}^{\mathpzc{ml}}}\leq\var{g}$.
		\item $\twonorm{\bar{f}^{\mathpzc{ml}}}\leq\twonorm{\bar{f}}=\twonorm{f}$ and $\twonorm{\bar{g}^{\mathpzc{ml}}}\leq\twonorm{\bar{g}}=\twonorm{g}$.
		\item Given two independent random variables $\mathbf{g}\sim \gamma_n$ and $\mathbf{x}\sim \gamma_{n\cdot t}$, the distributions of $f\br{\mathbf{g}}$ and $\bar{f}\br{\mathbf{x}}$ are identical and the distributions of $g\br{\mathbf{g}}$ and $\bar{g}\br{\mathbf{x}}$ are identical.
		
		\item $\twonorm{\bar{f}-\bar{f}^{\mathpzc{ml}}}\leq\frac{\delta}{2}\twonorm{f}$ and $\twonorm{\bar{g}-\bar{g}^{\mathpzc{ml}}}\leq\frac{\delta}{2}\twonorm{g}$
		\item For all $\br{i,j}\in[n]\times[t]$, it holds that \[\influence_{x^{\br{i}}_j}\br{\bar{f}^{\mathpzc{ml}}}\leq\delta\cdot\influence_i\br{f}~\mbox{and}~\influence_{y^{\br{i}}_j}\br{\bar{g}^{\mathpzc{ml}}}\leq\delta\cdot \influence_i\br{g}.\]
		\item $\abs{\innerproduct{\bar{f}^{\mathpzc{ml}}}{\bar{g}^{\mathpzc{ml}}}_{\G^{\otimes n\cdot t}_{\rho}}-\innerproduct{f}{g}_{\G^{\otimes n}_{\rho}}}\leq\delta\twonorm{f}\twonorm{g}$.
  \item For any constants $c_1, c_2$, it holds that $\br{\overline{c_1f+c_2g}}^{\mathpzc{ml}}=c_1\bar{f}^{\mathpzc{ml}}+c_2\bar{g}^{\mathpzc{ml}}$ and $\bar{f}^{\mathpzc{ml}}=f$ if $f$ is a constant function.
	\end{enumerate}
	In particular, we may take $t=O\br{\frac{m^2}{\delta^2}}.$
\end{fact}

We are now ready to prove \cref{lem:multiliniearization}.

\begin{proof}[Proof of \cref{lem:multiliniearization}]
	Applying \cref{fac:mulilinear} to $\set{p_{\sigma}}_{\sigma\in[m^2]_{\geq 0}^h}$ and $\set{q_{\sigma}}_{\sigma\in[m^2]_{\geq 0}^h}$ we get  $\set{\overline{p_{\sigma}}}_{\sigma\in[m^2]_{\geq 0}^h}$ and $\set{\overline{q_{\sigma}}}_{\sigma\in[m^2]_{\geq 0}^h}$.
	Let $p^{(1)}_{\sigma}\br{\cdot}=\overline{p_{\sigma}}^{\mathpzc{ml}}\br{\cdot}$ and
	$q^{(1)}_{\sigma}\br{\cdot}=\overline{q_{\sigma}}^{\mathpzc{ml}}\br{\cdot}$.
	Define
	\[\br{\mathbf{P}^{(1)},\mathbf{Q}^{(1)}}=\br{\sum_{\sigma\in[m^2]_{\geq 0}^h}p^{(1)}_{\sigma}\br{\mathbf{x}}\A_{\sigma},\sum_{\sigma\in[m^2]_{\geq 0}^h}q^{(1)}_{\sigma}\br{\mathbf{y}}\B_{\sigma}}_{\br{\mathbf{x},\mathbf{y}}\sim\G_{\rho}^{\otimes n\cdot t}}\]
	Item 1 and item 2 are implied by \cref{fac:mulilinear} item 1 and item 6, respectively.
	Item 3 follows from \cref{lem:randoperator} and the item 3 in \cref{fac:mulilinear}.

	We will prove the first inequality in item 4. The second inequality follows from the same argument. Set $\Lambda=\Lambda\br{\set{\A_{\sigma}}_{\sigma\in[m^2]_{\geq 0}^h}}$ as in \cref{eqn:lambda}, which is a closed convex set. Let $\R$ be a rounding map of $\Lambda$. Note that $0\in\Lambda$. Thus from \cref{fac:rounding}, for all $x\in\reals^{m^{2h}}$, we have
\begin{equation}\label{eqn:contraction}
\twonorm{\R(x)}\leq\twonorm{x}
\end{equation}

Define vector-valued functions
\[p=\br{p_{\sigma}}_{\sigma\in[m^2]_{\geq 0}^h}, \bar{p}=\br{\bar{p}_{\sigma}}_{\sigma\in[m^2]_{\geq 0}^h}~\mbox{and}~p^{(1)}=\br{p^{(1)}_{\sigma}}_{\sigma\in[m^2]_{\geq 0}^h}.\]
By \cref{lem:closedelta},
	\[\twonorm{p-\R\circ p}^2=\frac{1}{m^h}\expec{}{\Tr~\zeta\br{\mathbf{P}}}~\mbox{and}~\twonorm{p^{(1)}-\R\circ p^{(1)}}^2=\frac{1}{m^h}\expec{}{\Tr~\zeta\br{\mathbf{P}^{(1)}}}.\]
	Hence
	\begin{align*}
		&\frac{1}{m^h}\abs{\expec{}{\Tr~\zeta\br{\mathbf{P}^{(1)}}}-\expec{}{\Tr~\zeta\br{\mathbf{P}}}}\\ \\
		=&\abs{\twonorm{p^{(1)}-\R\circ p^{(1)}}^2-\twonorm{p-\R\circ p}^2}\\
		=&\abs{\twonorm{\bar{p}^{\mathpzc{ml}}-\R\circ \bar{p}^{\mathpzc{ml}}}^2-\twonorm{\bar{p}-\R\circ \bar{p}}^2}\quad\quad\mbox{(\cref{fac:mulilinear} item 4)}\\
		=&\abs{\br{\twonorm{\bar{p}^{\mathpzc{ml}}-\R\circ \bar{p}^{\mathpzc{ml}}}-\twonorm{\bar{p}-\R\circ \bar{p}}}\br{\twonorm{\bar{p}^{\mathpzc{ml}}-\R\circ \bar{p}^{\mathpzc{ml}}}+\twonorm{\bar{p}-\R\circ \bar{p}}}}\\
\leq&\abs{\br{\twonorm{\bar{p}^{\mathpzc{ml}}-\R\circ \bar{p}^{\mathpzc{ml}}}-\twonorm{\bar{p}-\R\circ \bar{p}}}\br{\twonorm{\bar{p}^{\mathpzc{ml}}}+\twonorm{\R\circ \bar{p}^{\mathpzc{ml}}}+\twonorm{\bar{p}}+\twonorm{\R\circ \bar{p}}}}\\
		\leq&4\twonorm{p}\abs{\twonorm{\bar{p}^{\mathpzc{ml}}-\R\circ \bar{p}^{\mathpzc{ml}}}-\twonorm{\bar{p}-\R\circ \bar{p}}}\quad\quad\mbox{(\cref{fac:mulilinear} item 3 and Eq.~\cref{eqn:contraction})}\\
		\leq&4\twonorm{p}\br{\twonorm{\bar{p}-\bar{p}^{\mathpzc{ml}}}+\twonorm{\R\circ\bar{p}-\R\circ\bar{p}^{\mathpzc{ml}}}}\quad\quad\mbox{(Triangle inequality)}\\
		\leq&8\twonorm{p}\twonorm{\bar{p}-\bar{p}^{\mathpzc{ml}}}\quad\quad\mbox{(\cref{fac:rounding})}\\
		\leq&4\delta\twonorm{p}^2\quad\quad\mbox{(\cref{fac:mulilinear} item 5)}\\
		=&4\delta N_2\br{\mathbf{P}}^2\quad\quad\mbox{(\cref{lem:randoperator})}.
	\end{align*}
	
	To prove item 5, consider
	\begin{align*}
		&\abs{\expec{}{\Tr\br{\br{\mathbf{P}\otimes \mathbf{Q}}\psi_{AB}^{\otimes h}}}-\expec{}{\Tr\br{\br{\mathbf{P}^{\br{1}}\otimes \mathbf{Q}^{\br{1}}}\psi_{AB}^{\otimes h}}}}\\
		=&\abs{\sum_{\sigma\in[m^2]_{\geq 0}^ h}c_{\sigma}\br{\innerproduct{p_{\sigma}}{q_{\sigma}}_{\G_{\rho}^{\otimes n}}-\innerproduct{p^{(1)}_{\sigma}}{q^{(1)}_{\sigma}}_{\G_{\rho}^{\otimes n\cdot t}}}}\\
		\leq&\delta\sum_{\sigma\in[m^2]_{\geq 0}^{ h}}\twonorm{p_{\sigma}}\twonorm{q_{\sigma}}\quad\quad\mbox{(\cref{fac:mulilinear} item 7)}\\
		\leq&\delta\br{\sum_{\sigma\in[m^2]_{\geq 0}^h}\twonorm{p_{\sigma}}^2}^{1/2}\br{\sum_{\sigma\in[m^2]_{\geq 0}^h}\twonorm{q_{\sigma}}^2}^{1/2}\\
		=&\delta N_2\br{\mathbf{P}}N_2\br{\mathbf{Q}}\quad\quad\mbox{(\cref{lem:randoperator})}.
	\end{align*}	
	
	Item 6 is implied by Item 8 in Fact~\ref{fac:mulilinear}.
\end{proof}

	\section*{Acknowledgments}

This work is supported by the National Key R\&D Program of China 2018YFB1003202, National Natural Science Foundation of China (Grant No. 61972191), Program for Innovative Talents and Entrepreneur in Jiangsu, the Fundamental Research Funds for the Central Universities 0202/14380068 and Anhui Initiative in Quantum Information Technologies Grant No. AHY150100. Part of the work was done when the second author was a Hartree postdoctoral fellow at QuICS, University of Maryland. The authors thank Thomas Vidick pointing out that a union bound on the question sets was missing in the previous version. The authors also thank Hong Zhang for the helpful discussion; Pritish Kamath and Ashley Montanaro for the correspondence; Srinivasan Arunachalam, Ziyi Guan and Sandy Irani's comments and the anonymous reviewers' helpful feedback.

\appendix

	\section{Facts on Fr\'{e}chet derivatives}\label{sec:frechet}
In this section, we summarize some basic facts on Fr\'echet derivatives.
\begin{fact}\label{fac:frechetderivative}
	Given $f,g:\M_m\rightarrow \M_m$ and $P,Q_1,\ldots, Q_k\in \M_m$, it holds that
	\begin{enumerate}
		\item $D\br{f+g}\br{P}\Br{Q}=Df\br{P}\Br{Q}+Dg\br{P}\Br{Q}$.
		
		\item $D\br{f\cdot g}\br{P}\Br{Q}=Df\br{P}\Br{Q}\cdot g\br{P}+f\br{P}\cdot Dg\br{P}\Br{Q}$.
		
		\item $D\br{g\circ f}\br{P}\Br{Q}=\br{Dg\br{f\br{P}}\circ Df\br{P}}\Br{Q}$. Here we are treating $Df(P)$ as
a function mapping a matrix $Q$ to a matrix $Df(P)(Q)$, and $\circ$ means composition.
		
		\item $D^kf\br{P}\Br{Q_1,\ldots, Q_k}=D^kf\br{P}\Br{Q_{\sigma\br{1}},\ldots, Q_{\sigma\br{k}}}$ for any integer $k>0$ and permutation $\sigma\in S_k$.	
	\end{enumerate}
\end{fact}

The following fact follows from elementary matrix calculations. Readers who are interested may refer to~\cite [Chapter X.4]{Bhatia}.
\begin{fact}~\cite[Page 311, Example X.4.2]{Bhatia}\label{fac:Bhatia}
	\begin{enumerate}
		\item Let $f\br{x}=x^2$. Then
		\[Df\br{P}\Br{Q}=\anticommutator{P}{Q},\]
where $\anticommutator{P}{Q}=PQ+QP$ is the anticommutator of $A$ and $B$.

		\item Let $f\br{x}=x^{-1}$. Then for any invertible $P$,
		\[Df\br{P}\Br{Q}=-P^{-1}QP^{-1}.\]
	\end{enumerate}
\end{fact}

\begin{fact}\label{fac:frechet}~\cite[Page 124, Theorem V.3.3]{Bhatia}
	Let $f\in\C^1(\reals)$ and $P,Q\in\H_m$. Suppose that $P$ has a spectral decomposition $U\Lambda U^\dagger$, where $\Lambda=\mathrm{diag}\br{\lambda_1,\lambda_2,\dots,\lambda_m}$ and $U$ is unitary. Then
	$$Df(P)[Q]=U\br{f^{[1]}(\Lambda)\circ\br{U^\dagger QU}}U^\dagger$$
where
$$
f^{[1]}(\Lambda)_{ij}=\begin{cases}
\frac{f\br{\lambda_i}-f\br{\lambda_j}}{\lambda_i-\lambda_j}&\mbox{if $\lambda_i\ne \lambda_j$}\\
f'\br{\lambda_i}&\mbox{otherwise}
\end{cases}
$$
\end{fact}

\begin{lemma}\label{lem:taylor}
	Let $f$ be a real function on $[a,b]$ such that $f^{(n-1)}$ is continuous on $[a,b]$ and $f^{\br{n}}\br{t}$ exists for all $t\in\br{a,b}$ except for a finite number of points $\set{t_1,\ldots, t_m}\linebreak\subseteq(a,b)$. Moreover, assume that $\abs{f^{\br{n}}\br{t}}\leq M$ for all $t\in\br{a,b}$ and $t\notin\set{t_1,\ldots, t_m}$. Then for any distinct points $\alpha,\beta$ in $[a,b]$, we have
	\[\abs{f\br{\beta}-P\br{\beta}}\leq\frac{M}{n!}\abs{\beta-\alpha}^n,\]
	where
	\[P\br{t}=\sum_{k=0}^{n-1}\frac{f^{\br{k}}\br{\alpha}}{k!}\br{t-\alpha}^k.\]
\end{lemma}

\begin{proof}
	Let $L$ be the number satisfying that
	\[f\br{\beta}=P\br{\beta}+\frac{L}{n!}\br{\beta-\alpha}^n.\]
	It suffices to show that $\abs{L}\leq M$. Set
	\[g\br{t}= f\br{t}-P\br{t}-\frac{L}{n!}\br{t-\alpha}^n.\]
	We assume that $t_1<t_2<\ldots<t_m$, without loss of generality.
	Then
	\[g\br{\alpha}=g'\br{\alpha}=\ldots=g^{\br{n-1}}\br{\alpha}=0.\]
	Note that $g\br{\beta}=0$. By the mean value theorem, $g'\br{\beta_1}=0$ for some $\beta_1\in(\alpha,\beta)$. Repeat this for $n-1$ steps, we get $\beta_{n-1}\in\br{\alpha,\beta}$ such that $g^{\br{n-1}}\br{\beta_{n-1}}=0$.
	Note that
	\[g^{\br{n-1}}\br{t}=f^{\br{n-1}}\br{t}-f^{\br{n-1}}\br{\alpha}-L\br{t-\alpha}.\]
	Set $t_0=\alpha$.
	Let $i_0$ be the largest integer such that $t_{i_0}<\beta_{n-1}$. Then
	\begin{multline*}g^{\br{n-1}}\br{\beta_{n-1}}=\br{f^{\br{n-1}}\br{\beta_{n-1}}-f^{\br{n-1}}\br{t_{i_0}}}\\+\sum_{i=0}^{i_0-1}\br{f^{\br{n-1}}\br{t_{i+1}}-f^{\br{n-1}}\br{t_i}}-L\br{t-\alpha}.
\end{multline*}
	Applying the mean value theorem, we have
	\[g^{\br{n-1}}\br{\beta_{n-1}}=f^{\br{n}}\br{\xi_{i_0}}\br{\beta_{n-1}-t_{i_0}}+\sum_{i=0}^{i_0-1}f^{\br{n}}\br{\xi_i}\br{t_{i+1}-t_i}-L\br{\beta-\alpha},\]
	where $\xi_{i_0}\in[t_{i_0},\beta]$ and $\xi_i\in[t_i,t_{i+1}]$.
	As $g^{\br{n-1}}\br{\beta_{n-1}}=0$ and $\abs{f^{\br{n}}\br{t}}\leq M$ for any $t$ where $f^{\br{n}}\br{t}$ is defined, we have
	\[\abs{L}\br{\beta-\alpha}\leq \abs{M\br{\beta-\alpha}}.\]
	Thus $\abs{L}\leq M$.
\end{proof}

\section{Proofs of \cref{lem:zetataylor} and \cref{lem:zetaadditivity}}\label{sec:zetataylor}

Before proving \cref{lem:zetataylor} and \cref{lem:zetaadditivity}, we first introduce the Lyapunov equation, a well studied equation in control theory~\cite{doi:10.1080/00207179208934253}.

\begin{definition}\label{def:sylvester}
	Let $P,Q$ be two Hermitian matrices in $\H_m$. We define the Lyapunov equation.
	\begin{equation}\label{eqn:lyapunov}
	PX+XP=Q.
	\end{equation}
	The solution to Eq.~\cref{eqn:lyapunov} is denoted by $L\br{P,Q}$.
\end{definition}

\begin{lemma}\label{lem:lyapunovsol}
	Given Hermitian matrices $P, Q\in\H_m$, the Lyapunov equation ~\cref{eqn:lyapunov} has a unique solution if and only if  $I_m\otimes P+P\otimes I_m$ is invertible. Note that the eigenvalues of $I_m\otimes P+P\otimes I_m$ are $\set{\lambda_i\br{P}+\lambda_j\br{P}}_{1\leq i,j\leq m}$. Thus, it is equivalent to the fact that $P$ and $-P$ have no common eigenvalues.
	
	Moreover, let $P=UDU^{\dagger}$ be a spectral decomposition of $P$, where \linebreak$D=\textsf{Diag}\br{d_1,\ldots,d_n}$ satisfies that $d_i+d_j\neq 0$ for any $0\leq i, j\leq n$. Then Eq.~\cref{eqn:lyapunov} has a unique solution $X_0$ and it satisfies that
	\[\br{U^{\dagger}X_0U}_{i,j}=\frac{\br{U^{\dagger}QU}_{i,j}}{d_i+d_j}.\]
\end{lemma}
\begin{proof}
	Let  $X'= U^{\dagger}XU$ and $Q'= U^{\dagger}QU$. Then we have
	\[DX'+X'D=Q',\]
	which is equivalent to
	\[\br{d_i+d_j}X'_{ij}=Q'_{ij},\]
	for $1\leq i,j\leq n$.
	Hence it has a unique solution if and only if  $d_i+d_j\neq 0$ for all $i, j$.
\end{proof}

\begin{fact}\label{fac:lysol2}~\cite[Page 205, Theorem VII.2.3]{Bhatia}
	Let $P$ be a positive definite matrix. Then
	\[L\br{P,Q}=\int_{0}^{\infty}e^{-tP}Qe^{-tP}dt.\]
\end{fact}

\begin{definition}\label{def:hfunction}
	For any Hermitian matrices $P, Q$ such that $P$ is invertible, we define
	\[\ell_Q\br{P}= L\br{\abs{P},PQ+QP},\]
	where $\abs{P}=\sqrt{P^2}$.
\end{definition}
\noindent It is easy to verify that $\ell_Q\br{P}=Q$ if $P>0$.

\begin{definition}\label{def:kappa}
	Given $P,Q\in\H_m$, suppose $P$ has a spectral decomposition $U\Lambda U^\dagger$, where $\Lambda=\mathrm{diag}\br{\lambda_1,\lambda_2,\dots,\lambda_m}$ and $U$ is unitary. Then define
	
	\[\kappa_Q\br{P}= U\br{\Lambda'\circ \br{U^\dagger QU}}U^\dagger\]
where $\circ$ represents the Hadamard product (a.k.a. entry-wise product) and $\Lambda'\in\H_m$ is defined as
\begin{equation}\label{eqn:Lambda}
  \Lambda'_{ij}=\begin{cases}
\frac{\br{\lambda_i+\lambda_j}^2}{\abs{\lambda_i}+\abs{\lambda_j}}&\mbox{if $\lambda_i\ne0$ or $\lambda_j\ne0$}\\
0&\lambda_i=\lambda_j=0.
\end{cases}
\end{equation}
\end{definition}

Recall that the anticommutator of $A$ and $B$ is $\anticommutator{A}{B}=AB+BA$.
\begin{lemma}\label{lem:kappa}
For any $P,Q\in\H_m$, we have
$$\kappa_Q(P)=\int_{0}^{\infty}\anticommutator{P}{e^{-t\abs{P}}\br{PQ+QP}e^{-t\abs{P}}}dt.$$
In particular, if $P=\mathrm{diag}\br{\lambda_1,\lambda_2,\dots,\lambda_m}$ is diagonal, then
\[\kappa_Q\br{P}_{ij}=\Lambda'_{ij}\cdot Q_{ij},\]
where $\Lambda'$ is defined in Eq.~\cref{eqn:Lambda}.
\end{lemma}
\begin{proof}

Suppose that $P$ has a spectral decomposition $U\Lambda U^\dagger$, where $\Lambda=$\linebreak$\mathrm{diag}\br{\lambda_1,\lambda_2,\dots,\lambda_m}$ and $U$ is unitary. Let $Q'=U^\dagger QU$, then
\begin{eqnarray*}
  &&\br{U^{\dagger}\int_{0}^{\infty}\anticommutator{P}{e^{-t\abs{P}}\br{PQ+QP}e^{-t\abs{P}}}dt U}_{ij}\\
  &=& \br{\int_{0}^{\infty}\anticommutator{\Lambda}{e^{-t\abs{\Lambda}}\br{\Lambda Q'+Q'\Lambda}e^{-t\abs{\Lambda}}}dt}_{ij}\\
  &=& \int_{0}^{\infty}\br{\lambda_i+\lambda_j}^2e^{-t\br{\abs{\lambda_i}+\abs{\lambda_j}}}Q'_{ij}~dt \\
  &=&\Lambda'_{ij}Q'_{ij}.
\end{eqnarray*}
We conclude the result.

\end{proof}

\begin{lemma}\label{lem:derivative}
	Let $P, Q$ be Hermitian matrices. The following holds.
	\begin{enumerate}
		\item Let $f\br{x}= \sqrt{x}$ for $x\geq 0$. Then
		$Df\br{P}\Br{Q}=L\br{\sqrt{P}, Q}$ if $P$ is positive definite.		
		\item Let $f\br{x}= \abs{x}$. Then
		$Df\br{P}\Br{Q}=\ell_Q\br{P}$ when $P$ is invertible.
		
		\item Let $f\br{x}=x\abs{x}$. Then    $Df\br{P}\Br{Q}=\frac{1}{2}\br{\anticommutator{\abs{P}}{Q}+\kappa_Q\br{P}}.$
		
		\item Let $p\br{x}=\begin{cases}
		x^2~\mbox{if $x\geq 0$}\\
		0~\mbox{otherwise}.
		\end{cases}$
		Then 		\[Dp\br{P}\Br{Q}=\frac{1}{2}\anticommutator{P}{Q}+\frac{1}{4}\anticommutator{\abs{P}}{Q}+\frac{1}{4}\kappa_Q\br{P}.\]
	\end{enumerate}
\end{lemma}

\begin{proof}

	\begin{enumerate}
		\item Let $g\br{x}= x^2$ for $x\in\reals$ and $X=Df\br{P}\Br{Q}$. Applying the composition rule in \cref{fac:frechetderivative}, we have
		\[Q=\br{Dg\br{f\br{P}}\circ Df\br{P}}\Br{Q}=Dg\br{\sqrt{P}}\Br{X}=\anticommutator{\sqrt{P}}{X}.\]
		By \cref{def:sylvester}, $X=L\br{\sqrt{P}, Q}$.
		
		\item 	Let $g\br{x}=x^2$ for $x\in\reals$ and $h\br{x}=\sqrt{x}$ for $x\geq0$. Then $f=h\circ g$. Again applying the composition rule in \cref{fac:frechetderivative}, we have
		\[Df\br{P}\Br{Q}=\br{Dh\br{g\br{P}}\circ Dg\br{P}}\Br{Q}=Dh\br{P^2}\Br{PQ+QP}=\ell_Q\br{P},\]
where the last equality follows from ~\cref{def:hfunction}.

		\item We assume that $P=\mathrm{Diag}(\lambda_1,\ldots,\lambda_n)$ is a diagonal matrix without loss of generality. From \cref{fac:frechet},
\begin{eqnarray*}
  &&f^{[1]}\br{P}_{ij}=\begin{cases}
                       \frac{\lambda_i\abs{\lambda_i}-\lambda_j\abs{\lambda_j}}{\lambda_i-\lambda_j}, & \mbox{if $\lambda_i\neq\lambda_j$ }  \\
                       2\abs{\lambda_i}, & \mbox{otherwise}.
                     \end{cases} \\
  &=& \begin{cases}
        \frac{1}{2}\br{\abs{\lambda_i}+\abs{\lambda_j}+\frac{\br{\lambda_i+\lambda_j}^2}{\abs{\lambda_i}+\abs{\lambda_j}}}, & \mbox{if $\lambda_i\neq\lambda_j$}  \\
        2\abs{\lambda_i}, & \mbox{otherwise}.
      \end{cases}
\end{eqnarray*}
		Combing with \cref{lem:kappa}, we conclude the result.
		
		\item It follows from the fact that $f\br{x}=\frac{1}{2}x^2+\frac{1}{2}x\abs{x}$, the previous item in this lemma, \cref{fac:Bhatia} item 1 and the linearity of Fr\'echet derivatives guaranteed by \cref{fac:frechetderivative} item 1.
	\end{enumerate}
\end{proof}

\begin{fact}\cite[Page 215, Corollary VII.5.6]{Bhatia}
$\forall A,B\in\H_m$,
\[\norm{\abs{A}-\abs{B}}_2\leq\norm{A-B}_2.\]
\end{fact}
\begin{remark}
Since for all $A\in\H_m$\[\norm{A}\leq\norm{A}_2\leq\sqrt{m}\norm{A},\]
we have
\begin{equation}\label{eqn:abs}
\norm{\abs{A}-\abs{B}}\leq\sqrt{m}\norm{A-B}
\end{equation}
\end{remark}
\begin{fact}\label{fac:exp}\cite[Page 502, Theorem 6.5.29]{rahorn1991topics}
$\forall A,B\in\H_m$,\[\norm{e^{A+B}-e^{A}}\leq\br{e^{\norm{B}}-1}\norm{e^A}\]
\end{fact}
\begin{lemma}
For all $P,Q\in\H_m$ that $P$ is invertible, $D\ell_Q(P)[Q]$ exists.
\end{lemma}
\begin{proof}
\[\ell_Q\br{P}= L\br{\abs{P},PQ+QP}=\int_{0}^{\infty}e^{-x\abs{P}}(PQ+QP)e^{-x\abs{P}}dx\]
Since $P$ is invertible, $P+tQ$ has eigenvalue 0 only for finite choices of $t$ and let $\delta$ be the minimum of their absolute values. Let $\set{a_n}$ be any sequence that converges to zero and $a_n\in(-\delta,0)\cup(0,\delta)$ for all $n$,
\[f_n(x)=\frac{e^{-x\abs{P+a_nQ}}\anticommutator{P+a_nQ}{Q} e^{-x\abs{P+a_nQ}}-e^{-x\abs{P}}\anticommutator{P}{Q} e^{-x\abs{P}}}{a_n}\]
Let $f_{n,i,j}(x)$ be the $(i,j)$-entry of $f_n(x)$. We have $\abs{f_{n,i,j}(x)}\leq\norm{f_n(x)}$. Thus by Lebesgue's Dominated Convergence Theorem~\cite[Page 26, 1.34]{rudin1987real} and Heine's Theorem~\cite[Page 186, Theorem 1]{brannan2006first}, it suffices to show that there exists a function $g$ satisfying $\int_0^\infty\abs{g}d\mu<\infty$ such that for all $x\in(0,+\infty)$
\[\norm{f_n(x)}\leq g(x)\]
\begin{align*}
\norm{f_n(x)}=&\norm{\frac{e^{-x\abs{P+a_nQ}}\anticommutator{P+a_nQ}{Q} e^{-x\abs{P+a_nQ}}-e^{-x\abs{P}}\anticommutator{P}{Q} e^{-x\abs{P}}}{a_n}}\\
=&\left\|\frac{e^{-x\abs{P+a_nQ}}\anticommutator{P}{Q} e^{-x\abs{P+a_nQ}}-e^{-x\abs{P}}\anticommutator{P}{Q} e^{-x\abs{P}}}{a_n}\right.\\
&+\left.2e^{-x\abs{P+a_nQ}}Q^2e^{-x\abs{P+a_nQ}}\right\|\\
\leq&\norm{\frac{e^{-x\abs{P+a_nQ}}\anticommutator{P}{Q} \br{e^{-x\abs{P+a_nQ}}-e^{-x\abs{P}}}}{a_n}}\\
&+\norm{\frac{\br{e^{-x\abs{P+a_nQ}}-e^{-x\abs{P}}}\anticommutator{P}{Q} e^{-x\abs{P}}}{a_n}}\\&+2\norm{e^{-x\abs{P+a_nQ}}}^2\norm{Q}^2\quad\quad\mbox{$\br{\norm{AB}\leq\norm{A}\norm{B}}$}\\
\leq&2\norm{P}\norm{Q}\norm{\frac{\br{e^{-x\abs{P+a_nQ}}-e^{-x\abs{P}}}}{a_n}}\br{\norm{e^{-x\abs{P+a_nQ}}}+\norm{e^{-x\abs{P}}}}\\
&+2\norm{e^{-x\abs{P+a_nQ}}}^2\norm{Q}^2\\
\leq&2\frac{e^{x\norm{\abs{P+a_nQ}-\abs{P}}}-1}{\abs{a_n}}\norm{P}\norm{Q}\norm{e^{-x\abs{P}}}\br{\norm{e^{-x\abs{P+a_nQ}}}+\norm{e^{-x\abs{P}}}}\\
&+2\norm{e^{-x\abs{P+a_nQ}}}^2\norm{Q}^2\quad\quad\mbox{(\cref{fac:exp})}\\
\leq&2\frac{e^{x\sqrt{m}\abs{a_n}\norm{Q}}-1}{\abs{a_n}}\norm{P}\norm{Q}\norm{e^{-x\abs{P}}}\br{\norm{e^{-x\abs{P+a_nQ}}}+\norm{e^{-x\abs{P}}}}\\
&+2\norm{e^{-x\abs{P+a_nQ}}}^2\norm{Q}^2\quad\quad\mbox{(Eq.~\cref{eqn:abs})}\\
\leq&4x\sqrt{m}\norm{P}\norm{Q}^2\norm{e^{-x\abs{P}}}\br{\norm{e^{-x\abs{P+a_nQ}}}+\norm{e^{-x\abs{P}}}}\\
&+2\norm{e^{-x\abs{P+a_nQ}}}^2\norm{Q}^2\\
\leq&4x\sqrt{m}\norm{P}\norm{Q}^2e^{-x\lambda_{\min}\br{P}}\br{e^{-x\lambda_{min}\br{P+a_nQ}}+e^{-x\lambda_{\min}\br{P}}}\\
&+2e^{-2x\lambda_{\min}\br{P+a_nQ}}\norm{Q}^2.
\end{align*}
It is easy to see
\[\int_0^\infty\abs{4x\sqrt{m}\norm{P}\norm{Q}^2e^{-x\lambda_{\min}\br{P}}\br{e^{-x\lambda_{min}\br{P+a_nQ}}+e^{-x\lambda_{\min}\br{P}}}\atop +2e^{-2x\lambda_{\min}\br{P+a_nQ}}\norm{Q}^2}d\mu<\infty.\]
Note that $\lambda_{\min}\br{P}>0$ and $\lambda_{\min}\br{P+a_nQ}>0$. By Lebesgue's Dominated Convergence Theorem, we conclude the result.
\end{proof}

\begin{lemma}\label{lem:derivativeh}
	Let $P, Q$ be Hermitian matrices where $P$ is invertible. It holds that
	\begin{equation}\label{eqn:dh}
	D\ell_Q\br{P}\Br{Q}=L\br{\abs{P},2Q^2-2\ell_Q\br{P}^2}.
	\end{equation}
	Moreover, if $P=\textsf{Diag}\br{a_1,\ldots,a_m}$ is diagonal, then
	\begin{equation}\label{eqn:hb}
	\br{\ell_Q\br{P}}_{i,j}=\frac{Q_{ij}\br{a_i+a_j}}{\abs{a_i}+\abs{a_j}}.
	\end{equation}
	\begin{equation}\label{eqn:Ds}
	\br{D\ell_Q\br{P}\Br{Q}}_{i,j}=2\frac{\sum_kQ_{ik}Q_{kj}\br{1-\frac{\br{a_i+a_k}\br{a_k+a_j}}{\br{\abs{a_i}+\abs{a_k}}\br{\abs{a_k}+\abs{a_j}}}}}{\abs{a_i}+\abs{a_j}}.
	\end{equation}
	
\end{lemma}
\begin{proof}
	From the definition of $\ell_Q\br{\cdot}$ in \cref{def:hfunction}, we have
	\[\abs{P}\ell_Q\br{P}+\ell_Q\br{P}\abs{P}=PQ+QP.\]
	Taking the Fr\'echet derivatives on both sides with respect to $Q$, we have
	\[\abs{P}D\ell_Q\br{P}\Br{Q}+D\ell_Q\br{P}\Br{Q}\abs{P}=2Q^2-2\ell_Q\br{P}^2.\]
	We conclude Eq.~\cref{eqn:dh}.
	%
\end{proof}

\begin{fact}~\cite{Davies:1988}[Corollary 5]\label{fac:davis}
	Given $a_1,\ldots, a_m,b_1,\ldots b_m>0$, let $M$ be a $d\times d$ matrix defined to be $M_{ij}=\frac{a_i-b_j}{a_i+b_j}$. For any $d\times d$ matrix $A$, it holds that
	\[\norm{A\circ M}_4\leq c\norm{A}_4,\]
	for some absolute constant $c$.
\end{fact}


\begin{lemma}\label{lem:aiaj}
	Given nonzero reals $a_1,\ldots, a_m$, let $M$ be a $d\times d$ Hermitian matrix defined to be $M_{ij}=\frac{a_i+a_j}{\abs{a_i}+\abs{a_j}}$. For any $d\times d$ Hermitian matrix $A$, it holds that
	\[\twonorm{M\circ A}\leq \twonorm{A},\]
	and
	\[\norm{M\circ A}_4\leq c\norm{A}_4,\]
	where $c\geq 1$ is an absolute constant.
\end{lemma}
\begin{proof}
	Note that $\twonorm{A}^2=\sum_{ij}\abs{A\br{i,j}}^2$. The first inequality follows from the fact that $\abs{M\br{i,j}}\leq 1$ for all $i,j$.
	
To prove the second inequality, we may assume without loss of generality that $a_1,\ldots, a_s> 0$ and $a_{s+1},\ldots, a_m<0$.
	Let $A=\begin{pmatrix}
	A_1 & A_2\\
	A_2^{\dagger} & A_3
	\end{pmatrix},$
	where $A_1,A_2,A_3$ are of size $s\times s$, $s\times (d-s)$ and $(d-s)\times (d-s)$, respectively. Let $M=\begin{pmatrix}
	M_1 & M_2\\
	M_2^{\dagger} & M_3
	\end{pmatrix}$ be the same block structure.
	Let $P$ be a $d\times d$ matrix defined to be
	\[P_{ij}=\frac{\abs{a_i}-\abs{a_j}}{\abs{a_i}+\abs{a_j}}.\]
	Then
	\[\begin{pmatrix}
	0 & A_2\\
	A_2^{\dagger} & 0
	\end{pmatrix}\circ P=\begin{pmatrix}
	0 & A_2\circ M_2\\
	-A_2^{\dagger}\circ M_2^{\dagger} & 0
	\end{pmatrix}.\]
	\cref{fac:davis} implies that
	\[\norm{\begin{pmatrix}
		0 & A_2\circ M_2\\
		-A_2^{\dagger}\circ M_2^{\dagger} & 0
		\end{pmatrix}}_4\leq c \norm{\begin{pmatrix}
		0 & A_2\\
		A_2^{\dagger} & 0
		\end{pmatrix}}_4,\]
	for some absolute constant $c$. Thus
	\[\norm{\begin{pmatrix}
		0 & A_2\circ M_2\\
		A_2^{\dagger}\circ M_2^{\dagger} & 0
		\end{pmatrix}}_4\leq c \norm{\begin{pmatrix}
		0 & A_2\\
		A_2^{\dagger} & 0
		\end{pmatrix}}_4.\]
	Then
	\begin{align*}
		\norm{A\circ M}_4&\leq\norm{\begin{pmatrix}
				A_1 & 0\\
				0 & A_3
		\end{pmatrix}}_4+\norm{\begin{pmatrix}
				0 & A_2\circ M_2\\
				A_2^{\dagger}\circ M_2^{\dagger} & 0
		\end{pmatrix}}_4\\&\leq \norm{\begin{pmatrix}
				A_1 & 0\\
				0 & A_3
		\end{pmatrix}}_4+ c \norm{\begin{pmatrix}
				0 & A_2\\
				A_2^{\dagger} & 0
		\end{pmatrix}}_4\\
		&\leq (c+1)\norm{A}_4,
	\end{align*}
	where the last inequality follows from the fact that
	\[\norm{\begin{pmatrix}
		A_1 & 0\\
		0 & A_3
		\end{pmatrix}}_4\leq\norm{A}_4~\mbox{and}~\norm{\begin{pmatrix}
		0 & A_2\\
		A_2^{\dagger} & 0
		\end{pmatrix}}_4\leq\norm{A}_4.\]
\end{proof}

\begin{lemma}\label{lem:dlq3}
	Let $P$ and $Q$ be Hermitian matrices where $P$ is invertible. It holds that
	\[\twonorm{\ell_Q\br{P}}\leq\twonorm{Q},\]
	and
	\[\norm{\ell_Q\br{P}}_4\leq c\norm{Q}_4,\]
	for some absolute constant $c\geq 1$.
\end{lemma}

\begin{proof}
	Without loss of generality, we assume that $P=\mathsf{Diag}\br{a_1,\ldots, a_m}$ is a diagonal matrix. Define $M$ be a $d\times d$ matrix  such that $M_{ij}=\frac{a_i+a_j}{\abs{a_i}+\abs{a_j}}$. Then by Eq.~\cref{eqn:hb} and \cref{lem:aiaj},
	\begin{eqnarray*}
		&&\twonorm{\ell_Q\br{P}} =\twonorm{Q\circ M}\leq\twonorm{Q};\\
		&&\norm{\ell_Q\br{P}}_4=\norm{Q\circ M}_4\leq c\norm{Q}_4
	\end{eqnarray*}
for some absolute constant $c\geq 1$.
\end{proof}

\begin{lemma}\label{lem:trlkappa}
	For any Hermitian matrices $P$ and $Q$, it holds that
	\begin{enumerate}
		\item $\Tr~\kappa_Q\br{P}=2\Tr~\abs{P}Q$.
		\item $\Tr~P\kappa_Q\br{P}=2\Tr~\abs{P}PQ$.
	\end{enumerate}
\end{lemma}
\begin{proof}
	Without loss of generality, we assume that $P=\mathsf{Diag}\br{a_1,\ldots, a_m}$ is a diagonal matrix. Then $\kappa_Q\br{P}_{i,j}=\frac{Q_{ij}\br{a_i+a_j}^2}{\abs{a_i}+\abs{a_j}}$ by \cref{lem:kappa}. Thus we have
	\[\Tr~\kappa_Q\br{P}=2\sum_i\abs{a_i}Q_{ii}=2\Tr~\abs{P}Q.\]
	For the second equality, consider
	\[\Tr~P\kappa_Q\br{P}=2\sum_i\abs{a_i}a_iQ_{ii}=2\Tr~\abs{P}PQ.\]
\end{proof}

Before proving \cref{lem:zetataylor}, we need to compute the first three orders of Fr\'echet derivatives of the function
\begin{equation}\label{eqn:qfunction}
q\br{x}=\begin{cases}
x^3~&\mbox{if $x\geq 0$}\\
0~&\mbox{otherwise}.
\end{cases}
\end{equation}

\begin{lemma}\label{lem:dq}
	Given integers $d,m>0$ and $P,Q\in\H_m$, let $f(t)=\Tr~q\br{P+tQ}$. Then $f', f''$ exist on $\reals$ and $f'''$ exists except for a finite number of points.
	
	Moreover, it holds that
	\begin{eqnarray}
	&&f'\br{0}=\Tr~\br{Qp\br{P}+P^2Q
		+P\abs{P}Q}; \label{eqn:dzetaprime}\\
	&&f''\br{0}=\Tr~\br{4PQ^2+\frac{3}{2}\abs{P}Q^2+\frac{3}{4}Q\kappa_Q\br{P}};\label{eqn:dzetaprimetwo}
	\end{eqnarray}
	If $P$ is invertible, then
	\begin{eqnarray}\label{eqn:dzetaprimze3} &&f'''\br{0}=\Tr~\br{4Q^3+3Q^2\ell_Q\br{P}+\frac{3}{4}Q\anticommutator{P}{D\ell_Q\br{P}\Br{Q}}}.
	\end{eqnarray}
\end{lemma}

\begin{proof}
Notice that
\begin{align*}
f'(0)&=\Tr~Dq\br{P}\Br{Q}\\
f''(0)&=\Tr~D^2q\br{P}\Br{Q,Q}\\
f'''(0)&=\Tr~D^3q\br{P}\Br{Q,Q,Q}\\
\end{align*}
	Note that $q\br{x}=xp\br{x}$, where $p\br{\cdot}$ is defined in \cref{lem:derivative} item 4.
	\begin{align}
	&\Tr~Dq\br{P}\Br{Q}=\Tr~Qp\br{P}+\Tr~P^2Q
	+\frac{1}{2}P\abs{P}Q+\frac{1}{4}\Tr~P\kappa_Q\br{P}\nonumber\\
	=&\underbrace{\Tr~Qp\br{P}}_{g_{1,Q}\br{P}}+\underbrace{\Tr~P^2Q
	}_{g_{2,Q}\br{P}}+\underbrace{\Tr~P\abs{P}Q}_{g_{3,Q}\br{P}},\label{eqn:dq}
	\end{align}
	where the second equality follows from \cref{lem:trlkappa}.
	
	Further taking derivatives of $g_{1,Q}$, $g_{2,Q}$, and  $g_{3,Q}$ we have	
	\begin{align}
	Dg_{1,Q}\br{P}\br{Q}&=\Tr~\br{\frac{1}{2}Q\anticommutator{P}{Q}+\frac{1}{4}Q\anticommutator{\abs{P}}{Q}+\frac{1}{4}Q\kappa_Q\br{P}}\nonumber\\
	&=\Tr~PQ^2+\frac{1}{2}\Tr~\abs{P}Q^2+\frac{1}{4}\Tr~Q\kappa_Q\br{P}.\label{eqn:dg1b}
	\end{align}
	And
	\begin{equation}
	Dg_{2,Q}\br{P}\Br{Q}=2\Tr~PQ	^2.\label{eqn:dg2b}
	\end{equation}
	By \cref{lem:derivative}
	\begin{equation}\label{eqn:dg3b}
	Dg_{3,Q}\br{P}\Br{Q}=\Tr~\br{\abs{P}Q^2+\frac{1}{2}Q\kappa_Q\br{P}}.
	\end{equation}
	Combining Eqs.~\cref{eqn:dg1b}\cref{eqn:dg2b}\cref{eqn:dg3b} we conclude
	\begin{equation}
	f''\br{t}=\underbrace{4\Tr~PQ^2}_{g_{4,Q}\br{P}}+\underbrace{\frac{3}{2}\Tr~\abs{P}Q^2}_{g_{5,Q}\br{P}}+\underbrace{\frac{3}{4}\Tr~Q\kappa_Q\br{P}}_{g_{6,Q}\br{P}}. \label{eqn:g789}
	\end{equation}
	Taking the derivative of $g_{4,Q}\br{P}$, we have
	\begin{equation}\label{eqn:a}
	Dg_{4,Q}\br{P}\Br{Q}=4\Tr~Q^3.
	\end{equation}
	Applying \cref{lem:derivative} item 2,
	
	\begin{equation}\label{eqn:b}
	Dg_{5,Q}\br{P}\Br{Q}=\frac{3}{2}\Tr~\ell_Q\br{P}Q^2.
	\end{equation}
	From \cref{def:kappa}, 	
	\begin{equation*}
	D\kappa_Q\br{P}\Br{Q}=\anticommutator{Q}{\ell_Q\br{P}}+\anticommutator{P}{D\ell_Q\br{P}\Br{Q}}.
	\end{equation*}
	Thus
	\begin{equation} \label{eqn:dg4b} Dg_{6,Q}\br{P}\Br{Q}=\Tr~\br{\frac{3}{2}Q^2\ell_Q\br{P}+\frac{3}{4}Q\anticommutator{P}{D\ell_Q\br{P}\Br{Q}}}.
	\end{equation}
	Combining Eqs.~\cref{eqn:a}\cref{eqn:b}\cref{eqn:dg4b}, we conclude Eq.~\cref{eqn:dzetaprimze3}.
\end{proof}

\begin{lemma}\label{lem:zetataylorbound}
	Given integers $d,m>0$ and $P,Q\in\H_m$, where $P$ is invertible, let $f(t)=\Tr~q\br{P+tQ}$. It holds that
	\[\abs{f'''\br{0}}\leq c\twonorm{Q}\norm{Q}_4^2,\]
	for some absolute constant $c$.
\end{lemma}
\begin{proof}
	We upper bound each term in Eq.~\eqref{eqn:dzetaprimze3}.
	For the first term, consider
	\begin{equation}\label{eqn:q3}
	\abs{4\Tr~Q^3}\leq4\twonorm{Q}\norm{Q^2}_2=4\twonorm{Q}\norm{Q}_4^2.
	\end{equation}
	For the second term,
	\begin{equation}
	\abs{3\Tr~Q^2\ell_Q\br{P}}\leq3\norm{Q}_4^2\twonorm{\ell_Q\br{P}}\leq3\norm{Q}_4^2\twonorm{Q},\label{eqn:qlq}
	\end{equation}
	where the second inequality follows from \cref{lem:dlq3}.
	For the final term, assuming that $P=\mathsf{Diag}\br{a_1,\ldots, a_m}$ is a diagonal matrix and applying \cref{lem:derivativeh}, we have
	\begin{align}
	&\abs{\frac{3}{4}\Tr Q\anticommutator{P}{D\ell_Q\br{P}\Br{Q}}}\nonumber\\
	=&\frac{3}{4}\abs{\sum_{ijk}Q_{ij}Q_{jk}Q_{ki}\br{\frac{a_i+a_j}{\abs{a_i}+\abs{a_j}}-\frac{\br{a_i+a_j}\br{a_j+a_k}\br{a_k+a_i}}{\br{\abs{a_i}+\abs{a_j}}\br{\abs{a_j}+\abs{a_k}}\br{\abs{a_k}+\abs{a_i}}}}}\nonumber\\
	\leq&\frac{3}{4}\abs{\sum_{ijk}Q_{ij}Q_{jk}Q_{ki}\frac{a_i+a_j}{\abs{a_i}+\abs{a_j}}}\nonumber\\
&+\frac{3}{4}\abs{\sum_{ijk}Q_{ij}Q_{jk}Q_{ki}\frac{\br{a_i+a_j}\br{a_j+a_k}\br{a_k+a_i}}{\br{\abs{a_i}+\abs{a_j}}\br{\abs{a_j}+\abs{a_k}}\br{\abs{a_k}+\abs{a_i}}}}\nonumber\\
	=&\frac{3}{4}\abs{\Tr~\br{\ell_Q\br{P}Q^2}}+\frac{3}{4}\abs{\Tr~\ell_Q\br{P}^3}\nonumber\quad\quad\mbox{(Eq.~\cref{eqn:hb})}\\
	\leq&\frac{3}{4}\twonorm{\ell_Q\br{P}}\norm{Q}_4^2+\frac{3}{4}\twonorm{\ell_Q\br{P}}\norm{\ell_Q\br{P}}_4^2\nonumber\\
	\leq&c\twonorm{Q}\norm{Q}_4^2\quad\quad\mbox{(\cref{lem:dlq3})},\label{eqn:qdpl}
	\end{align}
	for some constant $c>1$. Combining Eqs.~\cref{eqn:q3}\cref{eqn:qlq}\cref{eqn:qdpl}, the result follows.
	
\end{proof}

We are now ready to prove \cref{lem:zetataylor}.
\begin{proof}[Proof of \cref{lem:zetataylor}]

	We assume that $P$ is invertible. The general case follows by continuity. Then $P+tQ$ is invertible except for finite number of $t$'s.
	
	From the definition of $\zeta_{\lambda}$, we have
	\begin{equation}\label{eqn:zetap}
	\zeta_{\lambda}\br{x}=x^2+\frac{\lambda^2}{3}-\frac{q\br{\lambda+x}}{6\lambda}+\frac{q\br{x-\lambda}}{6\lambda}.
	\end{equation}
	Note that $q\br{\cdot}$ is twice differentiable. $q'''\br{\cdot}$ exists except for a finite number of points. Thus from \cref{lem:taylor}, it suffices to upper bound $\Tr~D^3\zeta_{\lambda}\br{P}\br{Q}$, which is directly implied by \cref{lem:zetataylorbound}.
\end{proof}

\begin{proof}[Proof of \cref{lem:zetaadditivity}]
	Note that $\zeta\br{x}=p\br{-x}$. Then from item 4 of \cref{lem:derivative}
		\begin{eqnarray*}
			&&\Tr~D\zeta\br{P}\Br{Q}=\Tr PQ-\frac{1}{2}\Tr~\abs{P}Q-\frac{1}{4}\Tr~\kappa_Q\br{-P}=\Tr \br{P-\abs{P}}Q
		\end{eqnarray*}
		where the second equality follows from \cref{lem:trlkappa}.
		
		Assuming that $P=\mathrm{Diag}\br{a_1,\ldots,a_d}$ is a diagonal matrix, we have
\begin{eqnarray*}
  &&\abs{\Tr~D\zeta\br{P}\Br{Q}}=\abs{\sum_i\br{a_i-\abs{a_i}}Q_{ii}}\leq2\sum_i\abs{a_iQ_{ii}} \\
  &\leq&2\br{\sum_i\abs{a_i}^2}^{1/2}\br{\sum_i\abs{Q_{ii}}^2}^{1/2}\leq2\twonorm{P}\twonorm{Q}.
\end{eqnarray*}	
		Then by the mean value theorem, there exists $\theta\in[0,1]$ such that

		\begin{align*}\abs{\Tr~\br{\zeta\br{P+Q}-\zeta\br{P}}}&=\abs{\Tr~D\zeta\br{P+\theta Q}\Br{Q}}\\&\leq2\twonorm{P+\theta Q}\twonorm{Q}\\&\leq\twonorm{P}\twonorm{Q}+\twonorm{Q}^2.
\end{align*}
\end{proof}

\section{Proof of  \cref{claim:1}}\label{sec:appinvariance}
Before proving \cref{claim:1}, we need the following claim.
\begin{claim}\label{claim:bc}It holds that
	\begin{align}
	&\expec{}{\Tr~\mathbf{B}f\br{\mathbf{A}}}=\expec{}{\Tr~\mathbf{D} f\br{\mathbf{C}}}=0;\label{eqn:bc}\\
	&\expec{}{\Tr~\mathbf{B}f\br{\mathbf{A}}\mathbf{B}g\br{\mathbf{A}}}=m\expec{}{\Tr~\mathbf{D}f\br{\mathbf{C}}\mathbf{D}g\br{\mathbf{C}}}\label{eqn:bc2}
	\end{align}
	for any $f,g\in L^2\br{\reals, \gamma_1}$, where $\mathbf{A},\mathbf{B},\mathbf{C},\mathbf{D}$ are defined in the proof of \cref{lem:hybrid}.
\end{claim}

\begin{proof}
	A crucial observation is that $\mathbf{A},\mathbf{B}$ and $\mathbf{C}$ can be expressed as
	\begin{eqnarray*}
	  \mathbf{A} &=& \id_m\otimes\mathbf{C} \\
	  \mathbf{B} &=& \sum_{\sigma\in[m^2]_{\geq 0}:\sigma\neq 0}\B_{\sigma}\otimes\mathbf{X}_{\sigma} \\
	  \mathbf{D} &=& \sum_{\sigma\in[m^2]_{\geq 0}:\sigma\neq 0}\mathbf{g}_{i+1,\sigma}\mathbf{X}_{\sigma}
	\end{eqnarray*}
for some random operators $\mathbf{X}_{\sigma}$'s, where $\id_m,\mathbf{B}_{\sigma}$'s and $\mathbf{g}_{i+1,\sigma}$'s are in the $\br{i+1}$-th register.
	$\mathbf{B}$ and $\mathbf{D}$ only differ in the $(i+1)$-th register. For the first equality,

	\begin{align*}
&\expec{}{\Tr~\mathbf{B}f\br{\mathbf{A}}}\\
=&\expec{}{\Tr~\mathbf{B}\br{\id_m\otimes f\br{\mathbf{C}}}}\\
=&\sum_{\sigma\in[m^2]_{\geq 0}:\sigma\neq 0}\expec{}{\Tr~\br{\B_{\sigma}\otimes\mathbf{X}_{\sigma}} \br{\id_m\otimes f\br{\mathbf{C}}}}\\
=&\sum_{\sigma\in[m^2]_{\geq 0}:\sigma\neq 0}\expec{}{\Tr~\B_{\sigma}\otimes\br{\mathbf{X}_{\sigma} f\br{\mathbf{C}}}}\\
=&\sum_{\sigma\in[m^2]_{\geq 0}:\sigma\neq 0}\Tr~\B_{\sigma}\expec{}{\Tr\br{\mathbf{X}_{\sigma} f\br{\mathbf{C}}}}\\
=&0.
\end{align*}
	The last equality follows from the orthogonality of $\set{\B_i}_{i\in[m^2]_{\geq 0}}$ and $\B_0=\id_m$, which implies that $\Tr~\B_{\sigma}=\Tr~\B_{\sigma}\B_0=0$.
	And
	\begin{align*}
&\expec{}{\Tr~\mathbf{D}f\br{\mathbf{C}}}\\
=&\sum_{\sigma\in[m^2]_{\geq 0}:\sigma\neq 0}\expec{}{\Tr~\br{\mathbf{g}_{i+1,\sigma}\mathbf{X}_{\sigma}} f\br{\mathbf{C}}}\\
=&\sum_{\sigma\in[m^2]_{\geq 0}:\sigma\neq 0}\expec{}{\mathbf{g}_{i+1,\sigma}}\expec{}{\Tr~\br{\mathbf{X}_{\sigma} f\br{\mathbf{C}}}}\\
=&0.
\end{align*}

	For the second equality,
	\begin{align*}
&\expec{}{\Tr~\mathbf{B}f\br{\mathbf{A}}\mathbf{B}g\br{\mathbf{A}}}\\
=&\expec{}{\Tr~\mathbf{B}\br{\id_m\otimes f\br{\mathbf{C}}}\mathbf{B}\br{\id_m\otimes g\br{\mathbf{C}}}}\\
=&\sum_{\sigma,\tau\neq 0}\expec{}{\Tr~\br{\B_{\sigma}\otimes\mathbf{X}_{\sigma}}\br{\id_m\otimes f\br{\mathbf{C}}}\br{\B_{\tau}\otimes\mathbf{X}_{\tau}}\br{\id_m\otimes g\br{\mathbf{C}}}}\\
=&\sum_{\sigma,\tau\neq 0}\expec{}{\Tr~\br{\B_{\sigma}\B_{\tau}}\otimes\br{\mathbf{X}_{\sigma}f\br{\mathbf{C}}\mathbf{X}_{\tau} g\br{\mathbf{C}}}}\\
=&\sum_{\sigma,\tau\neq 0}\expec{}{\Tr~\br{\B_{\sigma}\B_{\tau}}\Tr\br{\mathbf{X}_{\sigma}f\br{\mathbf{C}}\mathbf{X}_{\tau} g\br{\mathbf{C}}}}\\
=&m\sum_{\sigma\neq 0}\expec{}{\Tr~\br{\mathbf{X}_{\sigma}f\br{\mathbf{C}}\mathbf{X}_{\sigma} g\br{\mathbf{C}}}}.
\end{align*}

	And
	\begin{align*}
&\expec{}{\Tr~\mathbf{D}f\br{\mathbf{C}}\mathbf{D}g\br{\mathbf{C}}}\\
=&\sum_{\sigma,\tau\neq 0}\expec{}{\Tr~\br{\mathbf{g}_{i+1,\sigma}\mathbf{X}_{\sigma}} f\br{\mathbf{C}}\br{\mathbf{g}_{i+1,\tau}\otimes\mathbf{X}_{\tau}} g\br{\mathbf{C}}}\\
=&\sum_{\sigma,\tau\neq 0}\expec{}{\mathbf{g}_{i+1,\sigma}\mathbf{g}_{i+1,\tau}}\expec{}{\Tr\br{\mathbf{X}_{\sigma}f\br{\mathbf{C}}\mathbf{X}_{\tau} g\br{\mathbf{C}}}}\\
=&\sum_{\sigma\neq 0}\expec{}{\Tr~\br{\mathbf{X}_{\sigma}f\br{\mathbf{C}}\mathbf{X}_{\sigma} g\br{\mathbf{C}}}}.
\end{align*}
\end{proof}
\begin{proof}[Proof of \cref{claim:1}]
	
	To prove the first equality,  it suffices to show that
	\begin{equation}\label{eqn:tcase1}
	\expec{}{\Tr~t\br{\mathbf{A},\mathbf{B}}}=m\expec{}{\Tr~t\br{\mathbf{C},\mathbf{D}}},
	\end{equation}
	for
	\begin{align*}
	t\br{A,B}\in\set{p\br{A}B,A^2B,A\cdot\abs{A}\cdot B},
	\end{align*}
due to Eqs.~\cref{eqn:dq}\cref{eqn:zetap}. It directly follows from Eq.~\cref{eqn:bc} in \cref{claim:bc}.
	
	To prove the second equality it suffices to prove that
	
	\[\expec{}{\Tr~t\br{\mathbf{A},\mathbf{B}}}=m\expec{}{\Tr~t\br{\mathbf{C},\mathbf{D}}}\]
	for
	\begin{align*}
	t\br{A,B}\in\set{AB^2,\abs{A}B^2,B\kappa_B\br{A}},
	\end{align*}
due to Eqs.~\cref{eqn:zetap}\cref{eqn:g789}. Then the second equality in \cref{claim:1} follows from the continuity of $D^2\zeta_{\lambda}\br{\cdot}$ due to \cref{lem:dq}.
	
	The first two cases directly follow from Eq.~\cref{eqn:bc2} in \cref{claim:bc}. To prove the final case, we use \cref{lem:kappa} and have
	\begin{eqnarray*}
		&&\Tr~B\kappa_B\br{A}\\
		&=&\Tr~\br{AB+BA}\int_0^{\infty}e^{-t\abs{A}}\br{AB+BA}e^{-t\abs{A}}dt\\
		&=&2\int_0^{\infty}\Tr~\br{Ae^{-t\abs{A}}BAe^{-t\abs{A}}B+A^2e^{-t\abs{A}}Be^{-t\abs{A}}B}~dt
	\end{eqnarray*}
	Then the result follows from Eq.~\cref{eqn:bc2} in \cref{claim:bc}.
\end{proof}

\section{List of notations}\label{sec:tablenotations}
\begin{longtable}{ll}
$\Delta_{\rho}\br{P}$&noise operator, $\rho P+\frac{1-\rho}{m}\br{\Tr P}\cdot\id_m$\\
$\gamma_n$&standard $n$-dimensional normal distribution\\
$\br{\lambda_1\br{M},\ldots,\lambda_m\br{M}}$&eigenvalues of $M$. If $M$ is Hermitian, then they are sorted in non-increasing order\\
$|\sigma|$&the number of nonzeros in $\sigma$\\
$a_S$&the projection of $a$ to the coordinates specified in $S$\\
$a_{-i}$&$a_1,\ldots, a_{i-1},a_{i+1},\ldots, a_n$\\
$a_{<i}$&$a_1,\ldots, a_{i-1}$(similar for $a_{\leq i},a_{>i}, a_{\geq i}$)\\
$A\circ B$&Hadamard product, $\br{A\circ B}_{i,j}= A_{i,j}\cdot B_{i,j}$\\
$A\geq B$&the matrix $A-B$ is positive semidefinite \\
$A\otimes B$(or $AB$)& the composition of systems $A$ and $B$\\
$\anticommutator{A}{B}$&AB+BA\\
$\B_{\sigma}$&$\otimes_{i=1}^n\B_{\sigma_i}$\\
$\deg P$&$\max\set{\abs{\sigma}:\widehat{P}\br{\sigma}\neq 0}$\\
$\deg\br{\mathbf{P}}$&$\max_{\sigma\in[m^2]_{\geq 0}^h}\deg\br{p_{\sigma}}$\\
$\deg\br{f}$&$\max\set{\sum_i\sigma_i:~\widehat{f}\br{\sigma}\neq 0}$\\
&$\max_t\deg\br{f_t}$ for $f=\br{f_1,\ldots,f_k}$\\
$\D\br{A}$&the set of all density operators in $A$\\
$f\br{P}$&$\sum_if\br{\lambda_i}\ketbra{v_i}$, where $P=\sum_i\lambda_i\ketbra{v_i}$\\
$f\in L^p\br{\complex,\gamma_n}$&$f:\reals^n\rightarrow\complex, \int_{\reals^n}\abs{f(x)}^p\gamma_n\br{dx}<\infty$\\
$f\in L^p\br{\complex^k,\gamma_n}$&$f_1,\dots,f_k\in L^p\br{\complex,\gamma_n}$ for $f=\br{f_1,\ldots,f_k}$\\
$f\in L^p\br{\reals,\gamma_n}$&$f:\reals^n\rightarrow\reals, \int_{\reals^n}\abs{f(x)}^p\gamma_n\br{dx}<\infty$\\
$f\in L^p\br{\reals^k,\gamma_n}$&$f_1,\dots,f_k\in L^p\br{\reals,\gamma_n}$ for $f=\br{f_1,\ldots,f_k}$\\
$\innerproduct{f}{g}_{\gamma_n}$&$\expec{\mathbf{x}\sim\gamma_n}{\conjugate{f\br{\mathbf{x}}}g\br{\mathbf{x}}}$\\
&$\sum_{t=1}^k\innerproduct{f_t}{g_t}_{\gamma_n}$for $f=\br{f_1,\ldots,f_k},g=\br{g_1,\ldots,g_k}$\\
$\widehat{f}\br{\sigma}$&$\innerproduct{H_{\sigma}}{f}_{\gamma_n}$\\
&$\br{\widehat{f_1}\br{\sigma},\ldots,\widehat{f_k}\br{\sigma}}$ for $f=\br{f_1,\ldots,f_k}$\\
$\norm{f}_p$&$\br{\int_{\reals^n}\abs{f(x)}^p\gamma_n\br{dx}}^{\frac{1}{p}}$\\
&$\br{\sum_{t=1}^k\norm{f_t}_p^p}^{1/p}$ for $f=\br{f_1,\ldots,f_k}$\\
$\G_{\rho}$&$\rho$-correlated Gaussian distribution $N\br{\begin{pmatrix}
                                                            0 \\
                                                            0
                                                          \end{pmatrix},\begin{pmatrix}
                                                                          1 & \rho \\
                                                                          \rho & 1
                                                                        \end{pmatrix}}$\\
$\H\br{A}$&the set of all Hermitian operators in $A$\\
$\H_m$&the set of all Hermitian operators of dimension $m$ \\
$\H_m^{\otimes n}$&$\underbrace{\H_m\otimes\cdots\otimes\H_m}_{n\text{ times}}$\\
$H^c$&the complement of $H$\\
$H_r\br{x}$&Hermite polynomial, $\frac{(-1)^r}{\sqrt{r!}}e^{x^2/2}\frac{d^r}{dx^r}e^{-x^2/2}$\\
$H_{\sigma}\br{x}$&$\prod_{i=1}^nH_{\sigma_i}\br{x_i}$\\
$\id_A$&the identity operator in $A$\\
$\id_m$&the identity operator of dimension $m$\\
$\influence\br{P}$&$\sum_i\influence_i\br{P}$\\
$\influence\br{f}$&$\sum_i\influence_i\br{f}$\\
$\influence_i\br{P}$&$\innerproduct{\id}{\mathrm{Var}_{\set{i}}[P]}$\\
$\influence_i\br{f}$&$\expec{\mathbf{x}\sim \gamma_n}{\var{f\br{\mathbf{x}}|\mathbf{x}_{-i}}}$\\
&$\sum_t\influence_i\br{f_t}$ for $f=\br{f_1,\ldots,f_k}$\\
$\L\br{A,B}$&the set of quantum channels from $A$ to $B$\\
$\L\br{A}$&$\L\br{A,A}$\\
$\M\br{A}$&the set of all linear operators in $A$\\
$\M_m$&the set of all linear operators of dimension $m$ \\
$\M_m^{\otimes n}$&$\underbrace{\M_m\otimes\cdots\otimes\M_m}_{n\text{ times}}$\\
$M\geq0$&the matrix $M$ is positive semidefinite \\
$M^{\dagger}$&the transposed conjugate of $M$\\
$M_{i,j}$ or $M\br{i,j}$ &the $(i,j)$-entry of $M$\\
$\abs{M}$&$\sqrt{M^{\dagger}M}$\\
$\nnorm{M}$&$\nnorm{M}_{\infty}$, equals to $s_1(M)$\\
$\nnorm{M}_p$&$\br{\frac{1}{m}\sum_{i=1}^ms_i\br{M}^p}^{1/p}$\\
$\norm{M}$&$\norm{M}_{\infty}$, equals to $s_1(M)$\\
$\norm{M}_p$&$\br{\sum_{i=1}^{\min\set{m,n}}s_i\br{M}^p}^{1/p}$\\
$\vec{M}$&an ordered set of operators $\br{M_1,\ldots, M_n}$\\
$\widehat{M}\br{\sigma}$&$\innerproduct{\B_{\sigma}}{M}$, Fourier coefficient of $M$ with respect to $\B$\\
$[n]$&\set{1,\dots,n}\\
$[n]_{\geq0}$&\set{0,\dots,n-1}\\
$N_p\br{\mathbf{P}}$&$\br{\expec{}{\nnorm{\mathbf{P}}_p^p}}^{\frac{1}{p}}$\\
$p=\br{p_{\sigma}}_{\sigma\in[m^2]_{\geq 0}^h}$&the associated vector-valued function of $\mathbf{P}$ under $\set{\B_i}_{i=0}^{m^2-1}$\\
POVM& $M_1,\ldots, M_t\geq0$ satisfying $\sum_{i=1}^tM_i=\id$\\
$\innerproduct{P}{Q}$&$\frac{1}{m}\Tr~P^{\dagger}Q$\\
$P[S]$&$\sum_{\sigma\in[m^2]_{\geq 0}^n:\supp{\sigma}=S}\widehat{P}\br{\sigma}\B_{\sigma}$\\
$P^{\leq t}$&$\sum_{\sigma\in[m^2]_{\geq 0}^n:\abs{\sigma}\leq t}\widehat{P}\br{\sigma}\B_{\sigma}$(similar for $P^{< t}$, $P^{\geq t}$, $P^{>t}$, $P^{= t}$)\\
$P_S$&$\frac{1}{m^{|S^c|}}\Tr_{S^c}P$\\
$\mathbf{P}\in L^p\br{\H_m^{\otimes h},\gamma_n}$ &$p_{\sigma}\in L^p\br{\reals,\gamma_n}$ for all $\sigma\in[m^2]_{\geq 0}^h$\\
$\mathbf{P}\in L^p\br{\M_m^{\otimes h},\gamma_n}$& $p_{\sigma}\in L^p\br{\complex,\gamma_n}$ for all $\sigma\in[m^2]_{\geq 0}^h$\\
$\R\br{x}$&$\arg\min\set{\norm{x-y}_2^2:y\in\Delta}$\\
$\br{s_1\br{M},s_2\br{M},\ldots}$ & singular values of $M$ in non-increasing order\\
sub-POVM& $M_1,\ldots, M_t\geq0$ satisfying $\sum_{i=1}^tM_i\leq\id$\\
$\supp{\sigma}$&$\set{i\in[n]:\sigma_i>0}$\\
$S_k$&the permutation group on $[k]$\\
$\T\br{Q}$&Markov super-operator, $\Tr\br{\br{M^{\dagger}\otimes Q}\psi_{AB}}=\innerproduct{M}{\T\br{Q}}$\\
$\Tr_B\rho_{AB}$&partial trace, $\sum_i\br{\id_A\otimes\bra{i}}\rho_{AB}\br{\id_A\otimes\ket{i}}$\\
$U_{\rho}f\br{z}$&$\expec{\mathbf{x}\sim \gamma_n}{f\br{\rho z+\sqrt{1-\rho^2}\mathbf{x}}}$\\
&$\br{U_{\rho}f_1,\ldots, U_{\rho}f_k}$ for $f=\br{f_1,\ldots,f_k}$\\
$\mathrm{Var}_S[P]$&$\br{P^{\dagger}P}_{S^c}-\br{P_{S^c}}^{\dagger}\br{P_{S^c}}$\\
$\var{M}$&$\innerproduct{M}{M}-\innerproduct{M}{\id}\innerproduct{\id}{M}$\\
$\var{f\br{\mathbf{x}}|\mathbf{x}_S}$&$\expec{\mathbf{x}\sim \gamma_n}{\abs{f\br{\mathbf{x}}-\expec{}{f\br{\mathbf{x}}|\mathbf{x}_S}}^2\text{\Large$\mid$}\mathbf{x}_S}$\\
$\var{f}$&$\expec{\mathbf{x}\sim \gamma_n}{\abs{f\br{\mathbf{x}}-\expec{}{f}}^2}$\\
&$\sum_t\var{f_t}$ for $f=\br{f_1,\ldots,f_k}$\\
$\wt{\sigma}$&$\sum_i\sigma_i$\\
$\X^k$&$\underbrace{\X\times\cdots\times\X}_{k\text{ times}}$\\
$\pos{X}$& $U\pos{\Lambda} U^{\dagger}$ where $X=U\Lambda U$ is a spectral decomposition \\
&of $X$ and $\pos{\Lambda}_{i,i}=\Lambda_{i,i}$ if $\Lambda_{i,i}\geq 0$ and $\pos{\Lambda}_{i,i}=0$ otherwise.\\
$X^{+}$& Moore-Penrose inverse of $X$.
\end{longtable}

\bibliographystyle{plain}
	\bibliography{references}

\end{document}